%% file: x-frt_joe_r2.tex
\documentclass[11pt]{article}
\usepackage[left=1in,top=1in,right=1in,bottom=1in,head=0in,nofoot]{geometry}

\usepackage{setspace,url,bm,amsmath} 
\usepackage{scalerel}

 \usepackage{multirow}
 \usepackage{arydshln}
\usepackage{titlesec} 
\titlelabel{\thetitle.\quad} 
\titleformat*{\section}{\bf\Large\center}

\usepackage{float}
\usepackage{graphicx} 
\usepackage{bbm}
\usepackage{latexsym}
\usepackage{caption}
\usepackage[margin=20pt]{subcaption}
\usepackage{enumerate}

\graphicspath{ {./Graphs/} }
\usepackage{xr}

\newcommand{\proper}{the only test statistics in Table \ref{tb:tt} proper for testing $\HN$ via {\frt}.}
\newcommand{\rtn}{\sqrt N}

\input{cmdlist_frt_mar.tex}

\newcommand{\GG}[1]{}

\usepackage{amsthm}
\usepackage{amssymb}
\usepackage{amsmath}
\usepackage{color}

\usepackage{comment}
\theoremstyle{definition}

\newtheorem*{theorem*}{Theorem}
\newtheorem{theorem}{Theorem}
\newtheorem*{rmk*}{remark}
\newtheorem{proposition}{Proposition}
\newtheorem{lemma}{Lemma}

\newtheorem{condition}{Condition}

\newtheorem{definition}{Definition}
\newtheorem{remark}{Remark}
\newtheorem{corollary}{Corollary}
\newtheorem*{corollary*}{Corollary}

\usepackage{natbib} 
\bibpunct{(}{)}{;}{a}{}{,} 

\usepackage{etoolbox} 
\apptocmd{\sloppy}{\hbadness 10000\relax}{}{} 

\usepackage{multibib}
\newcites{sec}{References}

\usepackage{color}
\usepackage{listings}
\usepackage[hidelinks]{hyperref}
\usepackage{booktabs}

\usepackage{lscape}

\def\ind{\begin{picture}(9,8)
         \put(0,0){\line(1,0){9}}
         \put(3,0){\line(0,1){8}}
         \put(6,0){\line(0,1){8}}
         \end{picture}
        }

\newcommand{\red}{\color{red}}
\newcommand{\blue}{\color{blue}}

\newcommand{\dd}{\diag(\delta_i^2)}

\RequirePackage[normalem]{ulem}

\allowdisplaybreaks
\begin{document}


\onehalfspacing

\title{\bf 
Covariate-adjusted Fisher randomization tests for the average treatment effect
}
\author{Anqi Zhao and Peng Ding
\footnote{Anqi Zhao, Department of Statistics and Applied Probability, National University of Singapore, 117546, Singapore (E-mail: staza@nus.edu.sg). Peng Ding, Department of Statistics, University of California, Berkeley, CA 94720 (E-mail: pengdingpku@berkeley.edu).
Peng Ding was partially funded by the U.S. National Science Foundation (grant \# 1945136). We thank the Associate Editor and three referees for their most insightful comments. We thank Jason Wu, Cheng Gao, Kevin Guo, Thomas Richardson, Avi Feller, Xiaokang Luo, Xinran Li, Zhichao Jiang, Mengsi Gao, Bin Yu, and Philip Stark  for helpful suggestions. 
}
}
\date{}
\maketitle

\setcounter{page}{1}

\begin{abstract}
Fisher's  randomization test (\textsc{frt}) delivers exact $p$-values under the strong null hypothesis of no treatment effect on any units whatsoever and allows for flexible covariate adjustment to improve the power.  
Of interest is whether the resulting covariate-adjusted procedure could also be valid for testing the weak null hypothesis of zero average treatment effect.  
To this end, we evaluate two general strategies for conducting covariate adjustment in {\frt}s: 
the ``\pso" strategy  that uses 
the residuals from an outcome model with only the covariates as the pseudo, covariate-adjusted outcomes to form the test statistic, and the ``\mo" strategy  that directly uses the output from an outcome model with both the treatment and covariates as the covariate-adjusted test statistic. 
Based on theory and simulation, we recommend using the ordinary least squares (\textsc{ols}) fit  of the observed outcome on the treatment, centered covariates, and their interactions for covariate adjustment, and conducting \textsc{frt} with the robust $t$-value of the treatment as the test statistic. The resulting {\frt} is finite-sample exact for testing the strong null hypothesis, 
asymptotically valid for testing the weak null hypothesis, and more powerful than the unadjusted counterpart under alternatives, all irrespective of whether the linear model is correctly specified or not.
We start with complete randomization, and then extend the theory to cluster randomization, stratified randomization, and rerandomization,  respectively, giving a recommendation for the test procedure and test statistic under each design.
Our theory is design-based,  also known as randomization-based, in which we condition on the potential outcomes but average over the random treatment assignment.
\end{abstract}

\newpage 
\section{Fisher's randomization test with covariate adjustment}\label{sec:intro}

\citet{Fisher35} viewed randomization as a ``reasoned basis'' for inference and proposed the randomization test as a universal way to generate finite-sample exact $p$-values without imposing modeling assumptions on the experimental outcomes.  Fisher's  randomization test (\textsc{frt}) becomes increasingly important with the popularity of field experiments in social sciences in addition to the traditional biomedical experiments. 
\citet{proschan2019re} reviewed the use of \textsc{frt}s  in randomized clinical trials  and highlighted its strength in analyzing complex data.  
There is an increasing interest in economics and related fields to use {\frt}s to analyze various types of empirical data \citep{fl, knd, lee2014, cattaneo2015randomization, canay2017randomization, ganong2018permutation, athey2018exact, bugni2018inference, young2019channeling, hk, mackinnon2020randomization, heckman2020}.

The flexibility of {\frt} enables two natural strategies to incorporate covariate information via statistical modeling. First, we can fit a statistical model of the outcome on the covariates to obtain the residuals as the covariate-adjusted pseudo outcomes, and proceed with the usual {\frt} in a covariate-free fashion. 
\citet{tukey1993tightening} used it with linear models, \citet{gail1988tests} used it with generalized linear models, 
\citet{raz1990testing} used it with nonparametric regressions, 
\citet{stephens2013flexible} used it with the generalized estimating equation for clustered data, and  \citet{CovAdjRosen02} reviewed and extended it to not only randomized experiments but also matched observational studies. 
Second, we can directly fit an outcome model with both the treatment and covariates, and use the model output, such as the coefficient of the treatment or the corresponding $t$-value, as the test statistic. The canonical choice is a linear model on the treatment and covariates, often known as the analysis of covariance \citep{Fisher35, Freedman08a, Lin13, young2019channeling}.  
\citet{brillinger1978management} gave an early application of this strategy with more complex statistical models. 
This defines two model-assisted approaches to conducting covariate-adjusted {\frt}s.

The strong guarantees of \textsc{frt} hold only under the strong null hypothesis of zero individual treatment effects, 
which is often criticized for being too restrictive for many practical applications.  
Adaptation to the weak null hypothesis of zero average treatment effect is one important direction for broadening its application.
A natural class of test statistics for this purpose are the coefficients of the treatment from various outcome models, with or without covariate adjustment, along with their classic or robustly-studentized $t$-statistics \citep{eicker1967limit, huber::1967, White80}. 
These coefficients are consistent estimators of the average treatment effect and are thus  sensitive to deviations from both the strong and weak null hypotheses \citep{Freedman08a, Lin13}. 
Of interest is whether these intuitive test statistics can preserve the correct type one error rates when only the weak null hypothesis holds \citep{romano1990behavior, Romano13, wuanding2020jasa}, and if covariate adjustment delivers additional power under alternatives.

To this end, we extend the discussion by \cite{DD18} and \cite{wuanding2020jasa} on the utility of {\frt} for testing the weak null hypothesis to the presence of covariates under the finite-population framework, and examine the asymptotic operating characteristics of nine covariate-adjusted test statistics as the coefficients from three common outcome models and their respective classic and robustly-studentized $t$-statistics. 
The results establish the permutational limiting theorems of  the {\ols} coefficients and standard errors
under the possibly misspecified linear models, and shed light on the utility of model-assisted covariate adjustment for testing the weak null hypothesis via \frt.
Building upon previous work on using studentized statistics for permutation tests under the super-population \citep{Janssen97, Romano13, Pauly15, bugni2018inference} and finite-population \citep{wuanding2020jasa} frameworks, respectively, 
we extend the discussion to the covariate-adjusted test statistics under the finite-population framework, and show the necessity of robust studentization for ensuring the asymptotic validity of the covariate-adjusted test statistics when only the weak null hypothesis holds. 
The robustly studentized $t$-statistic based on \citet{Lin13}'s estimator, as it turns out, guarantees both asymptotic validity and the highest power for testing the weak null hypothesis. The  estimator equals the coefficient of the treatment in the \textsc{ols} fit of the outcome on the treatment, centered covariates, and their interactions, but the aforementioned theoretical guarantees hold irrespective of whether the linear model is correctly specified or not.  Together with its finite-sample exactness under the strong null hypothesis, it is thus our final recommendation for testing both the strong and weak null hypotheses under complete randomization.

We first focus on complete randomization and then generalize the theory to other types of design. 
The extension to cluster randomization and stratified randomization is direct whereas that to rerandomization \citep{morgan2012rerandomization} has some distinct features.
In particular, covariate adjustment becomes more crucial since studentization alone no longer ensures the appropriateness of \textsc{frt} for the weak null hypothesis. 
In addition, it is common that the designer and analyzer do not communicate \citep{bm, heckman2020, hk}, and if this happens, we recommend using \textsc{frt} pretending that the experiment was completely randomized. 
In this non-ideal case, the proposed {\frt} is no longer finite-sample exact under the strong null hypothesis unless the original experiment is indeed completely randomized, 
but at least it preserves the correct type one error rates under the weak null hypothesis. 
Based on extensive theoretical investigations, we make final recommendations for \textsc{\frt} and the test statistic  in each experimental design. 

We will use the following notation for permutations.
Let $\Pi$ be the set of all $N!$ random permutations of $\{1,\dots, N\}$, indexed by $\pi$. 
For an $N\times 1$ vector $a = (a_1, \dots, a_N)^\T$, let $a_\pi = (a_{\pi(1)}, \dots, a_{\pi(N)})^\T$ be a permutation of its elements. 
If $b = b(a)$ is a function of $a$, let $b^\pi = b(a_\pi)$ be its value evaluated at $a_\pi$.
Without introducing new notation, use $\pi$ to also represent a random draw from $\Pi$,  namely  $\pi \sim \Unif(\Pi)$, with meaning clear from the context.  
With a slight abuse of notation, assume sets like $\{a_\pi: \pi\in\Pi\}$
to contain $|\Pi| = N!$ elements defined by $\pi \in\Pi$ throughout, such that $a_\pi$ and $a_{\pi'}$ are two distinct elements so long as $\pi \neq \pi'$, even if $a_\pi = a_{\pi'}$.

\section{Basic setup under complete randomization}
\label{sec::basicsetup}

\subsection{Potential outcomes and Fisher's randomization test}
Consider an intervention of two levels, $z = 0, 1$, and a finite population of $N$ units, $i =1,\ldots, N$.  Let $Y_i(z)$ be the potential outcome of unit $i$ under treatment $z$ \citep{Neyman23}. 
The individual treatment effect is $\tau_i  = Y_i(1) - Y_i(0)$, and the finite-population average treatment effect is $\tau = N^{-1} \sumN  \tau_i$. 
We focus on the finite-population inference in the main text, and will refer to $\tau$ as the ``average treatment effect" when no confusion would arise.
Let $x_i = (x_{i1}, \ldots, x_{iJ})^\T$ be the covariates for unit $i$, concatenated as an $N\times J$ matrix $X = (x_1, \dots, x_N)^\T$.
Center the covariates at $\bar x = N^{-1}\sumi x_i = 0_J$ to simplify the presentation. 

The designer  assigns $N_z$ units to receive level $z$ with $N_1 + N_0 = N$ and $(p_1, p_0) =(N_1/N , N_0/N).$ 
Let $Z_i$ denote the treatment level received by unit $i$, with $Z_i=1$ for treatment and $Z_i=0$ for control, vectorized as $Z = (Z_1, \dots, Z_N)^\T$.
Complete randomization samples $Z$ uniformly from the set $\mZ$
that contains all permutations of $N_1$ 1's and $N_0$ 0's.
The observed outcome is $ Y_i  = \Zi\Yi(1) + (1-\Zi)\Yi(0) $ for unit $i$, vectorized as   $  Y  = (  Y _1, \ldots,   Y _N)^\T$. 
A test statistic is a function of the treatment vector, observed outcomes, and  covariates, denoted by $T = T(Z, Y, X)$.  

Write $Y = Y(Z)$ and $T = T(Z, Y(Z), X)$ to highlight the dependence of the observed outcomes and the test statistic on the treatment vector. 
Complete randomization induces a uniform distribution over $  \{ T(\bm z, Y(\bm z), X) : \bm z \in \mz\}$ as the \emph{sampling distribution} of $T$.
\citet{Fisher35} considered testing the strong null hypothesis  
$$
\HF: Y_i(1) = Y_i(0) \quad \text{for all $ i=1,\ldots, N$}
$$
and proposed  \textsc{frt} to compute the $p$-value as 
\begin{equation}\label{pfrt}
p_\frt 
 = |\Pi|^{-1}\sumpi1\{ T(Z_\pi, Y, X ) \geq T(Z,Y,X)\},
\end{equation}
assuming a one-sided test. Each $Z_\pi $ is a permutation of $Z$, and by symmetry, all possible values of  $Z_\pi $ over $\pi \in \Pi$ consist of $\mz$.  \textsc{frt} thus induces a uniform distribution over $  \{ T(\bm z, Y(Z), X) : \bm z \in \mz\}$ conditioning on the observed $Z$,  known as the {\it randomization distribution} of $T$. 
Let  $T^\pi  = T(Z_\pi, Y(Z),X)$, where $\pi\sim \unif(\Pi)$, be a random variable from this distribution {conditioning on $Z$}. 
The $p_\frt$ in \eqref{pfrt} gives the right-tail probability of the observed value of the test statistic with regard to its randomization distribution.
Under $\HF$, the randomization distribution equals the sampling distribution, and thereby ensures the finite-sample exactness of \textsc{frt} for arbitrary $T$. 

In practice, we need to choose a test statistic that is sensitive to deviations from $\HF$.
Computationally, {\frt} involves randomly permuting the treatment vector to generate $Z_\pi$. This justifies ``permutation test'' as its other name. 
If $|\Pi|=N!$ is too large, we can take a simple random sample from $\Pi$ to obtain a Monte Carlo approximation of  $p_\frt $. 

\begin{remark}\label{rmk:pt}
Based on the definition in \eqref{pfrt}, {\frt} does not require any algebraic group structure of the treatment assignment mechanism. Therefore, it is more general than the usual definition of permutation test which requires certain invariance under group transformations \citep{hoeffding1952large, Lehmann2005}. Nevertheless, we will focus on {\frt}s that can be implemented by permutations or restricted permutations in this paper. See \cite{Basse2019} for more discussion on the connections and distinctions. 
\end{remark}

The nice properties of \textsc{frt} under $\HF$ inspire endeavors to extend it to other types of hypotheses. 
Consider the weak null hypothesis of zero average treatment effect \citep{neyman1935statistical}: 
$$
\HN: \tau = 0.
$$
We can proceed with \textsc{frt} by permuting the treatment vector $Z$  and report $p_\frt$ by \eqref{pfrt} as if we were testing $\HF$. 
Under $\HN$, $Y(\bm{z})$ varies with $\bm{z} \in  \mz$ such that the randomization distribution $T^\pi$ no longer equals the sampling distribution $T$. Consequently, $p_\frt$ loses its finite-sample exactness as the basis for controlling the type one error rates in general. 
\citet{wuanding2020jasa} gave examples in which {\frt} yields invalid type one error rates under $\HN$ even asymptotically.

The discrepancy between the distributions of $T$ and $T^\pi$ in the absence of $\HF$ is at the heart of the loss of finite-sample exactness when applying {\frt} to hypotheses other than $\HF$. 
The sampling distribution of $T$, on the one hand, is induced by the randomization of $Z$ based on the ``true finite population" of $\{(Y_i(0), Y_i(1), x_i\}_{i=1}^N$. 
The {\frt} procedure, on the other hand, assumes the observed outcomes $Y= (Y_1, \dots, Y_N)^\T$ remain unchanged over all possible assignments, and is essentially using $\{(Y_i'(0), Y_i'(1), x_i\}_{i=1}^N$, where $Y_i'(0) = Y_i'(1) = Y_i$, as the ``pseudo finite population" to generate the randomization distribution of $T$, represented by $T^\pi$. 
This ends up mixing   $\{Y_i(0)\}_{i=1}^N$ and $\{Y_i(1)\}_{i=1}^N$ with proportions $p_0$ and $p_1$  in the absence of $\HF$, and thereby results in the different distributions of $T$ and $T^\pi$.

Despite the possibly liberal type one error rates  in general, sensible choice of the test statistic restores the validity of \textsc{frt} for testing $\HN$ at least asymptotically.  \citet{wuanding2020jasa} showed that {\frt}  preserves the correct type one error rates asymptotically  with a class of robustly studentized statistics.  
We extend their discussion to the setting with covariates,  and propose a general strategy for covariate-adjusted {\frt} that ensures both asymptotic validity and higher power for testing $\HN$.

Assume the finite-population asymptotic framework that embeds $\mathcal{S} = \{ Y_i(0), Y_i(1), x_i \}_{i=1}^N$ and  $Z = (Z_i)_{i=1}^N$ into a sequence of finite populations and assignments for $N = \ot{\infty}$.
Technically, all quantities depend on $N$, but we omit the subscript $N$ for simplicity. 
 
\begin{definition}\label{prop}
A test statistic $T$ is \emph{proper} for testing $\HN$ if under $\HN$,  
$$
\limN  \pr(  p_\frt \leq \alpha   ) \leq \alpha \qquad \text{for all }\alpha\in(0,1)
$$ 
holds for  any $\mathcal{S}$.
\end{definition}

Assume a one-sided test and $p$-value as the right-tail probability as in \eqref{pfrt}. 
A statistic $T$ is proper for testing $\HN$  if under $\HN$,  the sampling distribution of $T$ is stochastically dominated by its randomization distribution  for almost all sequences of $Z$ as $N \to \infty$.

%
\subsection{Two strategies for covariate-adjusted FRT and twelve test statistics}\label{sec:tt}
%
We review two general strategies for covariate adjustment in \frt. We focus on test statistics based on estimators of $\tau$ to accommodate both $\HF$ and $\HN$, 
and unify them under the {\ols} formulation for easy implementation.

Let $\hY(z) = N_z^{-1}\sum_{i:Z_i = z}  Y_i $ be the sample average of the outcomes under treatment $z$. 
The difference-in-means estimator $\hat{\tau}_\neyman = \hY(1) - \hY(0)$ is unbiased for $\tau$ under complete randomization \citep{Neyman23}, and affords a natural statistic for testing both $\HF$ and $\HN$. 
Algebraically, it equals the coefficient of $Z_i$ from the \textsc{ols} fit of $Y_i$ on $(1,Z_i)$. 
It is also common to use $\hat{\tau}_\neyman / \sen $ or $\hat{\tau}_\neyman / \tse_\neyman$ as the test statistic,  where $\sen$ and $\tse_\neyman$ are the classic and robust standard errors from the same \textsc{ols} fit: 
\begina
\ese_\nm^2 = \frac{N(N_1-1)}{(N-2)N_1N_0} \hat  S_{1}^2 +  \frac{N(N_0-1)}{(N-2)N_1N_0}  \hat  S_{0}^2 
\approx \frac{  \hat  S_{1}^2  }{ N_0 }    +  \frac{  \hat  S_{0}^2  }{ N_1} , 
\quad 
\tse_\nm^2 =   \frac{N_1-1}{ N_1^2}  \hat  S_{1}^2  +  \frac{N_0-1}{ N_0^2} \hat  S_{0}^2
\approx \frac{  \hat  S_{1}^2  }{ N_1 }    +  \frac{  \hat  S_{0}^2  }{ N_0} 
\enda
with $\hat{S}_z^2 = (N_z-1)^{-1}\sum_{i:Z_i=z} \{Y_i - \hY(z) \}^2 \ (z=0,1)$ 
 \citep{AngristEcon}.
 This yields three unadjusted test statistics as the baseline for discussing the possible improvement via covariate adjustment. 
 The randomization distributions can then be generated by replacing $Z_i$ with $Z_{\pi(i)}$ in the above {\olsf} over all $\pi\in\Pi$. 
In particular, 
we can compute $\htnp$ as the coefficient of $Z_{\pi(i)}$ from the {\olsf} of $Y_i$ on $(1, Z_{\pi(i)})$ with $(\hat\tau_\neyman/\ese_\neyman)^\pi = \htnp /\senp$  and $(\hat\tau_\neyman/\tse_\neyman)^\pi = \htnp /\tsenp$, where $\tsenp$ and $\senp$ are the corresponding classic and robust standard errors. The same intuition extends to the covariate-adjusted variants below as we shall introduce in a minute. 
\citet{Romano13} and \citet{wuanding2020jasa} showed that randomization tests with the robustly studentized $\htn/\tse_\nm$ 
are asymptotically valid for $\HN$   under the super- and finite-population frameworks, respectively.
We thus also consider studentization in covariate adjustment.

The first strategy for covariate adjustment is to fit an outcome model with covariates alone  and use the residuals as the fixed, covariate-adjusted
pseudo outcomes for conducting {\frt}.
This appears to be the dominating approach advocated by \citet{CovAdjRosen02}; see also
  \citet{gail1988tests}, \citet{raz1990testing}, \citet{tukey1993tightening} and  \citet{Otto}. 
Let $ e  = ( e _1, \dots,  e _N)^\T$ be the residuals from the \textsc{ols} fit of $Y_i$ on $(1,x_i)$, which can be viewed as pseudo outcomes unaffected by the treatment under $\HF$. 
The difference in means of  the residuals, $\htR =   \hat{e}    (1) - \hat{e}   (0)$, equals the coefficient of $Z_i$ from the \textsc{ols} fit of $e_i$ on $(1,Z_i)$ and affords an intuitive estimator of $\tau$ after adjusting for the covariates. 
Similar to the discussion of $\hat{\tau}_\neyman$, 
we can use $\htR$, $\htR/\ese_\rosenbaum$, and $\htR/\tse_\rosenbaum$ as the test statistics for testing $\HF$ or $\HN$ by {\frt}, 
where $\ese_\rosenbaum $ and $\tse_\rosenbaum$ are the classic and robust standard errors from the \textsc{ols} fit that yields $\hat{\tau}_\rosenbaum$.
We regress $Y_i$ on $(1,x_i)$ to form the residuals $e$ whereas \citet{CovAdjRosen02} regressed $Y_i$ on $x_i$ alone without the intercept. 
The difference does not affect $\htR$ with centered covariates.

The second strategy for covariate adjustment is to directly fit an outcome model with both the treatment and covariates, and use the coefficient or $t$-values of the treatment as the test statistic for $\frt$.
\citet{Fisher35} suggested an estimator $\htau_\fisher$ for $\tau$, which equals  the coefficient of $Z_i$ from the \textsc{ols} fit of $Y_i$ on $(1,Z_i, x_i)$.  \citet{Lin13} recommended an improved estimator, $\htau_\lin$, as the coefficient of $Z_i$ from the \textsc{ols} fit of $Y_i$ on $\{ 1,Z_i,  x_i - \bar x, Z_i (x_i-\bar x)\}$ with centered covariates and treatment-covariates interactions. 
These two covariate-adjusted estimators, along with their respective studentized variants, afford six additional test statistics, namely $\hts$, $\htss$, and $\htssr\ (*=\fisher,\lin)$, for testing $\HF$ or $\HN$ by {\frt}, where $\ese_*$ and $\tse_*$ are the classic and robust standard errors from the respective {\olsf}s.

\begin{table}[t]\caption{\label{tb:tt} Twelve test statistics  where $\ese$ and $\tse$ denote the classic and robust standard errors.}
\begin{center}
\begin{tabular}{|c|c|c|c|c|}\hline
&\citet{Neyman23}& \citet{CovAdjRosen02} & \citet{Fisher35} & \citet{Lin13} \\\hline
unstudentized & $\htN$&$\htR$ &$\htF$ &$\htl$  \\
studentized by $\ese$  & $\htns $ &$\htR/\ese_\rosenbaum$ &$\htF/\ese_\fisher$&$\htl/\ese_\lin$\\
studentized by $\tse$ & $\htnsr$ &$\htR/\tse_\rosenbaum$ &$\htF/\tse_\fisher$&$\htl/\tse_\lin$\\\hline
\end{tabular}
\end{center}
\end{table}

This gives us a total of twelve test statistics, three unadjusted and nine adjusted, for testing the treatment effects via {\frt}. Table \ref{tb:tt} summarizes them, with the subscripts $\neyman$, $\rosenbaum$, $\fisher$, and $\lin$ indicating \citet{Neyman23}, \citet{CovAdjRosen02}, \citet{Fisher35}, and \citet{Lin13}, respectively. 
All twelve statistics are finite-sample exact for testing $\HF$ irrespective of whether the models are correctly specified or not. 
Our goal is to evaluate their abilities to preserve the correct type one error rates under $\HN$. Without loss of generality, we assume two-sided {\frt} for the rest of the text, or, equivalently, we use the absolute values of the test statistics in Table \ref{tb:tt} to compute the $p_\frt$ in \eqref{pfrt}.

The two strategies for covariate adjustment unify nicely under the {\ols} formulation yet differ materially with regard to the role of covariates under the permutations induced by the {\frt} procedure. The first strategy, on the one hand, adjusts for the covariates only once to form the pseudo outcomes $e$  and proceeds with the permutations in a covariate-free fashion. The second strategy, on the other hand, adjusts for the covariates in each of the $ N!$ permutations of $Z$.

 Before giving the formal results on the finite-population asymptotics, we unify below the three covariate-adjusted estimators 
as the differences in means of distinct adjusted outcomes. 
Let $\sx = (N-1)^{-1}\sumi x_ix_i^\T$ and $\qxy = (N-1)^{-1}\sumi x_iY_i$ be the finite-population  covariance matrices of the centered $(x_i)_{i=1}^N$ with itself and $(Y_i)_{i=1}^N$, respectively.
Let $\htx= \hat x(1) - \hat x(0)$ be the difference in means of the covariates under treatment and control, where $\hat x(z) = N_z^{-1} \sum_{i: Z_i = z}x_i$.
Let $\hat{S}_{x(z)}^2= (N_z-1)^{-1} \sum_{i: Z_i = z} \{x_i - \hat x (z)\}\{x_i - \hat x (z)\}^\T$ 
and $\hat{S}_{xY(z)} = (N_z-1)^{-1} \sum_{i: Z_i = z} \{x_i - \hat x(z)\}\{Y_i - \hat Y(z)\} $
be the sample covariance matrices of $x_i$ with itself and $Y_i$ under treatment $z$.   
Let $\hgr$ and $\hgf$ be the coefficients of $x_i$ from the \textsc{ols} fits of $Y_i$ on  $(1,x_i)$ and $(1,Z_i,x_i)$, respectively. 
Let $\hg_\lin =p_0\hglo+ p_1\hglz$, where $\hglzz$ is the coefficient of $x_i$ from the \textsc{ols} fit of $Y_i$ on $(1,x_i)$ over the units under treatment $z$. 

\begin{proposition}\label{FLRtoN} 
We have 
\begina
\hts &=& N^{-1}_1\sum_{i:Z_i = 1}(Y_i -  x_i^\T\hg_*) - N^{-1}_0\sum_{i:Z_i = 0}(Y_i -  x_i^\T\hg_*) = \htN-  \htx^\T \hg_*, \quad (* = \rosenbaum, \fisher)\\
\htl  &=& N^{-1}_1\sum_{i:Z_i = 1}(Y_i - x_i^\T\hglo ) - N^{-1}_0\sum_{i:Z_i = 0}(Y_i - x_i^\T\hglz) = \htN- \htx^\T\hg_\lin ,
\enda
where $
\hgr =(\sxx)^{-1}\hat S_{xY}$,  $ \hgf = \hgr -  (1-1/N)^{-1}p_1p_0\htf ( S_x^2)^{-1} \htx$, and $\hglzz = (\hat S_{x(z)}^2)^{-1} \hat{S}_{xY(z)}$. 
\end{proposition}

Proposition \ref{FLRtoN} entails only the algebraic properties of the {\ols} fits and holds under arbitrary data generating process. It unifies
$\hts \  (* = \rosenbaum, \fisher, \lin)$ as the difference-in-means estimators defined on the adjusted outcomes, or, equivalently, as $\hat{\tau}_\neyman$ with corrections based on the imbalance in means of the covariates.

Under $\HN$, {\frt} with $\hat{\tau}_\neyman$ does not preserve the correct type one error rates but {\frt} with $\htnsr$ does \citep{DD18}. 
In the next section, we will extend the result to the nine covariate-adjusted test statistics in Table \ref{tb:tt} and establish  the properness of the four robustly studentized $t$-statistics, namely $\hts / \tse_*\ (\phs)$, for testing $\HN$. 
Refer to them as the \emph{robust $t$-statistics} hence. We will further show that among them, $\htl/\tse_\lin$ delivers the highest power under alternative hypotheses.

\begin{remark}\label{rmk:taup}
Inspired by the distinction between $\htf$ and $\htl$ under the second  covariate adjustment strategy, an alternative way to implement the first, pseudo-outcome-based strategy is to fit two separate {\ols} regressions of $Y_i$ on $x_i$ for the treated and control units, respectively, both without the intercept, and then use the resulting residuals for conducting {\frt}. 
Despite the computational advantage of this approach in that it adjusts for the covariates only once, the resulting tests lead to distinct sampling and randomization distributions even under $\HF$, and are thus not finite-sample exact for testing $\HF$. 
We see the finite-sample exactness under $\HF$ the first criterion for a test to qualify as {\frt}, and thus do not pursue this direction. In particular, properness under $\HN$ can be a rather weak requirement without the finite-sample exactness under $\HF$. See Section \ref{sec::pt} in the supplementary material for more examples of permutation tests of this type that we do not recommend in general. 
\end{remark}

\section{Asymptotic theory for FRTs for testing $  \tau = 0$}\label{main}

\subsection{Limiting distributions under complete randomization}\label{sec:thm1-4}

We will develop in Theorems \ref{neyman}--\ref{lin} the limiting distributions of the twelve test statistics in Table \ref{tb:tt}  under the finite-population framework conditioning on $\mathcal{S} = \{Y_i(0), Y_i(1), x_i\}_{i=1}^N$.
The theorems assume neither $\HF$ nor $\HN$ but hold for arbitrary $\mathcal{S}$ that satisfies the regularity conditions specified in Condition \ref{asym} below.  
Applying them to finite populations that actually satisfy $\HN$ elucidates the asymptotic validity and power of {\frt} for testing $\HN$ in Sections \ref{sec::typeoneerror} and \ref{sec:power}. 
Let $\bar{Y}(z) = N^{-1} \sumN Y_i(z)$ and  $S_{z}^2  = (N-1)^{-1} \sumN  \{ Y_i(z) - \bar{Y}(z)  \}^2$ be the finite-population mean and variance of $\{ Y_i(z) \}_{i=1}^N$, respectively. Let $ S_\tau^2  = (N-1)^{-1} \sumN (\tau_i -\tau)^2$  be the finite-population variance of $(\tau_i)_{i=1}^N$. 
Let $S_{xY(z)} = (N-1)^{-1}\sumN x_i Y_i(z)$ be the finite-population covariance matrix of $\{x_i, Y_i(z)\}_{i=1}^N$.
Let $w_i(z) =  (S_x^2)^{-1}x_i Y_i(z)$ with  $\bar w(z)=N^{-1}\sumi w_i(z) = (1-1/N) (S_x^2)^{-1} S_{xY(z)}$.

\begin{condition}\label{asym} As $N \to \infty$, for $z =  0,1$, 
(i) $p_z  $ has a limit in $  (0,1)$, 
(ii) the first two moments of $\{Y_i(0), Y_i(1), x_i\}_{i=1}^N $  have finite limits;
$S_x^2$ and its limit are both positive definite; 
$ S_{z}^2 - S^\T_{xY(z)} (\sxx)^{-1}S_{xY(z)}$ has a finite positive limit;  
the second moments of $\{  w_i(0),  w_i(1)\}_{i=1}^N$ have finite limits, and 
(iii)  there exists a $ c_0 < \infty$ independent of $N$ such that $N^{-1}\sum_{i=1}^N Y^4_i(z) \leq c_0$, $N^{-1}\sum_{i=1}^N \|x_i\|^4_4 \leq c_0$, and  $N^{-1}\sum_{i=1}^N \|w_i(z)\|_4^4 \leq c_0$.
\end{condition}

Condition \ref{asym}(ii) ensures $S_\tau^2$ has a finite limit. We also use $p_z$, $\bar{Y}(z)$, $\tau$, $S_{z}^2 $, $\sxx$, $S_{xY(z)}$, and $ S_\tau^2 $ to denote their limiting values without introducing new symbols. The exact meaning should be clear from the context. 

Denote by $\asz$ a statement that holds for almost all sequences of $Z$.  
We review in Theorem \ref{neyman} the asymptotic distributions of the three unadjusted test statistics from \citet{DD18}, and extend them to the covariate-adjusted cases in Theorems \ref{rosenbaum}--\ref{lin}.
 
\begin{theorem}\label{neyman}
{\asym}  
\begin{enumerate}
[(a)]
\item
$\rtn  (\htN - \tau) \rightsquigarrow  \mN(0, v_\neyman)$,
and
$ \rtn   \htnp     \rightsquigarrow  \mN(0, v_{\neyman0}) \ \asz$,
where
$v_\neyman  =  p_1^{-1} S_{1}^2 + p_0^{-1} S_{0}^2 - S_{\tau}^2$
and
$v_{\neyman0}  = p_0^{-1}S_{1}^2 + p_1^{-1} S_0^2 +\tau^2 $.
\item
$(\htN-\tau) /\sen   \rightsquigarrow  \mN(0, c'_\neyman)$,
and 
$(\htN /\sen  )^\pi   \rightsquigarrow  \mN(0, 1) \ \asz$, where $c'_\neyman = v_\neyman / (v_{\neyman0} - \tau^2)$.
\item
$(\htN -\tau) /\tse_\neyman   \rightsquigarrow  \mN(0, c_\neyman)$,
and 
$(\htN /\tse_\neyman)^\pi   \rightsquigarrow  \mN(0, 1) \ \asz$, 
where $c_\neyman = v_\neyman / (v_\neyman +     S_\tau^2     ) \leq 1$.
\end{enumerate}
\end{theorem}

 Theorem \ref{neyman} gives the asymptotic sampling and randomization distributions of $\htn$, $\htn/\ese_\nm$, and $\htn/\tse_\nm$.
  Building up intuitions  for  Theorem \ref{neyman} 
helps to understand Theorems \ref{rosenbaum}--\ref{lin} for the covariate-adjusted cases below.

First, it clarifies $\sqrt N \htn$ and $\sqrt N \htnp$ as both asymptotically normal with asymptotic variances $v_\neyman$ and $v_{\nm0}$ that are in general not equal.
Recall from the definition of randomization distribution that the distribution of $\htnp$ is always conditional on the assignment vector $Z$. 
The fact that $ \rtn   \htnp     \rightsquigarrow  \mN(0, v_{\neyman0}) \ \asz$ suggests this randomization distribution is asymptotically identical for almost all sequences of $Z$.

Second, the component $S^2_\tau$ in $v_\nm$ is unique to the finite-population inference and cannot be estimated consistently from the observed outcomes. 
The resulting inferences are thus necessarily conservative unless $\tau_i = \tau$ for all $i = \ot{N}$ \citep{Neyman23}. 
The classic and robust standard errors afford two convenient estimators of $\cov(\htn)$ with 
\begin{eqnarray}
\label{eq::standarderrors}
N\ese^2_\nm = v_{\nm0} - \tau^2 + \op,\quad N\tse^2_\nm = v_\nm + S_\tau^2 + \op 
\end{eqnarray}
by Lemma \ref{v_lim} in the supplementary material.
This implies $\tse^2_\nm$ is asymptotically conservative for $\cov(\htn)$
whereas $\ese^2_\nm$ in general is not, entailing the asymptotic sampling distributions of the classic and robust $t$-statistics with variances $c_\nm'$ and $c_\nm$, respectively. 
In contrast, the super-population framework has no analog of $S_\tau^2$ such that $\tse_\neyman^2$ is consistent for the true sampling variance \citep{Romano13}.

Further, Theorem \ref{neyman} does not assume $\HN$ or $\HF$  but holds for arbitrary $\mathcal{S}$ that satisfies Condition \ref{asym}.  
 The null hypotheses $\HN$ and $\HF$ impose restrictions on $\mathcal{S}$, and thereby enable more informative comparisons between the asymptotic sampling and randomization distributions. 
The weak null hypothesis $\HN$, on the one hand, ensures $\tau = 0$ such that all six normal distributions center at zero with $v_{\nm0} = p_0^{-1}S_{1}^2 + p_1^{-1} S_0^2$ and $c_\nm' = v_\nm/v_{\nm0}$. 
Compare the expressions of $v_\nm$ and $v_{\nm0}$ to see that the asymptotic distributions of $\htn$ and $\htnp$ still differ unless $(p_1-p_0)(S_{1}^2 - S_0^2)  + S_{\tau}^2 = 0$. 
This demonstrates that $p_\frt$  can be invalid  even asymptotically. 
Compare the distributions of $\htn/\tse_\nm$ and $(\htn/\tse_\nm)^\pi$ to see $|\htn/\tse_\nm|\leqst|(\htn/\tse_\nm)^\pi|$ asymptotically $\asz$. 
This underpins the asymptotic validity of {\frt} with $\htn/\tse_\nm$ for testing $\HN$ as we shall elaborate in more detail in Section \ref{sec::typeoneerror}. 

The strong null hypothesis $\HF$, on the other hand, further entails $S^2_0 = S^2_1$ and $S^2_\tau = 0$ such that $v_\nm = v_{\nm0} =  (p_0p_1)^{-1} S^2$ and $c_\nm= c_\nm' = 1$, where $S^2$ denotes the common value of $S^2_0$ and $S^2_1$. 
The resulting identical asymptotic distributions of $\htn$ and $\htnp$ are a trivial consequence of their exact equivalence in finite samples.  
From \eqref{eq::standarderrors},  the strong null hypothesis also ensures the asymptotic equivalence of the classic and robust standard errors.

Last but not least, the way in which {\frt} is conducted ensures $(\htn/\ese_\nm)^\pi$ and $(\htn/\tse_\nm)^\pi$ converge in distribution to standard normal. 
Recall   the ``pseudo finite population"   $\{Y_i'(0), Y_i'(1), x_i\}_{i=1}^N$ with $Y_i'(0)=Y_i'(1)=Y_i$  from which the {\frt} procedure generates the randomization distributions of all three test statistics. It satisfies $H_{0\fisher}$, so the same intuition from the last paragraph extends here and ensures the consistency of both $N(\senp)^2$ and $N(\tsenp)^2$ for estimating the asymptotic variance of $\rtn  \htnp$.
This guarantees the convergence of $(\htN /\sen  )^\pi$ and $(\htN /\tse_\neyman)^\pi $ to the standard normal irrespective of the true value of $\tau$. \color{black} The same intuition carries over to the nine covariate-adjusted test statistics 
with the original potential outcomes replaced by the adjusted counterparts.
We formalize the idea in Theorems \ref{rosenbaum}--\ref{lin}.

Let $\gamma_z = (\sxx)^{-1}S_{xY(z)}$ be the coefficient of $x_i$ from the {\olsf} of $Y_i(z)$ on $(1, x_i)$. 
Let $a_i(z) =  Y_i(z) - \bar Y(z) -  x_i^\T (p_1\gamma_1+p_0\gamma_0)$ be the adjusted potential outcomes under treatment $z$, with finite-population mean zero and  variance $ S_{a(z)}^2$.  
Theorem \ref{rosenbaum} gives the asymptotic distributions of $\htR$, $\htR/\ese_\rosenbaum$, and $\htR/\tse_\rosenbaum$ from the first,  pseudo-outcome-based strategy for covariate adjustment.

\begin{theorem}\label{rosenbaum}
{\asym}  
\begin{enumerate}
[(a)]
\item
$\rtn  (\htR-\tau)  \rightsquigarrow  \mN(0, v_\rosenbaum)$,
and
$\rtn   \htrp \rightsquigarrow  \mN(0, v_{\rosenbaum0}) \ \asz$, 
where
$v_\rosenbaum  =  p_1^{-1} S_{a(1)}^2  + p_0^{-1} S_{a(0)}^2  -   S_\tau^2  $
and
$v_{\rosenbaum0}  =   p_0^{-1}  S_{a(1)}^2  + p_1^{-1}  S_{a(0)}^2 + \tau^2 $. 
\item
$(\htR-\tau) /\ese_\rosenbaum  \rightsquigarrow  \mN(0, c'_\rosenbaum)$,
and 
$(\htR /\ese_\rosenbaum)^\pi   \rightsquigarrow  \mN(0, 1) \ \asz$,  where 
$c'_\rosenbaum  = v_\rosenbaum /  (v_{\rosenbaum0}-\tau^2)$.
\item
$(\htR-\tau) /\tse_\rosenbaum   \rightsquigarrow  \mN(0, c_\rosenbaum)$, 
and
$(\htR /\tse_\rosenbaum)^\pi  \rightsquigarrow  \mN(0, 1) \ \asz$, 
where
$c_\rosenbaum =  v_\rosenbaum/(v_\rosenbaum +  S_\tau^2  ) \leq 1.$ 
\end{enumerate}
\end{theorem}

Interestingly, Theorem \ref{rosenbaum} also holds for $\htau_\fisher, \htau_\fisher /\ese_\fisher$, and $ \htau_\fisher/ \tse_\fisher$ from the second,  model-output-based strategy.  
 This echos the numeric result from Proposition \ref{FLRtoN}, which implies that the difference between $\hgf$ and $\hgr$ is of higher order under complete randomization.

\begin{theorem}\label{fisher}
Theorem \ref{rosenbaum} holds if we replace all the subscripts \textsc{r} with \textsc{f}. 
\end{theorem}

The asymptotic equivalence of $\hat{\tau}_\rosenbaum$ and $\hat{\tau}_\fisher$ is perhaps no surprise after all, despite the distinction in procedure.  Both statistics use a common coefficient, namely $\hgr$   and $\hgf$, to adjust the observed outcomes under both treatment and control,  and estimate this coefficient using the pooled data.
Such practice, despite expeditious, can be problematic in experiments with unequal group sizes and heterogeneous treatment effects with respect to covariates \citep{Freedman08a}.

\citet{Lin13}'s estimator, on the other hand, accommodates separate adjustments for outcomes under treatment and control evident from Proposition \ref{FLRtoN}. 
Let $b_i(z) = Y_i(z)   - \bar Y(z)  - x_i^\T\gamma_z$ be the adjusted potential outcomes under treatment-specific coefficient  $\gamma_z$, with mean zero and finite-population variance  $S_{b(z)}^2$. 
Let $S_{\xi}^2$ be the finite-population variance of $\xi_i  =   b_i(1) - b_i(0)$ for $i = \ot{N}$. 
Theorem \ref{lin}  gives the asymptotic distributions of $\htau_\lin$, $\htau_\lin /\ese_\lin$, and $ \htau_\lin/ \tse_\lin$.

\begin{theorem}\label{lin}
{\asym}  
\begin{enumerate}
[(a)]
\item
$\rtn  ( \htl-\tau)    \rightsquigarrow  \mN(0, v_\lin)$,  
and
$\rtn   \htlp \rightsquigarrow  \mN(0, v_{\lin0}) \ \asz$,  where
$v_\lin  = p_1^{-1} S_{b(1)}^2  + p_0^{-1} S_{b(0)}^2  -  S_{\xi}^2$
and $v_{\lin 0} =v_{\rosenbaum 0} =  v_{\fisher 0} =  p_0^{-1}  S_{a(1)}^2  + p_1^{-1}  S_{a(0)}^2  +\tau^2 $.
\item
$
(\htl -\tau) /\ese_\lin   \rightsquigarrow  \mN(0, c'_\lin),
$   
and
$
(\htl /\ese_\lin )^\pi   \rightsquigarrow  \mN(0, 1) \ \asz$, where 
$
c'_\lin  =    v_\lin / (p_0^{-1}  S_{b(1)}^2  + p_1^{-1}  S_{b(0)}^2 ).
$
\item
$
(\htl -\tau)/\tse_\lin    \rightsquigarrow  \mN(0, c_\lin),
$
and 
$
(\htl /\tse_\lin)^\pi   \rightsquigarrow  \mN(0, 1) \ \asz$, where 
$
c_\lin =   v_\lin / (v_\lin +  S_{\xi}^2 ) \leq 1 .
$
\end{enumerate}
\end{theorem}

The asymptotic variance   of $\htl$ is less than or equal to $v_\fisher = v_\rosenbaum$ \citep{Lin13}, but those of the randomization distributions are all equal, $ v_{\lin 0 } =   v_{\rosenbaum 0} =  v_{\fisher 0} $. Similar to the comments on the ``pseudo finite population" after Theorem \ref{neyman}, this is due to the mixing of the treated and control outcomes in the {\frt} procedure, which effectively results in covariate adjustment based on the pooled data even in constructing \citet{Lin13}'s estimator. In fact, Lemma \ref{juxta} in the supplementary material gives a stronger result that $\hat{\tau}_*^\pi \ (*=\rosenbaum, \fisher, \lin)$ are all asymptotically equivalent.

The asymptotic sampling distributions in Theorems \ref{fisher} and \ref{lin} are not new \citep{Freedman08a, Lin13}, but the randomization distributions are. Both the asymptotic sampling and randomization distributions of $\htau_\rosenbaum$, $\htau_\rosenbaum /\ese_\rosenbaum$, and $ \htau_\rosenbaum/ \tse_\rosenbaum$ in Theorem \ref{rosenbaum} are new.   
The analysis of the randomization distributions builds upon 
the existing sampling distributions
but requires additional technical tools, such as the finite-population strong law of large numbers.  We unify the existing and new results in the above four theorems to facilitate discussions on the asymptotic validity. 

Technically, Condition \ref{asym} requires more moments than the usual asymptotic analysis of $\htau_*\ (*=\neyman, \fisher, \lin)$. 
This is due to the strong statement of the almost sure convergence of the randomization distributions in Theorems \ref{neyman}--\ref{lin}. 
This is sufficient but unnecessary for showing that  {\frt} controls the asymptotic type one error rates, which only requires that the quantiles of the asymptotic randomization distribution are greater than or equal to those of the asymptotic sampling distribution. 
We can use the ``subsequence argument," a standard proving device for the bootstrap \citep{vdv-yellowbook}, to relax the moment conditions. However, we keep the current version of Condition \ref{asym} to simplify the statements of the theorems and their proofs.

\begin{remark}\label{rmk:sp}
 All results in Theorems \ref{neyman}--\ref{lin} extend to the super-population framework under independent treatment assignments with minor modifications. 
A key distinction is that the robust standard error $\tse_\lin$ must be modified to ensure consistency \citep{berk2013covariance, negi2020revisiting}. 
Motivated by the similarity in procedure, we also evaluate the design-based properties of five existing permutation tests originally for linear models and show the superiority of the proposed {\frt} for testing the treatment effects \citep{romano, fl, knd, tb, manly}.
We relegate the details to Section \ref{sec:connections} in the supplementary material.
\end{remark}

\subsection{Asymptotic validity for testing $\tau = 0$}\label{sec::typeoneerror}
Theorems \ref{neyman}--\ref{lin} establish the sampling and randomization distributions for all twelve test statistics in Table \ref{tb:tt}
as asymptotically normal.  A statistic as such is proper under two-sided {\frt} if under $\HN$, the asymptotic variance of its randomization distribution is greater than or equal to that of its sampling distribution.
In general, $v_{*}/v_{*0}$ and $c'_*$ can be either greater or less than $1$, suggesting the improperness of $\hts$ and $\htss $ for $*=\neyman,\rosenbaum, \fisher, \lin$.  
On the other hand, $c_* \leq 1$, ensuring the properness of $\hts/\tse_*$ for $* = \neyman, \rosenbaum, \fisher, \lin$.

\begin{corollary}\label{rec}
{\asym}  The robust $t$-statistics $\htH/\tse_* \ph$
are {\proper}
\end{corollary}

Corollary \ref{rec} highlights the necessity of robust studentization in constructing asymptotically valid {\frt} for testing $\HN$. 
The other eight test statistics may also preserve the correct type one error rates asymptotically with additional conditions on $(p_0,p_1)$ or $\mathcal{S}$.
The former is within the control of the designer  whereas the latter is not.

 \begin{corollary}\label{cor:one half}
{\asym}  As $N$ goes to infinity, 
\begin{enumerate}[(a)]
\item all twelve test statistics in Table \ref{tb:tt} preserve the correct type one error rates  if $p_0 = p_1= 1/2$ or $\tau_i = \tau$ for all $i=\ot{N}$; 
\item $\htN$ and $\htns $ preserve the correct type one error rates  if $ S_0^2  =  S_1^2 $; $\htR$, $\htR/\ese_\rosenbaum$, $\htF$, and $\htF/\ese_\fisher$ do so  if $ S_{a(0)}^2  =  S_{a(1)}^2 $; 
$\htl$ and $\htl/\ese_\lin$ do so   if $  S_{b(0)}^2  =  S_{b(1)}^2 $.
\end{enumerate} 
\end{corollary}

 Corollary \ref{cor:one half} states the properness of the unstudentized coefficients and classic $t$-statistics when $p_0 = p_1 = 1/2$. 
The result echos the fact that for the usual two-sample $t$-test, one may use either the pooled or unpooled estimate of the variance whenever the ratio of the two sample sizes tends to one.
\citet{Freedman08a} and  \citet{Lin13} discovered this result in complete randomization under the finite-population inference;  
\cite{bugni2018inference} discovered parallel results  in covariate-adaptive randomization under the super-population inference.

\subsection{Insights for power under alternative hypotheses}\label{sec:power}
 
A natural next question is the relative power of {\frt} with the four robust $t$-statistics under alternative hypotheses. 
Recall that Theorems \ref{neyman}--\ref{lin} hold for arbitrary $\tau$. 
A deviation from $\HN$ shifts the center of $\htau_*/\tse_*$ while leaving its asymptotic randomization distribution intact at $\mN(0,1)$.
{With $|\tau|/\tse_*$ tending to $\infty$ for any fixed $\tau\neq 0$ as $N$ goes to infinity, all four statistics would have power converging to one under the alternative hypothesis  for fixed $\tau\neq 0$, with the relative power determined by $| \hat{\tau}_*/\tse_*|$ as the distance between the observed value of the test statistics, $\hts/\tse_*$, and the center of the reference distribution, namely 0.}
With $\hat\tau_* $ converging to $\tau$ in probability for all $* =\neyman,\rosenbaum,\fisher,\lin$, a heuristic argument is that the smaller the robust standard error, the higher the asymptotic relative power. 
We give in Corollary \ref{power} the relative order of the robust standard errors as $N$ goes to infinity, and demonstrate its heuristic relation with the power via simulation in Section \ref{simulation}.
Rigorous quantification of the power entails specification of the alternative distributions \citep{lehmann1975nonparametrics, rosenbaum2020}. 
We leave the technical details to future work.

\begin{corollary}\label{power} 
Assume complete randomization and Condition \ref{asym}. We have 
\begina
\frac{\tse_*^2}{\tse_\neyman^2} -  \frac{p_1^{-1}  S_{a(1)}^2+ p_0^{-1}  S_{a(0)}^2}{p_1^{-1}  S_1^2 + p_0^{-1}  S_0^2} = o(1) \text{\quad for $*=\rosenbaum$ and $\fisher$,}\quad\quad 
\frac{\tse_\lin^2}{\tse_\neyman^2} - \frac{p_1^{-1}  S_{b(1)}^2 + p_0^{-1}  S_{b(0)}^2}{p_1^{-1}  S_1^2 + p_0^{-1}  S_0^2} =o(1)
\enda
hold $\asz$, with the limiting values of $\tse_\lin^2/\tse_*^2$ less than or equal to 1  for $*=\neyman, \rosenbaum, \fisher$. 
\end{corollary}

With $S^2_{b(z)} \leq S^2_{a(z)}$ and $S^2_{b(z)} \leq S^2_{ z }$ for $z=0,1$, Corollary \ref{power} ensures $\tse_\lin$ has the smallest limiting value among $\tse_* \ (* = \neyman, \rosenbaum, \fisher, \lin)$, and thereby ensures $\htl/\tse_\lin$ has the highest power asymptotically. 
The limiting values of $\tse_\rosenbaum$ and $\tse_\fisher$, on the other hand, can be even greater than that of $\tse_\neyman$.
This mirrors the asymptotic efficiency theory of point estimation and suggests $\htR/\tse_\rosenbaum$ and $\htF/\tse_\fisher$ can be even less powerful than  $\htnsr$ despite the extra use of covariates \citep{Freedman08a, Lin13}. 
See also \citet[Corollary 1]{fogarty2018mitigating} for analogous results in the context of finely stratified experiments.

{\frt} with $\htl/\tse_\lin$, as a result, is finite-sample exact for testing $\HF$, asymptotically valid for testing $\HN$, and enjoys the highest power under alternatives, all irrespective of whether the linear model that generates it is correctly specified or not. 
It is thus our final recommendation for testing both $\HF$ and $\HN$ by {\frt}.

\subsection{Confidence interval by inverting FRTs}\label{CI}

We next extend the theory from testing hypotheses to constructing confidence intervals. 
This is conceptually straightforward given their duality. 
Consider using {\frt} to test $\HN(c): \tau = c$. 
We can pretend to be testing a strong null hypothesis of constant effect,  $\HF(c):\tau_i=c $ for all $i = \ot{N}$, 
and compute the $p$-value, denoted by $p_\frt(c)$, by using $ Y  - cZ$ as the fixed outcomes for  {\frt}.
Inverting a sequence of such {\frt}s on a bounded set of the possible values of $c$ yields 
$ 
\textsc{ci}_{\frt, \alpha}
$ 
as a tentative interval estimator for the average treatment effect $\tau$. 
By duality, it is an asymptotic $1-\alpha $ confidence interval for $\tau$  if we use the robust $t$-statistics to perform the {\frt}s. Duality further suggests the one based on $\htl/\tse_\lin$ to have the smallest width asymptotically.

Alternatively, the robust Wald-type confidence intervals $(\hat{\tau}_*-q_{1-\alpha/2}\tse_*, \hat{\tau}_* + q_{1-\alpha/2}\tse_*)\ (*=\neyman, \fisher, \rosenbaum, \lin)$ cover $\tau$ with probability approaching $1-\alpha$ as $N$ goes to infinity, where $q_{1-\alpha/2}$ is the $(1- \alpha/2)$th quantile of the standard normal.  These confidence intervals are asymptotically identical to $\textsc{ci}_{\frt, \alpha}$ based on the robust $t$-statistics. They are convenient approximations for $\textsc{ci}_{\frt, \alpha} $ which can be used as initial values in the grid search over $c$. We recommend using $\textsc{ci}_{\frt, \alpha} $ based on $\htl/\tse_\lin$ because of its multiple guarantees: 
it has finite-sample exact coverage rate when $\tau_i = \tau$ for all $i = \ot{N}$, has correct asymptotic coverage rate when $\tau_i$'s vary, and has smaller width compared to the confidence interval without covariate adjustment.

 \section{Extensions to other experimental designs}

\subsection{Cluster randomization}\label{sec:cluster}
Consider $N$ units nested in $M$ clusters of sizes $n_i \  (i = \ot{M}; \ \sumM n_i = N)$. The average cluster size is $\bar n = N/M$.  
Cluster randomization randomly assigns $M_1$ clusters to receive the treatment and the rest $M_0 = M - M_1$ clusters to receive the control.
Let $x_{ij}$ and $\{Y_{ij}(z): z= 0,1\}$ be the covariate and potential outcomes for the $j$th unit in cluster $i \ (i=1,\ldots,M;\ j=1,\ldots, n_i)$.
 The  average treatment effect equals
\beginy\label{cr}
\tau  = N^{-1}  \sumM \sum_{j=1}^{n_i} \{ Y_{ij}(1) - Y_{ij}(0) \}
= M^{-1} \sum_{i=1}^M
\{\tilde{Y}_{i\cdot}(1)-\tilde{Y}_{i\cdot}(0)\},
\endy
where $\tilde{Y}_{i\cdot}(z)=\sum_{j=1}^{n_i}Y_{ij}(z) / \bar n$ is the cluster total of potential outcomes  scaled by $1/\bar n$. 

Let $Z_i$ be the treatment level received by cluster $i$.  
The observed outcome for unit $ij$ is $Y_{ij}=Z_i Y_{ij}(1)+(1- Z_i )Y_{ij}(0)$. 
Let $\tilde{Y}_{i\cdot}= \sum_{j=1}^{n_i}Y_{ij} / \bar n$ 
and $\tilde{x}_{i\cdot}= \sum_{j=1}^{n_i}x_{ij} / \bar n$
be the scaled cluster totals of observed outcomes and covariates in cluster $i$.
Then $\tilde Y_{i\cdot} = Z_i \tilde Y_{i\cdot} (1)+ (1-Z_i) \tilde Y_{i\cdot}(0)$ gives the observed analog of $\tilde Y_{i\cdot}(z)$ under cluster randomization. 
This, together with the expression of $\tau$ from  \eqref{cr}, suggests the equivalence of $( Z_i,  \tilde{Y}_{i\cdot},  \tilde{x}_{i\cdot})_{i=1}^M$ to data from a complete randomization with potential outcomes $\{ \tilde Y_{i\cdot}(0), \tilde Y_{i\cdot}(1)\}_{i=1}^M$ and average treatment effect $\tau$  \citep{MiddletonCl15, DingCLT}, and allows us to derive results in parallel with Theorems \ref{neyman}--\ref{lin} under a modified version of Condition \ref{asym} on cluster totals by replacing $Y_i(z)$ with $\tilde Y_{i\cdot}(z)$, $x_i$ with $\tilde x_{i\cdot}$, and $(N, N_0, N_1)$ with $(M, M_0, M_1)$ as $M$ goes to infinity.
This requires a large number of clusters to ensure the  accuracy of the asymptotic approximation.
See \cite{DS} for more subtle requirements on the cluster sizes when the regularity conditions are given in terms of the individual potential outcomes.

\subsection{Stratified randomization}
Consider $N$ units in $K$ strata of sizes $\nj \ (k = 1, \dots, K; \ \sumj \nj = N)$.
Stratified randomization conducts an independent complete randomization in each stratum, and assigns at complete random $N_{[k]z}$ units to receive treatment $z$ in stratum $k \ ( k = \ot{K}; \ z = 0,1)$.
Denote by $x_\ki$, $\{Y_\ki(z): z = 0,1\}$, and $Z_\ki$ the covariate, potential outcomes, and treatment indicator for the $i$th unit in stratum $k \ ( k = \ot{K}; \ i = \ot{\nj})$.
The randomization scheme independently draws $Z_{[k]} = (Z_{[k]1}, \dots, Z_{[k]N_{[k]}})^\T$ as a random permutation of $N_{[k]1}$ $1$'s and $N_{[k]0}$ $0$'s  for $k=1,\ldots, K$. Equivalently, it draws
$Z = (Z_{[1]}^\T, \dots, Z_{[K]}^\T)^\T$  uniformly from the set 
\begina
\mathcal{Z}_{\textup{str}} = \left\{ \bm z = (z_{[k]i}) \in \{0,1\}^N: \sum_{i=1}^{N_{[k]}} z_{[k]i} = N_{[k]1} \text{ for } k = \ot{K}\right\}
\enda
subject to the stratum-wise treatment size restriction.
Let $Y$ and $X$ be the vectorization and concatenation of $Y_{[k]i}(z)$'s and $x_{[k]i}(z)$'s, 
respectively. 
For arbitrary test statistic $T(Z, Y, X)$, a two-sided {\frt} under stratified randomization permutes the treatment vector $Z$ within $\mathcal{Z}_{\textup{str}}$, and computes the $p$-value as 
\begina
p_{\frt, \textup{str}} = |\mathcal{Z}_\textup{str}|^{-1}\sum_{ \pi: Z_\pi\in \mathcal{Z}_\textup{str}}1\{ |T(Z_\pi,   Y, X )| \geq |T(Z, Y, X)|\}.
\enda
The finite-population average treatment effect equals 
$$
\tau = N^{-1}\sumj \sumji \{Y_\ki(1)-Y_\ki(0)\} = \sumj \omega_{[k]} \tauj,
$$
where $\omega_{[k]} = \nj / N$ and 
$\tauj = \nj^{-1} \sum_{i=1}^{\nj} \{Y_\ki(1)-Y_\ki(0)\} $
define the relative size and stratum-wise average treatment effect of stratum $k$, respectively. 
Of interest is the choice of the test statistic that ensures valid test of $\HN: \tau = 0$ via \frt.


Assume $\hat{\tau}_{*[k]}$ and  $\tse_{*[k]}$ as the basic estimator and robust standard error obtained from stratum $k$, where $*$ can be $\neyman$, $\rosenbaum$, $\fisher$, and $\lin$.  
The weighted average $\hat{\tau}_{*} = \sumj   \omega_{[k]} \hat{\tau}_{*[k]}$, with a slight abuse of notation, affords an intuitive estimator of $\tau$ with squared robust standard error $\tse^2_{*} = \sumj  \omega_{[k]}^2 \tse^2_{*[k]} $. 
The abuse of notation causes little confusion because $\hat{\tau}_{*}$ and $\tse_{*} $ reduce to their definitions under complete randomization when $K=1$. 
This suggests $\hat{\tau}_{*} / \tse_{*} \ (*=\nm,\rosenbaum,\fisher,\lin)$ as four intuitive choices of the test statistic for testing both $\HF$ and $\HN$ under stratified randomization.
The properness of $\hat{\tau}_{*} / \tse_{*} $ for testing $\HN$ is a direct application of Theorems \ref{neyman}--\ref{lin}.

 \begin{corollary}\label{sre}
Assume stratified randomization and Condition \ref{asym} holds within all strata $k=\ot{K}$. We have $(\hat\tau_*-\tau) /\tse_*   \rightsquigarrow  \mN(0, c_*)$,  and $(\hat\tau_* /\tse_*)^\pi \rightsquigarrow  \mN(0, 1)$ $P_Z$-a.s., where $c_* \leq 1$, for $*= \textsc{n}, \textsc{r}, \textsc{f}, \textsc{l}$.
\end{corollary}

\cite{bugni2018inference} focused on covariate-adaptive experiments in which the proportions of the treatment are homogeneous across strata. 
They showed that in those covariate-adaptive experiments, one can form a simpler estimator from the {\ols} fit of the outcome on the treatment and the stratum indicators. 
When the proportions of treatment vary across strata,  this simple {\ols} fit yields inconsistent estimator for $\tau$. 
To address this issue, \cite{bugni2019} proposed to use the weighted average of the stratum-specific estimators, and \cite{liu2019regression} and \cite{ye2020inference} discussed covariate adjustment allowing for additional covariates beyond the stratum indicators. 
We further their theory to {\frt}s,  and use the weighted average of the stratum-specific estimators to allow for the proportions of treatment to vary across strata. 
%
%

Even if the original experiment is completely randomized,  if a discrete covariate $X$ is available, we can condition on the numbers of treated and control units landing in each stratum.  
The resulting assignment mechanism is identical to stratified randomization, 
such that we can permute the subvector of $Z$ within each stratum of $X$ as if the original experiment were stratified.
This is known as the {\it conditional randomization test}. 
\citet{zheng2008multi} and \citet{hennessy2016conditional} observed that
they typically enhance the power if the covariates are predictive of the outcomes.

Among the four robust $t$-statistics,  $\hat{\tau}_\neyman / \tse_\neyman $ is the simplest and $\hat{\tau}_\lin / \tse_\lin $ is the most powerful.
Corollary \ref{sre} assumes $N_{[k]}$ goes to infinity for each $k$. With a large number of small strata, we need to modify the test statistic and the asymptotic scheme \citep{liu2019regression}. Since this involves different technical tools, we defer the technical details to future work.

 \subsection{Rerandomization}\label{rem}

\subsubsection{FRT with rerandomization}\label{sec:rem_basic}

Rerandomization, termed by \citet{cox:1982} and \citet{morgan2012rerandomization}, samples the treatment indicators under covariate balance constraints. 
\cite{bm} reviewed several field experiments in economics and suggested that rerandomization is widespread although often poorly documented. \cite{banerjee} discussed the pros and cons of such experimental design. Although rerandomization can improve covariate balance, it also imposes challenges for  the subsequent data analysis.
We focus here on a special rerandomization that uses the Mahalanobis distance between covariate means as the balance criterion, known as ReM. 
Although it might not be the exact rerandomization used in field experiments in economics, it has nice statistical properties that allow for simple analysis of the experimental data.

%
%
%
%
%

The designer of ReM accepts a treatment vector $Z$ if and only if 
\beginy\label{A}
\ma: \quad \htx^\T \{\cov(\htx)\}^{-1}\htx <a
 \endy
for a predetermined constant $a$. 
Let $\pia= \{\bm z: \text{$\bm z \in \mZ$ satisfies \eqref{A}}\}$ be the set of acceptable assignments under threshold $a$. 
The sampling distribution of the test statistic $T$ is uniform over $  \{ T(\bm z, Y(\bm z), X) : \bm z \bm \in\pia  \}$. 
{\frt} under ReM proceeds by permuting $Z$ within  $\pia$  and computes the    $p$-value as
\beginy\label{pfrt_a}
p_{\frt, \ma} = |\pia|^{-1}\sum_{ \pi: Z_\pi\in \pia}1\{ |T(Z_\pi,   Y, X )| \geq |T(Z, Y, X)|\}.
\endy
It compares the observed value of $T$ to its randomization distribution under ReM, denoted by $T^{\pi|\ma}$. Under ReM in \eqref{A}, $p_{\frt,\ma}$ is  finite-sample exact for $\HF$ for arbitrary $T$. 
Of interest is its large-sample validity for testing $\HN$, which depends on the stochastic dominance relation between the asymptotic distributions of the test statistic. Theorem  \ref{Rem} summarizes the results based on the additional notation below. 

Let $\ep \sim \mN(0,1)$ and $\mL \sim D_1 \mid (\|D\|_2^2\leq a)$,  where $D = (D_1, \dots, D_J)^\T \sim \mN(0_J,I_J)$, be independent standard and truncated normals, respectively, and let $ r_{J,a} = P (\chi^2_{J+2} \leq a)/ P (\chi^2_J \leq a) \in (0, 1]$ be the variance of $\mL$.
Let $\mU (\rho ) = (1-\rho ^2)^{1/2}\cdot \ep + \rho \cdot \mL $ be a linear combination of $\ep$ and $\mL$ for $\rho  \in [0,1]$ with mean 0 and variance $v(\rho ) = 1-(1- r_{J,a}) \rho ^2$.
Recall $v_*$ and $v_{*0}$ in Theorems \ref{neyman}--\ref{lin} as the asymptotic variances of $\rtn  \hts$ and $\rtn \htsp$ under complete randomization, with $v_\rosenbaum = v_\fisher$ and $v_{\rosenbaum0}= v_{\fisher0}$. Let 
$
\rho _*^2   = 1-v_\lin/v_* 
$
and
$
 \rho _{* 0 }^2  =   1-v_{\lin0}/v_{*0} 
$
for $*=\neyman, \rosenbaum, \fisher$, with $\rho _{\rosenbaum 0} =  \rho _{\fisher 0 } = 0$. 

\begin{theorem}\label{Rem}
{\assrem}
\begin{enumerate}
[(a)]
\item
$
\rtn  ( \htN - \tau)   \rightsquigarrow  v^{1/2}_\neyman \cdot \mU (\rho _\neyman),\ 
 ( \htN - \tau) /\sen     \rightsquigarrow  (c'_\neyman)^{1/2} \cdot \mU (\rho _\neyman)
$,
and
$
 ( \htN - \tau) /\tsen      \rightsquigarrow  c^{1/2}_\neyman \cdot \mU (\rho _\neyman);
$
$
\rtn  \htN^{\pi|\ma}  \rightsquigarrow  v_{\neyman0}^{1/2} \cdot \mU (\rho _{\neyman0}) ,\ 
(\htns )^{\pi|\ma}   \rightsquigarrow   \mU (\rho _{\neyman0}) 
$,
and
$
(\htnsr)^{\pi|\ma}  \rightsquigarrow  \mU (\rho _{\neyman0})
$
hold $P_Z$-a.s..

\item 
$ \rtn  ( \hts - \tau)   \rightsquigarrow  v^{1/2}_* \cdot \mU (\rho _*), \ 
( \hts - \tau) /\ese_*     \rightsquigarrow  (c'_*)^{1/2} \cdot \mU (\rho _*)
$,
and
$
( \hts - \tau) /\tse_*      \rightsquigarrow c^{1/2}_*  \cdot \mU (\rho _*);
$
$
\rtn  \hts^{\pi|\ma}  \rightsquigarrow  \mN(0, v_{*0}) ,\ 
(\htss )^{\pi|\ma}   \rightsquigarrow   \mN(0,1)
$,
and
$ (\htssr)^{\pi|\ma}  \rightsquigarrow   \mN(0,1)$ hold $P_Z$-a.s. ($* = \rosenbaum, \fisher$).
 \item $\htl$, $\htl/\ese_\lin$, and $\htl/\tse_\lin$ have identical sampling and randomization distributions as under complete randomization in Theorem \ref{lin}.
\end{enumerate}
\end{theorem}

Compare Theorem \ref{Rem} under ReM with Theorems \ref{neyman}--\ref{lin} under complete randomization. 
The asymptotic sampling and randomization distributions of $\htn$, $\htns$, and $\htnsr$ change to non-normal. 
The asymptotic sampling distributions of $\hts$, $\htss$, and $\htssr\ (* = \rosenbaum, \fisher)$ change to non-normal, whereas their asymptotic randomization distributions remain the same. 
ReM does not affect these two sets of asymptotic randomization distributions because of the asymptotic independence between $\hts^\pi \ ( * = \rosenbaum, \fisher)$ and $\htx^\pi$. 
The asymptotic sampling and randomization distributions of $\htl$, ${\htl}/{\ese_\lin}$, and ${\htl}/{\tse_\lin}$ all remain unchanged. ReM does not affect them because of   the asymptotic independence between $\htl$ and $\htx$ and that between $\htl^\pi$ and $\htx^\pi$.

In the case of symmetric yet non-normal limiting distributions as those of $\hts$, $\htss$, and $\hts/\tse_*$ for $*=\neyman, \rosenbaum, \fisher$,  determination of properness entails comparisons of not only the variances but also all the central quantile ranges. 
A test statistic $T$ is proper under a two-sided {\frt} if $T$ has wider or equal central quantile ranges than $T^{\pi|\ma}$  for all quantiles.

\begin{corollary}\label{cor:rem}
{\assrem} The covariate-adjusted robust $t$-statistics $\hat{\tau}_* / \tse_*\ (*=\rosenbaum, \fisher, \lin)$
are {\proper} 
\end{corollary}

Compare Corollary \ref{cor:rem} with Corollary \ref{rec} to see that the unadjusted $\htnsr$, whereas proper under complete randomization,  is no longer proper under ReM due to  the non-normal limiting distribution of $\htnp$.
\citet{colin} also noticed this phenomenon and gave a numeric example. 
They proposed a prepivoting approach to improve studentization. 
We do not pursue that direction given $\htnsr$ is inferior to $\htL/\tse_\lin$  even under complete randomization.  
The three covariate-adjusted robust $t$-statistics, namely $\htR/\tse_\rosenbaum$, $\htF/\tse_\fisher$, and  $\htl/\tse_\lin$, are the only options in Table \ref{tb:tt} proper for testing $\HN$ under ReM. 
Covariate adjustment is thus essential for securing properness under ReM in addition to robust studentization.  
The same reasoning as that leads to Corollary \ref{power} ensures \textsc{frt} with ${\htl}/{\tse_\lin}$ delivers the highest power among the three proper statistics. 
It is thus our recommendation for conducting \textsc{frt} under ReM.

\begin{remark}
When covariates   have different levels of importance for the outcomes, \citet{morgan2012rerandomization} proposed using ReM with differing criteria for different tiers of covariates. The resulting {\frt}  permutes the treatment vector $Z$ within the subset of $\mZ$'s that satisfy the tiered balance criteria. 
The   sampling and randomization distributions of the twelve test statistics in Table \ref{tb:tt} parallel those in Theorem \ref{Rem}, with ${\htl}/{\tse_\lin}$ being  the most powerful among the proper options. It is thus our recommendation for this extension as well. We omit the technical details due to its repetitiveness.
\end{remark}

\subsubsection{FRT in the case of designer-analyzer information discrepancy}\label{sec:rem_d}
Discussion so far assumes the analyzer and the designer  use the same covariates $(x_i)_{i=1}^N$ and threshold $a$ for doing ReM in the design and analysis stages, respectively.  
An interesting question, also a real concern in practice, is what if the designer and the analyzer do not communicate? \citet{bm},  \cite{hk}, and \cite{heckman2020} gave examples arising in field experiments in economics. \citet{LD20} discussed optimal covariate adjustment based on estimation precision.

A relatively easy case is that the analyzer has access to additional covariates beyond those used in the design of ReM. 
Using \textsc{frt} under this ReM with ${\htl}/{\tse_\lin}$ is again our recommendation in this case. 
A more challenging case is that the analyzer is either unaware of the ReM in the design stage or does not have access to all covariates used in the ReM.  
In the absence of full information about the design, \citet{hk} proposed to  use the maximum $p$-value from the worst-case {\frt} over a set of designs consistent with the available information.
Without completely specifying these designs, an alternative  option is to use $p_\frt$ in \eqref{pfrt} such that the analysis coincides with that under complete randomization. 
Under $\HF$, the finite-sample exactness is lost unless the original experiment is  indeed completely randomized. Of interest is how such information discrepancy further affects the test's properness for testing $\HN$.  

Keep $x_i$ as the covariates the analyzer uses in the analysis stage, and let $\xdi $ be the covariates the designer used for conducting ReM in the design stage, 
possibly different from $x_i$.  
The designer accepts an allocation if  $\htd^\T \{\cov(\htd)\}^{-1} \htd <a$ with $\htd $ being the difference in means of $d_i$. 
The analyzer, on the other hand,  uses $X = (x_1, \dots, x_N)^\T$ in addition to $Y$ and $Z$ to form the test statistic, and proceeds with the standard, ``unrestricted" {\frt} that permutes $Z$ over all possible permutations in $\mathcal{Z}$. 
Of interest is whether the resulting $p$-value, namely $p_\frt$ in \eqref{pfrt}, can still preserve the correct type one error rates under $\HN$ despite the information discrepancy.   

Focus on the twelve test statistics in Table \ref{tb:tt} for the rest of the discussion.
 The key is, again, the comparison of the stochastic dominance relations between their respective sampling and randomization distributions when only the weak null hypothesis holds. 
The randomization distributions, on the one hand, are readily available from Theorems \ref{neyman}--\ref{lin} given the analysis is based on the unrestricted \frt.
The possible discrepancy between $\xdi $ and $x_i$, on the other hand, causes the sampling distributions to deviate from those in Theorem \ref{Rem}.
 We furnish this missing piece in Proposition \ref{thm:rem_d}, and state the sampling distributions of the twelve test statistics under ReM using $d_i$'s for arbitrary $\mathcal{S}' = \{Y_i(0), Y_i(1), x_i, d_i\}_{i=1}^N$  that satisfies the regularity conditions.  

Let $S_{z|d}^2$, $S_{a(z)|d}^2$, $S_{b(z)|d}^2$, $S_{\tau|d}^2$, and $S_{\xi|d}^2$ be the finite-population variances of the linear projections of $Y_i(z)$, $a_i(z)$, $b_i(z)$, $\tau_i$, and $\xi_i$ onto $\xdi $, respectively, for $z = 0,1$.  Let  
\begina
&&\rho_{\neyman\mid d}^2 =  \frac{p_1^{-1}S_{1|d}^2+ p_0^{-1}S_{0|d}^2 - S_{\tau|d}^2}{p_1^{-1}S_1^2+ p_0^{-1}S_0^2 - S_{\tau }^2},  \\
&&\rho_{\rosenbaum\mid d}^2 = \rho_{\fisher\mid d}^2 = \frac{p_1^{-1}S_{a(1)|d}^2+ p_0^{-1}S_{a(0)|d}^2 - S_{\tau|d}^2}{p_1^{-1}S_{a(1)}^2+ p_0^{-1}S_{a(0)}^2 - S_{\tau }^2} ,\\
&&\rho_{\lin\mid d}^2 =  \frac{p_1^{-1}S_{b(1)|d}^2+ p_0^{-1}S_{b(0)|d}^2 - S_{\xi|d}^2}{p_1^{-1}S_{b(1)}^2+ p_0^{-1}S_{b(0)}^2 - S_{\xi }^2}
\enda
be the squared multiple correlations between $\hts \ (\phs)$ and $\htau_d$.

\begin{proposition}\label{thm:rem_d} 
Assume Condition  \ref{asym} holds for 
$\{  Y_i(0), Y_i(1), x_i' \} _{i=1}^N$ with $x_i'$ being the union of the $x_i$ and $d_i$, and ReM using $d_i$'s. For $\phs$, we have 
\begina
\rtn  (\hts  - \tau)     \rightsquigarrow  v_*^{1/2} \cdot  \mU (\rho _{*\mid d}), \quad 
(\hts - \tau) /\ese_*   \rightsquigarrow  (c_*')^{1/2} \cdot  \mU (\rho _{*\mid d}),\quad 
(\hts - \tau) /\tse_*    \rightsquigarrow c_*^{1/2} \cdot  \mU (\rho _{*\mid d}) .
\enda
\end{proposition}

Assuming the actual assignment is conducted by ReM using $d_i$'s that are possibly different from $x_i$'s, Proposition \ref{thm:rem_d} is a special case of \citet{LD20} and  generalizes the sampling distributions in Theorem \ref{Rem} to allow for distinct covariates for the design and analysis stages, respectively. The resulting sampling distributions are in general scaled $\mU$ distributions as linear combinations of independent standard and truncated normals.  In particular, $\rho_{\lin\mid d} = 0$ if $x_i$ can linearly represent $\xdi $, rendering the limiting distributions of $\htl$, $\htl/\ese_\lin$, and $\htl/\tse_\lin$ identical to those under complete randomization in Theorem \ref{lin}. 
The following corollary holds by comparing Proposition \ref{thm:rem_d} with Theorems \ref{neyman}--\ref{lin}.

\begin{corollary}\label{cor:rem_d}
Assume Condition  \ref{asym} holds for 
$  \{ Y_i(0), Y_i(1), x_i'\}_{i=1}^N $ with $x_i'$ being the union of the $x_i$ and $d_i$, ReM using $d_i$'s in design, and $p_\frt$ in \eqref{pfrt} in analysis. The robust $t$-statistics $\hat{\tau}_* / \tse_*\ (*=\neyman, \rosenbaum, \fisher, \lin)$
are {\proper}
\end{corollary}

The four robust $t$-statistics thus ensure $p_\frt$ in \eqref{pfrt} remains asymptotically valid  under ReM even if the analyzer has only partial information on the covariates  the designer used to form the balance criterion.
Ironically, a less informed analysis restores the properness of $\hat{\tau}_\neyman / \tse_\neyman$ under ReM by restoring its asymptotic randomization distribution back to the standard normal. Nevertheless, this properness comes at the cost of being overly conservative.

Further, the asymptotic randomization distributions of $\hat{\tau}_* / \tse_*\ (* = \rosenbaum, \fisher, \lin)$ remain  unchanged in computing $p_\frt$ in \eqref{pfrt} and $p_{\frt,\ma}$ in \eqref{pfrt_a}. 
It  might thus be tempting to ignore the rerandomization and conduct unrestricted \textsc{frt} in the analysis stage whatsoever, even when exact information is available.
We do not encourage such practice given its lack of finite-sample exactness under $\HF$ in the first place.

\section{Simulation} 
\label{simulation}

We examine in this section the validity and power of the proposed method for testing the weak null hypothesis via simulation. 
We conducted the simulation under complete randomization, stratified randomization, and rerandomization, respectively, and summarized the $p$-values over 1,000 independent repetitions to approximate the error rates. 
The patterns are almost identical for the three design types, 
highlighting the importance of robust studentization and efficient covariate adjustment for securing large-sample validity and additional power, respectively. 
To avoid repetitiveness, 
we present below the results from the stratified randomization.

We first examine the large-sample validity of the twelve test statistics under $\HN$.
Consider a finite population of $N=500$ units, $i = \ot{N}$, with a univariate covariate, $(x_i)_{i=1}^N$, as i.i.d.\ ${\rm Unif}(-1,1)$. 
We generate the potential outcomes as $Y_i(1) \sim \mathcal{N}( x_i^3 , 1)$ and $Y_i(0) \sim \mathcal{N}( -x_i^3 , 0.5^2)$, and center $Y_i(1)$'s and $Y_i(0)$'s respectively to ensure $\tau= 0$. 

We divide the units into $K=3$ strata by the values of their covariates at cutoffs $-0.3$ and $0.3$. 
The resulting strata consist of units with $x_i$'s in $[-1, -0.3]$,  $x_i$'s in $(-0.3, 0.3]$, and $x_i$'s in $(0.3,1]$, respectively.  
Denote by $N_{[k]}$ the number of units in stratum $k$, and set $N_{[k]1} = [0.2N_{k}]$ and $N_{[k]0} = N_{[k]} - N_{[k]1}$ as the corresponding stratum-wise treatment sizes.  
We fix $\{Y_i(0), Y_i(1), x_i\}_{i=1}^N$ in the simulation, and draw a random permutation of $N_{[k]1}$ 1's and $N_{[k]0}$ 0's within stratum $k$ for $k = 1,2,3$ to obtain the observed outcomes and conduct {\frt}s.  

The procedure is repeated 1,000 times, with the $p$-values approximated by 500 independent permutations of the treatment vector in each replication. 
Figure \ref{fig::typeoneerror-stratified}(a) shows the $p$-values under $\HN $.  
The four robust $t$-statistics, 
as shown in the last row, are the only ones that preserve the correct type one error rates. 
In fact, they are conservative, which is coherent with Corollary \ref{sre}. 
All the other eight statistics yield type one error rates greater than the nominal levels and are thus not proper for testing $\HN$.

We then evaluate the power of the four proper test statistics when $\tau \neq 0$. 
Take $Y_i(1) \sim \mathcal{N}( 0.1 + x_i , 0.4^2)$ and $Y_i(0) \sim \mathcal{N}( -x_i , 0.1^2)$
for an alternative with $\tau$ close to 0.1, and inherit the rest of the settings from the last two paragraphs. 
Figure \ref{fig::typeoneerror-stratified}(b) shows the $p$-values of the four proper test statistics under the alternative. 
The theoretically most powerful $\hat{\tau}_\lin/\tse_\lin$ indeed delivers the highest power among the four proper options.
The tests based on $\htF/\tse_\fisher$ and $\htR/\tse_\rosenbaum$, on the other hand, show even lower power than the unadjusted $\htnsr$. 
This is coherent with the theoretical results from Corollary \ref{power} and concludes $\htl/\tse_\lin$ as our final recommendation for conducting {\frt} under stratified randomization.

\begin{figure}[t]
\centering
\includegraphics[width=0.9\textwidth]{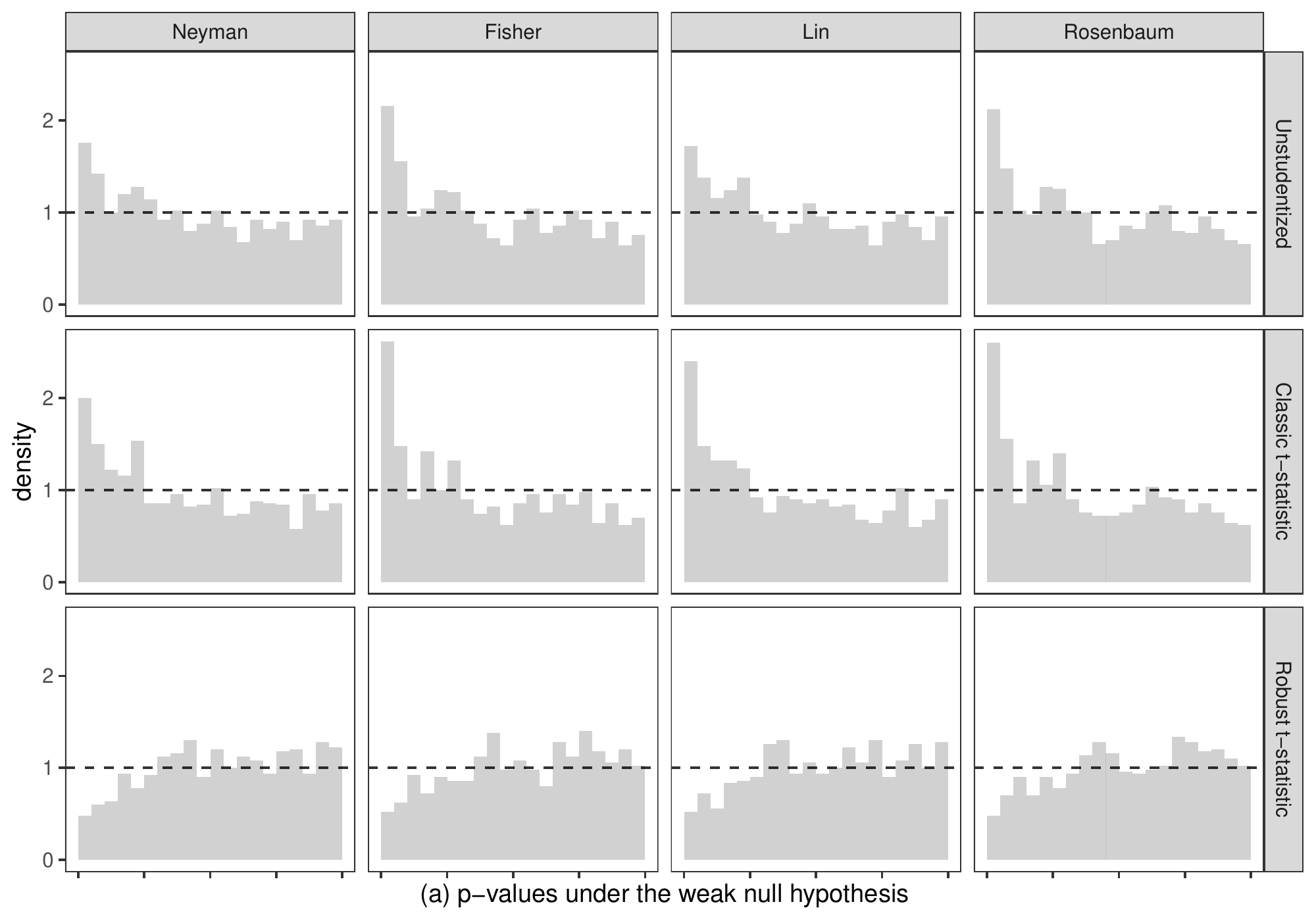}
\includegraphics[width=0.9\textwidth]{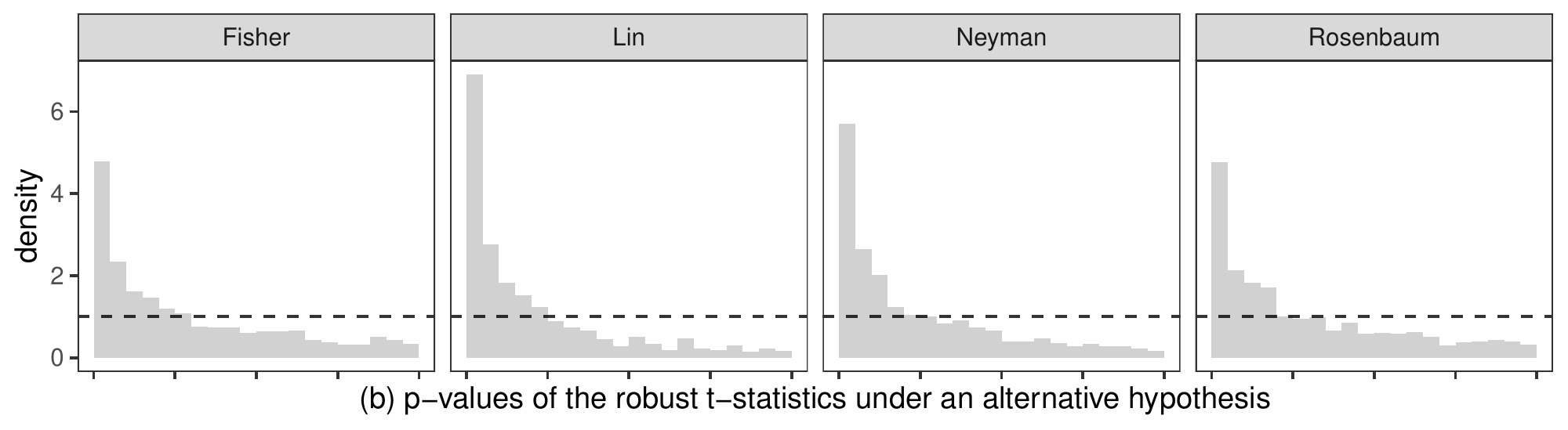}
\caption{Empirical histograms of the $p_{\frt,\textup{str}}$'s with 20 bins in $(0,1)$.
}\label{fig::typeoneerror-stratified}
\end{figure}

\section{Application} 
\citet{chong2016iron} conducted a randomized experiment on   219 students of a rural secondary school in the Cajamarca district of Peru during the 2009 school year. 
They first provided the village with free iron supplements  and trained the local staffs to distribute one free iron pill  to any adolescent who requested one in person. 
They then randomly assigned the students to three arms with three different types of videos: in the first video, a popular soccer player was encouraging the use of iron supplements to maximize energy (``soccer'' arm); in the second video, a physician was encouraging the use of iron supplements to improve overall health (``physician'' arm); the third video did not mention iron and served as the control (``control'' arm). 
The experiment was stratified by the class level from 1 to 5. 
The treatment group sizes within classes are shown in the matrix below:
$$
\bordermatrix{
                     & \text{class 1}& \text{class 2} &\text{class 3} &\text{class 4} &\text{class 5}      \cr
\text{soccer} & 16& 19& 15 &10 &10 \cr
  \text{physician} &    17 &20& 15& 11 &10 \cr
  \text{control}      & 15 &19& 16 &12 &10 \cr
}.
$$


One outcome of interest is the average grade in the third and fourth quarters of 2009, and an important background covariate is the anemia status at baseline. 
We make pairwise comparisons of the ``soccer'' arm versus the ``control'' arm and the ``physician'' arm versus the ``control'' arm. We also compare   {\frt}s with and without adjusting for the covariate of baseline anemia status. 
We use their data set to illustrate   {\frt}s under complete randomization and stratified randomization.  
The ten subgroup analyses use   {\frt}s for complete randomization within each class level.
The two overall analyses use   {\frt}s for stratified randomization averaging over all class levels.

\begin{table}[t]
\centering
\caption{Re-analyzing the data from \citet{chong2016iron}. ``\textsc{n}'' denotes the unadjusted estimators and tests, and ``\textsc{l}'' denotes the covariate-adjusted estimators and tests. The ``$p_\frt$" values for the overall comparisons in the last two rows are all $p_{\frt,\text{str}}$. }\label{tb::pv-values-chong}
\begin{subtable}[h]{0.48\textwidth}
\caption{soccer versus control}
    \begin{tabular}{rrrrr}
    \hline 
   &    est & s.e. &$p_\text{normal}$ & $p_\frt$ \\\hline
       class 1& &&&\\
\textsc{n}  &  0.051 &0.502 &     0.919  & 0.924 \\
\textsc{l}   &    0.050 & 0.489 &     0.919 &  0.929 \\
       class 2& &&&\\
\textsc{n}  &  $-0.158$ &0.451  &    0.726  & 0.722\\
\textsc{l}   &  $-0.176$ & 0.452   &   0.698  & 0.700 \\
       class 3& &&&\\
\textsc{n}  &  0.005 &0.403  &    0.990 &  0.989 \\
\textsc{l}   &   $-0.096$& 0.385 &     0.803&   0.806 \\
       class 4& &&&\\
\textsc{n}  &  $-0.492$& 0.447  &    0.271 &  0.288 \\
\textsc{l}   &    $-0.511$& 0.447  &    0.253 &  0.283 \\
       class 5& &&&\\
\textsc{n}  &  0.390& 0.369 &     0.291 &  0.314 \\
\textsc{l}   &   0.443 &0.318  &    0.164  & 0.186 \\
       all& &&&\\
\textsc{n}  &  $-0.051$ & 0.204  &    0.802 &  0.800 \\
\textsc{l}   &    $-0.074$& 0.200   &   0.712  & 0.712\\
\hline 
    \end{tabular}
\end{subtable}
\begin{subtable}[h]{0.48\textwidth} 
\caption{physician versus control}
    \begin{tabular}{rrrrr}
    \hline 
   &    est & s.e. &$p_\text{normal}$ & $p_\frt$ \\\hline
       class 1& &&&\\
\textsc{n}  &  0.567 &0.426  &    0.183 &  0.192 \\
\textsc{l}   &   0.588& 0.418  &   0.160  & 0.174 \\
       class 2& &&&\\
\textsc{n}  &   0.193& 0.438 &     0.659 &  0.666 \\
\textsc{l}   &  0.265& 0.409   &   0.517  & 0.523 \\
       class 3& &&&\\
\textsc{n}  &  1.305& 0.494  &    0.008 &  0.012\\
\textsc{l}   &  1.501& 0.462   &   0.001  & 0.003 \\
       class 4& &&&\\
\textsc{n}  &  $-0.273$& 0.413  &    0.508 &  0.515 \\
\textsc{l}   &  $-0.313$& 0.417   &   0.454  & 0.462\\
       class 5& &&&\\
\textsc{n}  & $-0.050$ &0.379  &    0.895 &  0.912 \\
\textsc{l}   &  $-0.067$& 0.279  &    0.811 &  0.816 \\
       all& &&&\\
\textsc{n}  &  0.406 & 0.202   &   0.045 &  0.047\\
\textsc{l}   &    0.463 & 0.190   &   0.015  & 0.017 \\
\hline 
    \end{tabular}
\end{subtable}
\end{table}

Table \ref{tb::pv-values-chong} shows the point estimators, the robust standard errors, the $p$-values based on large-sample approximations of the robust $t$-statistics, and the $p$-values based on {\frt}s. 
In most strata, covariate adjustment decreases the standard errors since the baseline anemia status is predictive of the outcome. 
Two exceptions are the pairwise comparison of the ``soccer'' arm versus the ``control'' arm within class 2 and the pairwise comparison of the ``physician'' arm versus the ``control'' arm within class 4, 
with differences both in the third digit after the decimal point. 
This is likely due to the small group sizes within these strata, leaving the asymptotic approximations dubious. 
The $p$-values from the large-sample approximations and {\frt}s are close with the latter being slightly larger in most cases. 
Based on the theory, the $p$-values based on {\frt}s should be trusted more given their additional guarantee of finite-sample exactness under the strong null hypothesis. This becomes important in this example given the relatively small group  sizes within strata.

\citet{bind2020possible} suggested reporting not only the $p$-values but also the randomization distributions of the test statistics when conducting \frt.  Echoing their recommendation, we show in Figure \ref{fig::randomization-distributions} the histograms of the randomization distributions of the robust $t$-statistics alongside the asymptotic approximations. 
The discrepancy is quite clear in the subgroup analyses yet becomes unnoticeable after averaged over all class levels.
Overall, the $p$-values based on large-sample approximations do not differ substantially from those based on {\frt}s in this application. 
The two approaches yield coherent conclusions: the video with a physician telling the benefits of iron supplements improved the academic performance and the effect was most significant among students in class 3; in contrast, the video with a popular soccer player telling the benefits did not have any significant effect.

\begin{figure}[t]
\centering
\includegraphics[width =0.9\textwidth]{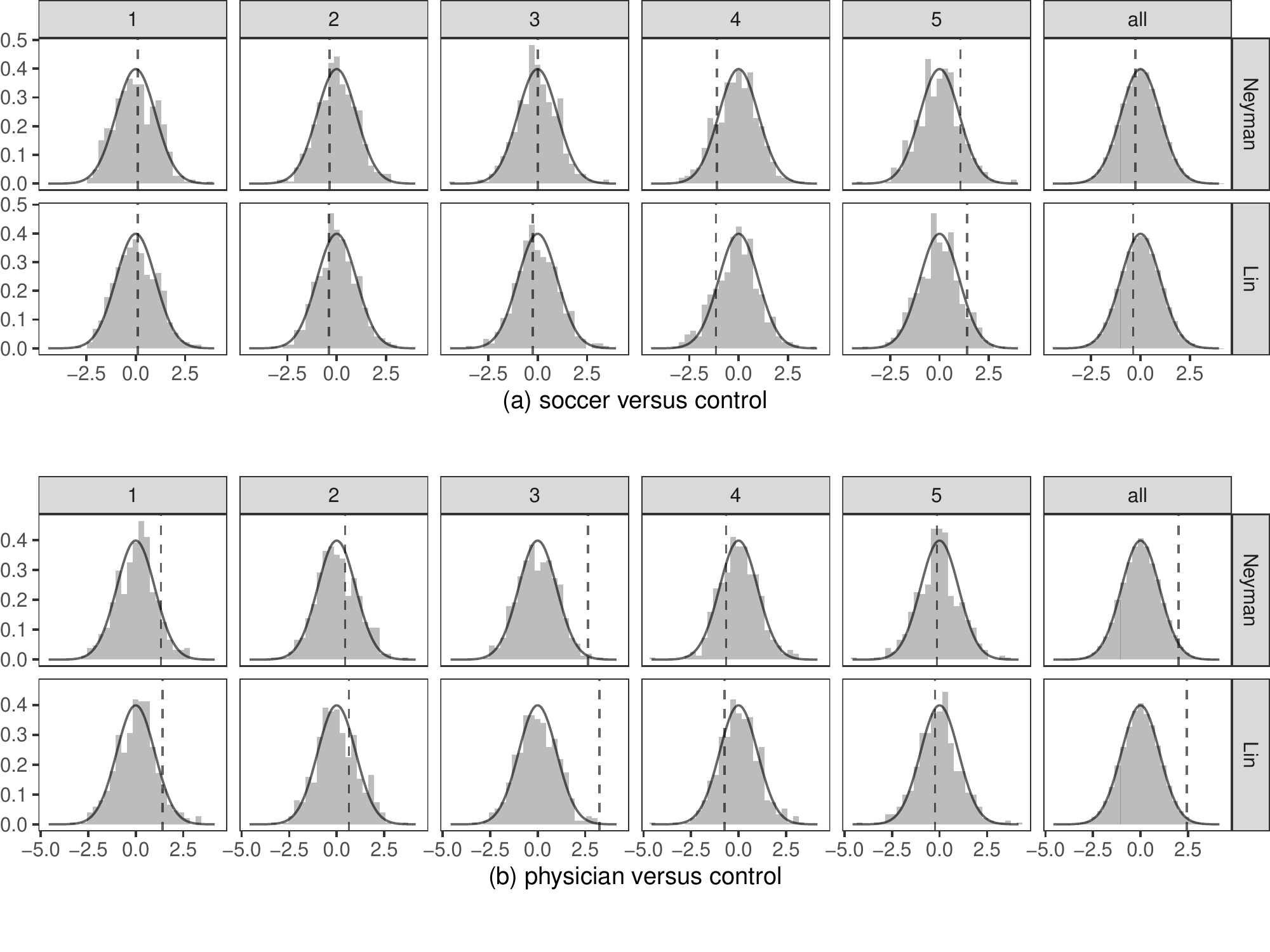}
\caption{Randomization distributions based on $5\times 10^4$ Monte Carlo simulations versus  
$\mathcal{N}(0,1)$.}\label{fig::randomization-distributions}
\end{figure}

\section{Discussion}\label{sec:discussion}
Echoing \citet{Fisher35}, \citet{proschan2019re}, \citet{young2019channeling}, and \citet{bind2020possible}, 
we believe {\frt} should be the default choice for analyzing experimental data given its flexibility to accommodate complex randomization schemes and arbitrary outcome generating processes. 
We established in this paper the theory for covariate adjustment in {\frt} under complete randomization, cluster randomization, stratified randomization, and rerandomization using the Mahalanobis distance, respectively, with final recommendations of the test statistics summarized in Table \ref{fr}. Equipped with the finite-sample exactness under the strong null hypothesis, 
the recommended {\frt}s promise  an additional guarantee under the weak null hypothesis and strictly dominate the counterparts based on large-sample approximations. 
A key point to note is that robust studentization is necessary for the resulting {\frt} to retain asymptotic validity when only the weak null hypothesis holds.
A casual choice of the test statistic is likely to lead to misleading conclusions. 

\begin{table}[t]\caption{\label{fr}Final recommendations for  {\frt} and test statistic $\hts/\tse_*$  in different experiments. 
} 
\begin{tabular}{| l|p{1.6cm}|p{1.6cm} | c| }\hline
 \multicolumn{1}{|c|}{design} & \multicolumn{2}{c|}{ presence of covariates} & other comments\\\cline{2-3}
&\multicolumn{1}{c|}{{\color{white}aaa}no{\color{white}aaa}} &  \multicolumn{1}{c|}{yes} & \\ \hline
complete randomization & \multicolumn{1}{c|}{$*=\textsc{n}$} &   \multicolumn{1}{c|}{$*=\textsc{l}$} & \\\hline
cluster randomization &  \multicolumn{1}{c|}{$*=\textsc{n}$} &   \multicolumn{1}{c|}{$*=\textsc{l}$}  & use cluster total outcomes \\\hline
stratified randomization &  \multicolumn{1}{c|}{$*=\textsc{n}$} &   \multicolumn{1}{c|}{$*=\textsc{l}$}  & weighted average over strata \\\hline
ReM, complete design information &  &   \multicolumn{1}{c|}{$*=\textsc{l}$}  & \\\hline
ReM, incomplete design information &  \multicolumn{1}{c|}{$*=\textsc{n}$} &   \multicolumn{1}{c|}{$*=\textsc{l}$}  & use $p_\frt$ not $p_{\frt,\ma}$\\\hline
\end{tabular}
\end{table}

We conjecture that the strategy of appropriately studentizing an efficient, covariate-adjusted estimator works for {\frt}  in general experiments as well \citep[e.g.,][]{DasFact15, lu2016covariate, mukerjee2018using, middleton2018unified, fogarty2018mitigating, fogarty2018regression}. This strategy works for estimators with normal limiting distributions and may also work for estimators with non-normal limiting distributions as shown in the asymptotic theory of rerandomization. \citet{colin}'s prepivoting approach may work more broadly but we leave the general theory to future research. 


We focused on procedures based on  {\ols}. It is of great interest to extend the theory to high dimensional settings \citep{LassoTE16, lei2018regression}, nonlinear models \citep{zhang2008improving, moore2009covariate, moore2011robust, jiang2019robust, guo2020generalized}, and even estimators based on machine learning algorithms \citep{wager2016high, wu2018loop, farrell2020deep, chen2020efficient}.

If the main parameter of interest is the average treatment effect, the asymptotic theory inevitably involves some moment conditions. 
Without these conditions, the inference becomes challenging \citep{bahadur1956nonexistence}, and {\frt}  may not control type one error rates even asymptotically with heavy-tailed outcomes. An alternative class of {\frt}s use rank statistics to gain robustness with respect to outliers \citep{lehmann1975nonparametrics, CovAdjRosen02}. Although different rank statistics always work under the strong null hypothesis, they in general target parameters other than the average treatment effect \citep[e.g.,][]{rosenbaum1999reduced, Rosenbaum:2003en, chung2016asymptotically}. \citet{chung2016asymptotically} proposed to studentize the Wilcoxon statistic in a permutation test, shedding light on the general theory of {\frt} with rank statistics.


\bibliographystyle{plainnat}
\bibliography{refs_xfrt}


\newpage
\setcounter{equation}{0}
\setcounter{section}{0}
\setcounter{figure}{0}
\setcounter{example}{0}
\setcounter{proposition}{0}
\setcounter{corollary}{0}
\setcounter{theorem}{0}
\setcounter{table}{0}
\setcounter{condition}{0}
\setcounter{lemma}{0}
\setcounter{remark}{0}

\renewcommand {\theproposition} {S\arabic{proposition}}
\renewcommand {\theexample} {S\arabic{example}}
\renewcommand {\thefigure} {S\arabic{figure}}
\renewcommand {\thetable} {S\arabic{table}}
\renewcommand {\theequation} {S\arabic{equation}}
\renewcommand {\thelemma} {S\arabic{lemma}}
\renewcommand {\thesection} {S\arabic{section}}
\renewcommand {\thetheorem} {S\arabic{theorem}}
\renewcommand {\thecorollary} {S\arabic{corollary}}
\renewcommand {\thecondition} {S\arabic{condition}}
\renewcommand {\thepage} {S\arabic{page}}

\setcounter{page}{1}

\begin{center}
\bf \Large 
Supplementary Material  
\end{center}


Section \ref{sec:connections} discusses the extensions to the super-population framework and other permutation tests based on linear models.

Section \ref{sec:af} reviews the notation and some algebraic facts that hold for   arbitrary data generating process. We omit the proofs because they are straightforward. It contains Lemma \ref{lmz} on the univariate \textsc{ols} fit, a basic yet powerful tool in later proofs when coupled with the Frisch--Waugh--Lovell theorems for both the regression coefficients and standard errors \citep{fwl20}. When referring to the Frisch--Waugh--Lovell theorems, we will simply say ``by \textsc{fwl}.''

Section \ref{sec::probability-clt-lln} reviews the central limit theorems under complete randomization and random permutation, and gives a new finite population strong law of large numbers that works under not only simple random sampling and complete randomization but also rejective sampling and ReM \citep{fuller, morgan2012rerandomization}. 

Section \ref{sec:app_fp} gives the proofs of the main results  under complete randomization. 

Section \ref{sec:app_rem} gives the proofs of the results under ReM.

Section \ref{sec:app_sp} gives the proofs of the results related to the extensions to the super-population framework and permutation tests based on linear models in Section \ref{sec:connections}.

\section{Extensions to super-population inference}\label{sec:connections}
\subsection{Overview}
{We extend in this section the theory to the super-population framework and  show the   validity of the proposed procedures when the potential outcomes are independent draws from a super population. Further, the proposed procedures, though model-free in theory, make use of the \textsc{ols} coefficients and $t$-statistics for easy implementation.
It is thus curious to study their connections with existing permutation tests for coefficients in linear models \citep{fl,tb,knd,manly, romano}, as reviewed by \citet{anderson1999empirical}, \citet{ad}, and \citet{lei2019assumption}. 
We evaluate the operating characteristics of these permutation tests for testing the treatment effects,  and demonstrate the superiority of \textsc{frt} by various criteria. 
Among them, the recent proposal by \citet{romano} is the closest to \textsc{frt}  and coincides in procedure with \textsc{frt} based on \citet{Fisher35}'s estimator studentized by the robust standard error.
Recall from Corollary \ref{power} in the main text that this {\frt} does not necessarily improve the power when testing the treatment effects.
A possible improvement is to include adding the treatment-covariates interactions into the linear 
model, such that it coincides with our final recommendation of \textsc{frt} based on \citet{Lin13}'s estimator studentized by the robust standard error. \citet{romano}'s original theory was developed under the linear model assumption for testing whether a coefficient is zero.
Our extension provides an additional justification for it under the potential outcomes framework for testing the treatment effects. 
}

\subsection{Connection with the super-population framework}\label{sp}
In addition to the finite-population perspective that conditions on the potential outcomes for inferences, it is also conventional to view the experimental units as random samples from a super population \citep{tsiatis2008covariate, berk2013covariance, bugni2018inference, negi2020revisiting, ye2020principles}. 
We now extend the discussion to this framework and examine the operating characteristics of the proposed strategies for conducting {\frt} on random potential outcomes. 
Assume $\{Y_i(0), Y_i(1), Z_i, x_i\}_{i=1}^N$ are independent and identically distributed (\textsc{iid}) samples from a super population. 
With a slight abuse of notation, let $\tau = E\{ Y_i(1) - Y_i(0)\}$ be the population average treatment effect throughout Section \ref{sp}. 
The goal is to test 
$$H_{0}: \tau = 0$$
as the analog of $\HN$ in the finite-population setting.

Recall $p_z = N_z/N$ as the treatment proportions for $z = 0,1$. 
Without introducing new notation, let $p_z$ also denote $P(Z_i = z)$, equaling the probability limit of $N_z/N$ as $N$ goes to infinity. 
Let 
$$
\mu_z = E\{Y_i(z)\},\quad 
\sigma^2_z = \var\{Y_i(z)\},\quad 
\mu_x = E(x_i),\quad 
\sigma_x^2 =\cov(x_i),\quad
\sigma_{xY(z)} = \cov\{x_i, Y_i(z)\}
$$ 
be the first two moments of the potential outcomes and covariates.
We impose the following conditions under the super-population framework.

\begin{condition}\label{asym_SP}
$\{Y_i(0), Y_i(1), Z_i, x_i\}_{i=1}^N$ are \textsc{iid} draws from the population with (i)  $E(\|x_i\|_4^4) \leq \infty$, $E\{Y^4_i(z)\} \leq \infty$, $E\{\|x_iY_i(z)\|_4^4\} \leq \infty$, $\sigma_z^2 - \sigma_{xY(z)}^\T(\sigma_x^2)^{-1} \sigma_{xY(z)} >0$ for $z=0,1$, and (ii) $Z_i \ \ind \  \{ Y_i(0), Y_i(1), x_i\}$. 
\end{condition}

With a slight abuse of notation,  
let 
$\gamma_z =(\sigma_x^2)^{-1}\sigma_{xY(z)}$  
 be the coefficient of $x_i$ in the  population {\olsf} of $Y_i(z)$ on $(1, x_i)$, and let  $a_i(z) =  Y_i(z) - \mu_z - (x_i - \mu_x)^\T (p_1\gamma_1+p_0\gamma_0)$ 
 and 
 $b_i(z) = Y_i(z) - \mu_z - (x_i - \mu_x)^\T\gamma_z$ be the residuals. 
Let $\sigma^2_{a(z)}$ and $\sigma_{b(z)}^2$ be the variances of $a_i(z)$ and $b_i(z)$, respectively. 
\citet{negi2020revisiting} reviewed the asymptotic distributions of the estimators as
$\rtn  (\hts  -\tau)  \rightsquigarrow  \mN(0, v_*)\ (* = \neyman, \rosenbaum, \fisher, \lin)$  with
\begin{eqnarray}\label{eq::super-var}
v_\neyman  =  p_1^{-1} \sigma^2_1 + p_0^{-1} \sigma^2_0,  \quad  \quad  
  v_\rosenbaum = v_\fisher  =  p_1^{-1} \sigma^2_{a(1)}  + p_0^{-1} \sigma^2_{a(0)}, \quad \quad   
   v_\lin  =  p_1^{-1} \sigma_{b(1)}^2  + p_0^{-1} \sigma_{b(0)}^2 + \Delta_{\bar{x}}, 
\end{eqnarray}
where $\Delta_{\bar{x}} = (\gamma_1-\gamma_0)^\T\sigma_x^2(\gamma_1-\gamma_0)$.
Technically, \citet{negi2020revisiting} did not discuss $\hat{\tau}_\rosenbaum$, but we can show that $\hat{\tau}_\textsc{r}$ has the same asymptotic distribution as $\hat{\tau}_\textsc{f}$. So we unify the results above. 
A key distinction from the finite-population analysis is that $\htl$ now has extra variability due to the centering of the covariates. In contrast, other estimators do not have this extra term $\Delta_{\bar{x}} $ because they remain unchanged regardless of whether we center the covariates or not.

The extra term due to centering goes away if we condition on all covariates. 
But if we stick to the \textsc{iid} assumption in Condition \ref{asym_SP}, we must modify the  standard errors for \citet{Lin13}'s estimator to account for $\Delta_{\bar{x}}$ \citep{berk2013covariance, negi2020revisiting}. 
In particular, define $\ese_{\lin}^{2}$ and $\tse_{\lin}^{2}$  as the classic and robust standard errors squared plus $ \hat\Delta_{\bar{x}} /N = \hat\theta^\T S_x^2 \hat\theta/ N$, respectively, where $S_x^2$ is the sample covariance matrix of the $x_i$'s and  $\hat\theta =  \hat{\gamma}_{\lin, 1} - \hat{\gamma}_{\lin, 0} $ is the coefficient of $Z_i(x_i - \bar{x})$ in the {\olsf} of $Y_i$ on $\{1,Z_i,x_i-\bar{x}, Z_i(x_i - \bar{x})\}$.
We then use them to construct the studentized statistics 
and obtain in total twelve test statistics as in Table \ref{tb:tt}. Of interest is whether the resulting tests preserve the correct type one error rates under the super-population framework for testing $H_0$.

Theorem \ref{lin_SP} gives the asymptotic sampling and randomization distributions of the twelve test statistics in Table \ref{tb:tt}.
The same reasoning that underpins Corollaries \ref{rec} and \ref{power} ensures the four robust $t$-statistics are the only options proper for testing $H_0$, among which $\htl/\tse_\lin$ enjoys the highest power. 
Below, $\assp$ indicates a statement that holds for almost all sequences of $\{Y_i(0), Y_i(1), Z_i, x_i\}_{i=1}^N$.

\begin{theorem}\label{lin_SP}
Assume Condition \ref{asym_SP}. 
\begin{enumerate}
[(a)]
\item
$\rtn  (\hts  -\tau)  \rightsquigarrow  \mN(0, v_*)$,  
and
$\rtn  \htsp \rightsquigarrow  \mN(0, v_{*0}) \   \assp$,  
with $v_* \ (* = \neyman, \rosenbaum, \fisher, \lin)$  defined in \eqref{eq::super-var},  
$
v_{\neyman0}  = p_0^{-1}\sigma^2_1 + p_1^{-1}\sigma^2_0 + \tau^2
$, 
and
$
 v_{\rosenbaum0} =v_{\fisher0}  = v_{\lin0}  =  p_0^{-1}  \sigma_{a(1)}^2  + p_1^{-1}  \sigma_{a(0)}^2 + \tau^2.
$

\item
$(\hts -\tau)/\ese_*   \rightsquigarrow  \mN(0, c_*) $, and
$
(\hts  /\ese_* )^\pi   \rightsquigarrow  \mN(0, 1) \ \assp$, 
with 
$
c'_* = v_* / (v_{*0}-\tau^2)
$
for  $* = \neyman, \rosenbaum, \fisher$, and
$
c'_\lin  =    v_\lin / \{ p_0^{-1}  \sigma_{b(1)}^2  + p_1^{-1}  \sigma_{b(0)}^2 + \Delta_{\bar{x}} \} . 
$
 
\item
$(\hts-\tau)  /\tse_*    \rightsquigarrow  \mN(0, 1)$, 
and $(\hts  /\tse_*)^\pi   \rightsquigarrow  \mN(0, 1) \  \assp$ for $* = \neyman, \rosenbaum, \fisher, \lin$.
\end{enumerate}
\end{theorem}

The sampling distributions, except for those indexed by \rosenbaum, are known, but the permutation distributions, especially those with corrections for  $\Delta_{\bar{x}} $, are new.
Intuitively, $\sigma^2_z$ mirrors the finite-population limit of $S_z^2$ in Condition \ref{asym} and affords the probability limit of $S_z^2$ under the super-population framework. The asymptotic variances of $\hts\ (* = \neyman, \rosenbaum, \fisher)$ in Theorem \ref{lin_SP} thus mirror their finite-population analogs in Theorem \ref{neyman}--\ref{fisher} without the conservativeness issue. In contrast, the  asymptotic variance of $\htl$  features the additional term $\Delta_{\bar{x}} $. 
The randomization distribution of $\hat\Delta_{\bar{x}}$ 
is of higher order such that the randomization distributions $(\hat{\tau}_\lin  /\ese_\lin)^\pi$ and $(\hat{\tau}_\lin  /\tse_\lin)^\pi$ remain at $\mathcal{N}(0,1)$ even after this correction in the standard error.

\begin{corollary}
Assume Condition \ref{asym_SP}. 
The robustly studentized $\htH/\tse_* \ph$
are the only test statistics in Table \ref{tb:tt} proper for testing $H_0$, with the limiting values of $\tse_\lin^2/\tse_*^2$ less than or equal to 1  for $*=\neyman, \rosenbaum, \fisher$ $\assp$.
\end{corollary}

\citet{agl} proposed an improvement on $\tse_\neyman$ under the finite population framework, which can be extended to $\tse_*\ (* = \rosenbaum, \fisher, \lin)$ as well. These alternatives, however, underestimate the asymptotic variances of $\hat{\tau}_*\ (*=\neyman, \rosenbaum, \fisher, \lin)$ under the super-population framework. We thus do not pursue them in our recommendation. 

 Much recent work has been done characterizing the asymptotic behaviors of more sophisticated assignment mechanisms \citep[e.g.,][]{bai2020, bugni2019, ye2020inference, zhang2020}).
We leave the general theory to future work.

%
%
%
%
%
%
\subsection{Connection with permutation tests for coefficients in linear models}\label{sec::pt}
%

We unified the twelve test statistics in Table \ref{tb:tt} as outputs from \textsc{ols} fits and permuted $Z$ to compute the $p$-values via {\frt}. 
The procedures are similar in form to permutation tests based on linear models. 
Consider a linear model
\beginy\label{lmf}
Y  = \on \alpha + Z \beta  +   X\gamma+ \epsilon 
\endy
that characterizes  the treatment effect by the coefficient of $Z$. 
Testing the treatment effect reduces to testing $\beta = 0$. 
It is standard to use the  $t$-test based on the classic or robust standard error. 
Permutation tests, on the other hand, afford compelling alternatives due to their accuracy in finite samples \citep{anderson1999empirical, ad}. 
We review here five existing permutation tests for testing $\beta=0$ under model \eqref{lmf}  and show their inferiority to \textsc{frt} with $\hat{\tau}_\textsc{l} /\tse_\textsc{l} $ for testing the treatment effects. 

\subsubsection{A review of existing permutation tests for linear models}
\label{subsubsec::review-of-permutations}

Without introducing new symbols, let $ \htf $ be the coefficient  of $ Z$ from the \textsc{ols} fit of model \eqref{lmf}, with $\ese_\fisher$ and $\tse_\fisher$ as the classic and robust standard errors.  \cite{draper} used the square of the multiple correlation coefficient
as the test statistic, and constructed the reference distribution by permuting $Z$. 
A recent proposal by \citet{romano} used the $\htf/\tse_\fisher$ as the test statistic and constructed the reference distribution by permuting $Z$. It coincides with {\frt} based on $\htf/\tse_\fisher$ in procedure and delivers asymptotically robust randomization test for $\beta=0$ under the linear model framework.
Our theory gives another justification of it for testing both $\HF$ and $\HN$ under the potential outcomes framework.
Recall from Corollary \ref{power} that {\frt} with $\htf/\tse_\fisher$ does not necessarily improve the power when testing the treatment effects. A possible improvement is thus to add treatment-covariates interactions to model \eqref{lmf} such that the procedure coincides with the more powerful test based on $\htl/\tse_\lin$. Depending on whether we condition on the covariates or treat them as \textsc{iid} draws from a super population,  we might need to modify the robust standard error as discussed in Section \ref{sp}.

 The rest four permutation tests, on the other hand, all used the classic $\htf/\ese_\fisher$ as the test statistic yet employed distinct permutation schemes to generate the reference distributions. We review below their respective procedures to highlight the difference, and extend them to a unified theory with test statistics being the coefficients and the classic and robust $t$-statistics, respectively. For simplicity of presentation, we only give the $p$-value formulas based on their original proposals with $\htf/\ese_\fisher$ as the test statistic. Those based on $\htf/\tse_\fisher$ and $\htf$ can be derived similarly by replacing all occurrences of the classic standard errors with their robust counterparts or constant $1$.

\paragraph{\citet{fl}}
Let $e  $ be the residual vector from the \textsc{ols} fit of $Y$ on $(1_N, X)$ as the reduced model, with the fitted vector $Y-e.$
\citet{fl} proposed to permute $e$ and obtain the $p$-value as follows: 
\begin{enumerate}[\fl-1:]
\item\label{fl_1} Permute $e = (e_1, \dots, e_N)^\T$ to obtain $e_\pi  $, construct $Y^\pi = Y- e + e_\pi$ as the synthetic outcomes, and  compute $(\hbflp,\seflp, \tseflp)$ as the coefficient of $Z$ and its classic and robust standard errors from the \textsc{ols} fit of $Y^\pi$ on $(\on, Z, X)$. 
\item\label{fl_2} Compute  
$
p_\fl= |\Pi|^{-1} \sum_{\pi\in\Pi} 1 ( |\hbflp| / \seflp  \geq |\htf| / \hat{\text{se}}_\fisher).
$
\end{enumerate}

The pair $(Z_i, x_i)$ remains intact under this procedure yet gets reshuffled under \textsc{frt}.
This exemplifies the difference   between $p_\frt$ and $p_\fl$.   \citet{fl} first proposed the method without the intention of making it a formal test, but rather ``an alternative interpretation of a reported significance level." 
It now becomes a standard permutation test for coefficients in linear models \citep{anderson1999empirical, ad}.     

\paragraph{\citet{knd}}
\citet{knd} also proposed to permute $e$ with a slight modification of \citet{fl}. 
Let $\delta  $ be the residual vector from the {\olsf} of $Z$ on $(\on, X)$. 
The procedure proceeds as follows:
\begin{enumerate} [\knd-1:]
\item   Permute $e  $ to obtain $e_\pi  $, and compute $(\hbkp,\sekp, \tsekp)$ as the coefficient of $\delta$ and its classic and robust standard errors from the \textsc{ols} fit of $e_\pi$ on $(\on, \delta)$.
\item
Compute  
$
p_\knd = |\Pi|^{-1} \sum_{\pi\in\Pi}   1(   |\hbkp| / \sekp  \geq |\htf| /  \hat{\text{se}}_\fisher ).
$
\end{enumerate}

\citet{knd} and \cite{ad} pointed out that $\hbkp=\hbflp$ due to the Frisch--Waugh--Lovell theorem whereas $\sekp\neq\seflp$ and $\tsekp\neq\tseflp$.

\newcommand{\epfpi}{\ep_{\fisher,\pi(i)}}
\paragraph{\citet{tb}}
\citet{tb} proposed to permute the residuals  $\epf  $  from the {\olsf} of model \eqref{lmf}, with $Y - \epf$ being the fitted vector.
The procedure proceeds as follows:
\begin{enumerate}[\tb-1]
\item\label{tb_1} Permute $\epf $ to obtain $\epfp  $, construct $Y^\pi = Y- \epf +\epfp$ as the synthetic outcomes, and  compute $(\hbtbp, \setbp, \tsetbp)$ as the coefficient of $Z$ and its classic and robust standard errors from the \textsc{ols} fit of $Y^\pi$ on $(\on, Z, X)$. 
\item\label{tb_2}
Compute 
$
p_\tb = |\Pi|^{-1} \sum_{\pi\in\Pi}  1  ( |\hbtbp-\htf| / \setbp \geq |\htf| / \ese_\fisher ) .
$
\end{enumerate}

Compare \tb-\ref{tb_1} with \fl-\ref{fl_1}  to see the difference in the models used for constructing the synthetic outcomes, with
\citet{fl} using the reduced model whereas  \citet{tb} using the full model. Consequently, \tb-\ref{tb_2} differs from  \fl-\ref{fl_2} in that $\hbtbp$ must be centered by $\htf$ under \citet{tb}'s procedure.

\paragraph{\citet{manly}}
\citet{manly}  proposed to permute $Y$ as follows:

\begin{enumerate}[\mly-1]
\item\label{m_1} Permute $Y$ to obtain $Y_\pi$; compute $(\hbmp, \semp, \tsemp)$ as the coefficient of $Z$ and its classic and standard errors from the \textsc{ols} fit of $Y_{\pi}$ on $(\on, Z, X)$. 
\item
Compute 
$
p_\mly = |\Pi|^{-1} \sum_{\pi\in\Pi}  1 (   |\hbmp| / \semp \geq |\htf| / \ese_\fisher ) .
$
\end{enumerate}

\subsubsection{Finite-sample exactness for testing $H_{0\textsc{F}}$ and properness for testing $H_{0\textsc{N}}$}\label{sec::cc}

The procedures in Section \ref{subsubsec::review-of-permutations}  employed distinct permutation schemes to generate the reference distributions.  
  \citet{romano} permuted $Z$, \citet{fl} and \citet{knd} permuted $e$, \citet{tb} permuted $\epf$, and \citet{manly} permuted   $Y$.  
Table \ref{tb:pt}   summarizes them.

\begin{table}[t]\caption{\label{tb:pt}Five permutation tests for testing $\beta=0$ in model \eqref{lmf}. 
}
\begin{tabular}{| l |c|c|c|c|c|}
\hline
procedure &  $T$ & model for computing $T$  & $T^\pi$ \\\hline
\citet{romano}
&$\htf/\tse_\fisher$  &$ Y \sim \on  + Z_\pi +X$&$   \htfp/\tse_\fisher^\pi $\\
\citet{fl}& $\htf/\ese_\fisher$
&  $Y - e + e_\pi  \sim \on  + Z +X$ & $  \hbflp /\seflp $\\
\citet{knd}
&same as above& $ e_\pi  \sim \on  + \delta$ & $  \hbkp/\sekp $ \\
\citet{tb}
&same as above& $ Y - \epf + \epfp  \sim \on  + Z+X$ & $(\hbtbp-\htf)/\setbp$\\
\citet{manly}
&same as above&
 $ Y_\pi \sim \on  + Z+X$ &$  \hbmp/\semp $\\ 
\hline
\end{tabular}
\end{table}


An interesting question is how they compare with {\frt} and each other when applied to testing the treatment effects. Whereas the operating characteristics of \citet{romano} follow directly from Theorem \ref{fisher} due to its identicalness to {\frt}, those of the rest four tests are less straightforward and hinge on their respective reference distributions resulting from the distinct permutation schemes.

With a slight abuse of notation, we use the random variables under the reference distributions to represent the respective procedures for the rest of this section. 
The finite-sample exactness for testing $\HF$, on the one hand, requires exact match of the reference distribution with the sampling distribution under $\HF$. 
The difference in permutation schemes leaves the distributions of 
the unstudentized $\hbflp$, $\hbkp$, $\hbtbp-\htau_\fisher$, $\hbmp$ and their studentized variants distinct from those of $\htf$, $\htf/\ese_\fisher$,  and $\htf/\tsef$ under $\HF$,
such that none of them is finite-sample exact for testing $\HF$.
Their properness and relative power for testing $\HN$, on the other hand, depend on the  asymptotic behaviors of the reference distributions, which are summarized below. 
 

\begin{theorem}\label{nay_fl}
Assume complete randomization and Condition \ref{asym}. We have
\begin{enumerate}[(a)]
\item 
$\rtn \hbsp \rightsquigarrow \mN(0, v_{\fisher0})$ for $* = \fl, \knd; \quad
\rtn (\hbtbp  - \htf)\rightsquigarrow \mN(0, v_{\fisher0} -\tau^2); \quad\rtn \hbmp \rightsquigarrow \mN(0, v_{\neyman0})$; 
\item 
 $\hbssp \rightsquigarrow \mN(0,1)$ and $ \hbssrp {\rightsquigarrow}\mN(0, 1)$  for $* = \fl, \knd,\mly$;\\ $(\hbtbp  - \htf)/\setbp \rightsquigarrow \mN(0, 1)$ and $ (\hbtbp  - \htf)/\tsetbp \rightsquigarrow \mN(0, 1)$
\end{enumerate}
hold $P_Z$-a.s., recalling $v_{\neyman0}$ and $v_{\fisher0}$ as the asymptotic variances of $\rtn  \htnp$ and $\rtn  \htfp$ defined in Theorems \ref{neyman} and \ref{fisher}, respectively. 
\end{theorem}

The asymptotic distribution of $\hbflp/\seflp$ first appeared in \citet{fl}. \citet{ad} generalized it and gave a sketch of the proof for all four classic $t$-statistics in Table \ref{tb:pt}. We flesh out their proofs and furnish the new results on the unstudentized  coefficients and their robustly studentized variants.

Intuitively, the asymptotic variances of $\hbflp = \hbkp$, $\hbtbp  - \htf$, and $\hbmp$ are proportional to the finite-population variances of $e$, $\epf$, and $Y$, respectively,  with increasing variability in the order of $\epf$, $e$, and $Y$. In particular,  the {\olsf} of the reduced model adjusts for $X$ such that $e$ is less variable than $Y$;  the {\olsf} of the full model \eqref{lmf} further adjusts for $Z$ such that  $\epf$ is less variable than $e$. The four unstudentized test statistics elucidate the impact of different permutation schemes on the resulting reference distributions, with $\hbflp$ and $\hbkp$ sharing the same limiting distribution with $\htfp$ and $\htrp$ irrespective of the true value of $\tau$, whereas $\hbmp$ sharing that with $\htnp$.    The limiting distribution of $\hbtbp-\htf$, on the other hand, coincides with that of $\htfp$ or $\htrp$ if and only if $\HN$ is true.


Despite the distinction between $\hbsp$ for $* = \fl, \knd, \tb, \mly$, all eight studentized statistics  are asymptotically standard normal and thus coincide with the limiting distributions of $(\htf/\ese_\fisher)^\pi$ and $(\htf/\tsef)^\pi$   under {\frt}.
Juxtapose Theorem \ref{nay_fl} with Theorem \ref{fisher} to see the four robustly studentized variants as the only options proper for testing $\HN$. We state the result in Corollary \ref{coro::permutation-ols-student}. 

\begin{corollary}\label{coro::permutation-ols-student}
{\asym}   
The robustly studentized $\hbflp/\tseflp$, $\hbkp/\tsekp$,  $(\hbtbp  - \htf)/\tsetbp$, and $\hbmp/\tsemp$ are proper for testing $\HN$ whereas the unstudentized and the classically studentized alternatives are not. 
\end{corollary}

Despite the positive result in Corollary \ref{coro::permutation-ols-student},  we do not recommend these permutation tests  for testing the treatment effects for several reasons. 
First,  they are not finite-sample exact for $\HF$ in the first place, such that properness under $\HN$ can be a rather weak requirement. 
For instance, \citet{manly}'s permutation test with the unstudentized statistic yields a reference distribution far from the true distribution of $\hat{\tau}_\fisher$ although the robust studentization repairs it asymptotically.
Second, even within the scope of testing $\HN$, the test statistic $\hat{\tau}_{\fisher} / \tse_\fisher$ is suboptimal compared to $\hat{\tau}_{\lin} / \tse_\lin$ and may deliver even less power than  $\htnsr$ despite the extra use of covariates. 
Further,  Corollary \ref{coro::permutation-ols-student} holds under complete randomization but may not extend to general designs even asymptotically. Overall, \textsc{\frt} is strictly superior for testing the treatment effects.  

\section{Notation and algebraic facts}\label{sec:af}

Let $1_m$ and $0_m$ be the $m\times 1$ vectors of ones and zeros, respectively,  and let $I_m$ be the $m\times m$ identity matrix. 
We suppress the subscript $m$ when the dimensions are clear from the context. 

For $(u_i, v_i)_{i=1}^N$, where $u_i$ and $v_i$ are $m \times 1$ and $k\times 1$ vectors, respectively, let 
$$
\bar u  = N^{-1}\sumN u_i,\quad
S_u^2 =  (N-1)^{-1}\sumN(u_i - \bar u) (u_i - \bar u)^\T,\quad
S_{uv}=  (N-1)^{-1}\sumN(u_i - \bar u) (v_i - \bar v)^\T
$$ 
be the finite-population mean and covariance matrices of $(u_i)_{i=1}^N$ with itself and $ (v_i) _{i=1}^N$, respectively. 
The matrices $S_u^2$ and $S_{uv}$ degenerate to the finite-population variance and covariance for $m=k = 1$.
We suppress the ``finite-population" when no confusion would arise.
Let $\lambda    = (N-1)/N$ be the scaling factor to accommodate the difference in normalizing by $N$ or $N-1$.  

\subsection{Potential outcomes and covariates}

Assume standardized covariates with mean $\bar x = 0_J$ and covariance matrix $\sx = I_J$  
throughout the finite-population analysis  to simplify the presentation. 
Let $X_j = (x_{1j}, \dots, x_{Nj})^\T$ be the $j$th column of $X$ with mean $N^{-1}\sumi x_{ij} =0$ and variance $(N-1)^{-1} \sumi x_{ij}^2 = 1$. All procedures discussed in this paper are invariant to non-degenerate transformations of the covariates and thus unaffected by the standardization.

Recall $\gamma_z = (S_x^2)^{-1}S_{xY(z)} = (N-1)^{-1} \sumi x_i Y_i(z)$ as the coefficient of $x_i$ in the \textsc{ols} fit of $Y_i(z)$ on $(1,x_i)$.
Recall  $S^2_z$, $ S_{a(z)}^2   $, $S_{b(z)}^2$, $S_\tau^2$, and $S_\xi^2$ as the finite-population variances of $Y_i(z)$, $a_i(z) =  Y_i(z) - \bar Y(z) -  x_i^\T (p_1\gamma_1+p_0\gamma_0)$, $b_i(z) = Y_i(z)  - \bar Y(z)  - x_i^\T\gamma_z$, $\tau_i = Y_i(1)-Y_i(0)$, and $\xi_i= b_i(1)-b_i(0)$, respectively. 
Let 
\begina
\gamma =  p_1  \gamma_1+ p_0  \gamma_0,\quad
S^2 = p_1 S_1^2 +p_0 S_0^2,\quad 
 S_a^2  = p_1 S_{a(1)}^2 + p_0 S_{a(0)}^2,\quad 
 S_b^2  = p_1 S_{b(1)}^2  + p_0 S_{b(0)}^2.
\enda 
We have $a_i(z) = Y_i(z)  - \bar Y(z)  - x_i^\T \gamma$, 
\beginy\label{Sab} 
\wa{1.4}
 S_{a(z)}^2  =  S_z^2  + \|\gamma\|^2_2 - 2\gamma_z^\T \gamma,\quad  S_a^2  = S^2 -\|\gamma\|^2_2,\quad S_{b(z)}^2  =  S_z^2  - \|\gamma_z\|^2_2 ,&S_\xi^2  = S_\tau^2  - \|\gamma_1-\gamma_0\|^2_2,\quad\quad
\endy
with  $ S_{a(z)}^2  -  S_{b(z)}^2  = \|\gamma_z-\gamma\|^2_2 \geq 0$ for $z = 0,1$. 

\subsection{Observed outcomes and \textsc{OLS} fits}

Let $\hY(z) = N_z^{-1}\sum_{i: Z_i = z}  Y_i $ and $\hat{S}_z^2 = (N_z-1)^{-1}\sum_{i:Z_i=z} \{Y_i - \hY(z) \}^2$  be the mean and variance of $\{Y_i: Z_i = z\}$, respectively.
Let $\htx= \hat x(1) - \hat x(0)$ be the difference in means of the covariates under treatment and control, where $\hat x(z) = N_z^{-1} \sum_{i: Z_i = z}x_i$.
Centered covariates ensure
\begin{equation}
\label{eq::simple-mean-zero-x}
\hat x   (1) = p_0 \htx,\quad\hat x   (0) = - p_1\htx . 
\end{equation}
Let $\hat{S}_{x(z)}^2= (N_z-1)^{-1} \sum_{i: Z_i = z} \{x_i - \hat x (z)\}\{x_i - \hat x (z)\}^\T$ 
and $\hat{S}_{xY(z)} = (N_z-1)^{-1} \sum_{i: Z_i = z} \{x_i - \hat x(z)\}\{Y_i - \hat Y(z)\} $
be the  covariance matrices of $\{x_i: Z_i = z\}$ with itself and $\{Y_i: Z_i = z\}$.  

Let $\hY = N^{-1}\sumN  Y_i $ and 
$\hat S^2 = (N-1)^{-1}\sumN (  Y_i  - \hY)^2$ be the mean and variance of $(Y_i)_{i=1}^N$, respectively. 
They satisfy 
\beginy\label{as_numeric}
\hY = p_1 \hY(1) + p_0 \hY(0), \quad
\hs^2 = \frac{N_1-1}{N-1} \hs^2_1 + \frac{N_0-1}{N-1} \hs^2_0 +  \frac{N}{N-1} p_0p_1 \htn^2. 
\endy
Let $\qxy = (N-1)^{-1}\sumi x_iY_i$ be the  covariance matrix of $\{(x_i, Y_i)\}_{i=1}^N$. 
The {\olsf} of $Y$ on $(\on,X)$ has coefficients $( \hat Y, \hgr )$ and residuals $e = (e_1, \dots, e_N)^\T$ that satisfy
\beginy\label{algebra_r}
&&\hgr = (N-1)^{-1}\sumi x_i Y_i = \hat S_{xY}, \quad \quad   e_i =Y_i - \hat Y - x_i^\T \hgr \quad (i = \ot{N}),\nonumber\\
&&\hat e = N^{-1}\sumN  e_i = 0, \quad \quad
\hat  S_e^2  = (N-1)^{-1}\sumN(\ei - \hat e)^2
= (N-1)^{-1}\esq 
=  \hat S^2 - \|\hgr\|^2_2.\quad 
\endy
Let $\he(z)= N_z^{-1}\sum_{i:Z_i=z} \e_i$ and $\hat  S_{e(z)}^2= (N_z-1)^{-1}\sum_{i:Z_i=z} \{\e_i - \he(z)\}^2$ be the sample mean and variance of $\{e_i: Z_i = z\}$ for units in treatment group $z$. 

Let $\ep_* = (\ep_{*,1}, \dots, \ep_{*,N})^\T$ be the residuals from the \textsc{ols} fit that generates $\hts$, where $\phs$.
Let 
$\hat\ep_*(z) = N_z^{-1}\sum_{i: Z_i = z}  \ep_{*,i} $ and $\hat S^2_{*(z)} = (N_z-1)^{-1} \sum_{i:Z_i = z}\{\ep_{*,i} - \hat\ep_*(z)\}^2$  be the sample mean and variance of $\{\ep_{*,i}: Z_i = z\}$ under treatment $z$. 
They satisfy 
\beginy\label{syse}
&&\hsnz = \hat S^2_z, \quad \hsrz = \hat  S_{e(z)}^2,  \quad \hssz   
=
\hat  S_z^2  +
\hg_*^\T \hat S_{x(z)}^2 \hg_* - 2\hg_*^\T\hat S_{xY(z)}\quad\text{for $*=\rosenbaum,\fisher$},\nonumber\\
&&\hslz 
=
\hat  S_z^2  +
\hg_{\lin,z}^\T \hat S_{x(z)}^2 \hg_{\lin,z} - 2\hg_{\lin,z}^\T\hat S_{xY(z)},\nonumber\\
&&\ese_*^2 = \frac{N(N_1-1)}{(N-2)N_1N_0} \hat  S_{*(1)}^2 +  \frac{N(N_0-1)}{(N-2)N_1N_0}  \hat  S_{*(0)}^2 \qquad \text{for $* = \neyman, \rosenbaum$},\nonumber\\
&&
\tse_*^2 =  \frac{N_1-1}{N_1^2} \hat  S_{*(1)}^2  + \frac{N_0-1}{N_0^2}\hat  S_{*(0)}^2 \qquad \text{for $* = \neyman, \rosenbaum$}.
\endy
The numerical expressions of $\ese_*^2$ and $\tse_*^2$ follow from standard least squares theory and \citet[][Chapter 8]{AngristEcon}. 
Analogous results for $* = \fisher, \lin$ are given in Lemma \ref{alg_f} and \cite{LD20}.

 \subsection{Projection matrices}
Let $H$ and $H_1=N^{-1}\on\ont$ be the projection matrices onto the column spaces of $(\on, X)$ and $\on$.
Let $\delta = (I-H)Z = (\delta_1, \dots, \delta_N)^\T$ be the residuals from the {\olsf} of $Z$ on $(\on,X)$. They satisfy 
\beginy
&&e=(I-H)Y, \quad He = 0_N, \quad H_1 e = 0_N,  \quad  H\on  = \on, \quad HX = X, \nonumber \\
&&H = H_1 + X (X^\T X)^{-1} X^\T = N^{-1}\on\ont + (N-1)^{-1}X X^\T,\nonumber\\
&&\delta_i = Z_i - p_1 - \lambda^{-1}p_1p_0 x_i^\T\htx, \quad   \hat \delta = N^{-1}\sumi \delta_i  = 0, \quad H\delta = 0_N,\quad  H_1\delta = 0_N, \nonumber\\
&&\|\delta\|_2^2 =  Z^\T(I-H)Z =N (p_1p_0 -   \lambda^{-1}p_1^2 p_0^2 \htx^\T \htx). \label{H}
\endy

\subsection{A lemma on the univariate \textsc{OLS}}

\begin{lemma}\label{lmz}
Let $u = (u_1, \dots, u_N)^\T$ and $v=(v_1, \dots, v_N)^\T$ be two $N\times 1$ vectors, and let $\htau_0$ be the coefficient of $v$ from the {\olsf} of $u$ on $v$,
with the residual vector $\eta = u -  v\htau_0 $. Let $\ese_0$ and  $\tse_0$ be the classic and robust standard errors, respectively. 
We have
\begina
\htau_0  = \frac{v^\T u}{\vsq}, \quad \ 
\ese_0^2 = \frac{1  }{N-1}  \frac{ \|\eta\|^2_2 }{\vsq} = 
 \frac{1  }{N-1}\left(\frac{\usq}{\vsq} - \htau_0^2\right), \quad\ 
\tse_0^2  = \frac{  v^\T \diag(\eta_i^2) v} { (\vsq)^2 } 
=\frac{\eta^\T \diag(v_i^2) \eta}{ (\vsq)^2} .
\enda
\end{lemma}

\section{Probability measures and basic limiting theorems}
\label{sec::probability-clt-lln}

\subsection{Probability measures}
Recall $T^\pi = T(Z_\pi, Y(Z), X)$, where $\pi \sim \Unif(\Pi)$ and $Z_\pi \sim \Unif(\mz)$, as a random variable following the randomization distribution of $T$ conditioning on $Z$.  
Let $\sesp $ and $\tsesp $ be the classic and robust standard errors of $\htsp$ under $Z_\pi$  for $\phs$. 
By definition, $\htsp$, $
(\hat\tau_*/\ese_*)^\pi = \htsp /\sesp$ and $ (\hat\tau_*/\tse_*)^\pi = \htsp /\tsesp$ are the outputs from the {\olsf}s of 
$Y$ on $(\on, Z_\pi)$ for $*=\neyman$; 
 $e$ on $(\on, Z_\pi)$  for $*=\rosenbaum$;  
 $Y$ on $(\on, Z_\pi, X)$  for $*=\fisher$; 
 and
$Y$ on $(\on, Z_\pi, X, W^\pi)$  for $*=\lin$, respectively, where $W^\pi = (W_1^\pi,\dots, W_N^\pi)^\T $ with $W_i^\pi = Z_{\pi(i)} x_i$.
Let 
$(\hat S_{*(z)}^2)^\pi$ be the analogs of 
$\hat S_{*(z)}^2$ based on the corresponding residuals for $\phs$. 

Index by  $Z$ and \textsc{s}  the probability measures induced by the treatment assignment and random sampling from the population under the finite and super-population frameworks, respectively. 
Write $\asz$ if a result holds for almost all sequences of $Z$ under the finite-population framework,  
write $\assp$ if a result holds for almost all sequences of $\{Y_i(1), Y_i(0), x_i, Z_i\}_{i=1}^N$ under the super-population framework, and  write $\asp$ if a result holds for almost all sequences of $\pi\sim \unif(\Pi)$ conditioning on a sequence of observed data $\mathcal{D} = (Y_i, x_i, Z_i)_{i=1}^N$.  Write $\aszp$ if a result holds $\asp$ for almost all sequences of $Z$ under the finite-population framework.

Let $\var_\infty$ and $\cov_\infty$ be the asymptotic variance and covariance,  with the probability measure clear from the context. Write $A \dsim B$ for $\rtn (A- B) =\opss$ for $* = Z, \pi, \textsc{s}$.

\subsection{Central limit theorems}

To analyze complete randomization and ReM, we need a central limit theorem from \citet[][Theorem 5]{DingCLT}: Under complete randomization and Condition \ref{asym},
$$
\rtn   \left(\begin{array}{cc}
\htN - \tau \\
\htx
\end{array}
\right)
\rightsquigarrow  \mN \left\{
0_{J+1},  
  \left(\begin{array}{cc}
v_\neyman & p_1^{-1} \gamma^\T_1+p_0^{-1}\gamma^\T_0\\
p_1^{-1} \gamma_1+p_0^{-1}\gamma_0 & (p_1p_0)^{-1} I_J\end{array}\right)
\right\},
$$
recalling $v_\neyman = p_1^{-1}S_1^2 + p_0^{-1}S_0^2 - S^2_\tau$ from Theorem \ref{neyman}.  When citing this result, we will simply say ``by \textsc{fpclt}.''

To analyze random permutation, we need the following lemma due to \citet{hajek1961some}. To simplify the presentation, we give a version that involves slightly stronger moment conditions than  \citet[][Theorem 4.1]{hajek1961some}.

\begin{lemma} \label{romanoS33}
Let $u = (u_1, \dots, u_N)^\T$ and  $v = (v_1, \dots, v_N)^\T$ be two $N\times 1$ vectors of real numbers, possibly depending on $N$. 
Let $\bar u = N^{-1} \sumi u_i $, $S^2_u=(N-1)^{-1}\sumi  ( u_i - \bar u)  ^2 $, $\bar v = N^{-1} \sumi v_i $, and $S^2_v= (N-1)^{-1}\sumi  ( v_i - \bar v)  ^2$ be the means and variances, respectively. 
We have 
\begin{enumerate}[(a)]
\item\label{clt_1} $
E( N^{-1} u^\T v_\pi) = \bar u \bar v, \quad \cov(N^{-1} u^\T  v_\pi) = N^{-2}(N-1)   S^2_u S^2_v$;
\item\label{clt_2} $\rtn  ( N^{-1}   u^\T v_\pi - \bar u \bar v) \rightsquigarrow \mN(0, S^2_u S^2_v)$ if
(i) $S^2_u$ and $S^2_v$ have finite limits, and (ii) there exists an $\epsilon>0$ such that $N^{-1}\sum_{i=1}^N (u_i - \bar u)^{2+\epsilon} \leq c_0$ and $N^{-1}\sum_{i=1}^N (v_i - \bar v)^{2+\epsilon} \leq c_0 $ for some $ c_0 < \infty$ independent of $N$. 
\end{enumerate}
\end{lemma}

\subsection{A finite-population strong law of large numbers}

Based on \citet[][Lemma S1]{LassoTE16}, \citet[][Lemma A3]{wuanding2020jasa} proved a finite-population strong law of large numbers under simple random sampling. 
We further improve it to allow for rejective sampling in the sense of \citet{fuller}, which includes simple random sampling as a special case with $a = \infty$ below. 
This new finite-population strong law of large numbers in Lemma \ref{slln_rem} is useful for analyzing both complete randomization and ReM. 
Condition \ref{asym} ensures the sequences of $\{Y_i(z)\}_{i=1}^N$, $(x_{ij})_{i=1}^N$, and $\{x_{ij} Y_i(z)\}_{i=1}^N$  satisfy the condition required by Lemma \ref{slln_rem} for all $z = 0,1$ and $j = \ot{J}$.

\begin{lemma} \label{slln_rem} 
Let $(W_i, x_i)_{i=1}^N$ be a sequence of finite populations with means $\bar W = N^{-1}\sumi W_i$ and variances $ S^2_W  = (N-1)^{-1}\sumN (W_i - \bar W)^2$ for $N=\ot{\infty}$.  
Let $\mathcal{I}\subset \{\ot{N}\}$ be a random  sample under rejective sampling, in the sense that we start with $\mathcal{I}$ as a simple random sample yet only accept it if $\htx =|\mathcal{I}|^{-1}\sum_{i\in \mathcal{I}} x_i - (N-|\mathcal{I}|)^{-1}\sum_{i\not\in \mathcal{I}} x_i$ satisfies $ \htx ^\T \{\cov(\htx )\}^{-1}\htx   < a $.
Let  $\hat W_\mathcal{I}= |\mathcal{I}|^{-1}\sum_{i\in \mathcal{I}} W_i$ and 
$\hs^2_\mathcal{I}=(|\mathcal{I}|-1)^{-1}\sum_{i\in\mathcal{I}} (W_i-\hat W_\mathcal{I})^2$ be the sample mean and variance, respectively, and denote by $\ma$  the event of $ \htx ^\T \{\cov(\htx )\}^{-1}\htx   < a $. 

Assume as $N\to \infty$,  
(i) $\bar W$ and  $S^2_W$ have finite limits,   
(ii) there exists a $ c_0 < \infty$ independent of $N$ such that $N^{-1}\sum_{i=1}^N W^4_i \leq c_0$, 
(iii) $\lim_{N\to\infty}|\mathcal{I}|/N >0$, and (iv) $\lim_{N\to\infty} \prob(\mM) = r > 0$. 
We have $
\hat W_\mathcal{I} - \bar W = \oo $ and $
\hs^2_\mathcal{I} - S^2_W = \oo $ for almost all sequences of $\mathcal{I}$.
\end{lemma}

\begin{proof}[Proof of Lemma \ref{slln_rem}]
The probability measure under rejective sampling is equivalent to the probability measure under simple random sampling conditioning on $\ma$. 
We take $\mathcal{I}\subset \{\ot{N}\}$ as a simple random sample of size $|\mathcal I|$ throughout the proof and reflect the rejective sampling via conditioning on $\ma$. 
We proceed by verifying that there exists an $ n_0 $ such that for all  $N >  n_0 $, 
\beginy\label{slln_rem_g}
\max\big\{ \prob( \hat W_\mathcal{I} - \bar W \geq t \mid \ma), \ \prob(  \hat W_\mathcal{I} - \bar W \leq -t \mid \ma)\big\} \leq
2 r^{-1} \exp \left( - \frac{|\mathcal I|^2t^2}{4N S_W^2 }\right)\ \text{for all $t\geq 0$}.
\endy
The result then follows from the Borel--Cantelli lemma via identical reasoning as in the proof of \citet[][Lemma A3]{wuanding2020jasa}.

Let $p_\ma =\prob(\mM)$ for notational simplicity, and let $\mMc$ be the complement of $\mM$. 
The law of total probability ensures
$
\prob(\hat W_\mathcal{I} -\bar W\geq t  ) 
 \geq p_\ma\cdot \prob(\hat W_\mathcal{I} -\bar W\geq t \mid \mM)  
$
and 
$
\prob(\hat W_\mathcal{I} -\bar W\leq t  ) 
 \geq p_\ma\cdot \prob(\hat W_\mathcal{I} -\bar W\leq t \mid \mM) 
$ for all $t\geq 0$, 
such that
\begina
&& \max\big\{ \prob( \hat W_\mathcal{I}  - \bar W \geq t \mid \mM), \ \prob(  \hat W_\mathcal{I}  - \bar W \leq -t \mid \mM)\big\}\\
&&\quad\quad \leq   p_\ma^{-1} \max\big\{ \prob( \hat W_\mathcal{I}  - \bar W \geq t ), \ \prob(  \hat W_\mathcal{I}  - \bar W \leq -t )\big\} \leq
  p_\ma^{-1} \exp \left( - \frac{|\mathcal I|^2t^2}{4 N S_W^2 }\right)\quad\text{for all $t\geq 0$}
  \enda
by \citet[][Lemma A2]{wuanding2020jasa}. The sufficient condition in \eqref{slln_rem_g} then follows from $\limN p_\ma=  r$ such that there exists an $ n_0 $ with
$ p_\ma \geq 2^{-1} r$ for all $N \geq  n_0 $.

\end{proof}

%

\section{Finite-population inference under complete randomization}\label{sec:app_fp}

%
\subsection{Core lemmas}\label{sec:lemma}

The following lemma gives some useful facts about \citet{Fisher35}'s analysis of covariance, i.e.,  the {\ols} fit of $Y$ on $(\on,Z, X)$.
\begin{lemma}\label{alg_f}
{\dgp} the coefficients from the {\ols} fit of $Y$ on $(\on,Z, X)$ are
\begina
\hat \mu_\fisher = \hY - p_1 \htf, \quad  \htF =  \frac{Z^\T(I-H)Y}{Z^\T(I-H)Z}= \htN -\hgf^\T\htx, \quad \hgf = (X^\T X)^{-1} X^\T Y - N p_1p_0\htf (X^\T X)^{-1} \htx,
\enda
respectively, with residuals $\epfi = Y_i -\hY - (Z_i - p_1) \htf - x_i^\T\hgf$ for $i = \ot{N}$ and standard errors
\begina
\ese^2_\fisher = \frac{1}{N-2-J}\left\{ \frac{Y^\T(I-H)Y  }{ Z^\T(I-H)Z} - \htF^2\right\}, \quad 
\tse^2_\fisher = \frac{\eta^\T \dd  \eta}{\{Z^\T(I-H)Z\}^2},
\enda
where $\eta = (I-H) Y -  (I-H)Z \htF$ and $\delta_i$ are the residuals from the \textsc{ols} fit  of  $Z$ on $(\on, X)$. 

For a sequence of $Z$ that ensures $\htx = o(1)$, $\htau_* - \tau = o(1)$ for $* = \neyman, \fisher$, $N^{-1}\sumi \|x_i\|_4^4 = O(1)$, and $N^{-1} \sumi \epfi^4 = O(1)$ as $N$ goes to infinity,  we have
\begina
N \ese_\fisher^2 - \left( p_0^{-1} \hat  S_{\fisher(1)}^2 +  p_1^{-1}  \hat  S_{\fisher(0)}^2 \right) =  o(1),  \quad\quad
N\tse_\fisher^2 - \left( p_1^{-1}\hs_{\fisher(1)}^2 + p_0^{-1}\hs_{\fisher(0)}^2\right) = o(1).
\enda 
\end{lemma}


\begin{proof}[Proof of Lemma \ref{alg_f}]
First, let $u = e = (I-H)Y$ and $v =\delta = (I-H)Z $ in Lemma \ref{lmz} to see
\begina
\htau_0 =   \frac{\delta^\T e}{\dsq} ,\quad  \ese_0^2 = \frac{1}{N-1}\left( \frac{\esq  }{ \dsq} - \htF^2\right),   
\quad \tse_0^2 = \frac{\eta_0^\T \dd  \eta_0}{(\dsq)^2},  
\enda
where $\eta_0 = e-\delta\htau_0$. 
The numerical result for $\htf$, $\sef$, and $\tsef$ follows from $\htF = \htau_0$, $(N-2-J)\ese _\fisher^2  = (N-1)\ese_0^2$, and $\tse_\fisher^2=\tse_0^2$ by \textsc{fwl}. 

Second, let $\chi = (\on, Z, X)$ be the design matrix.
That $N^{-1}X^\T Z =  p_1p_0\htx$ by \eqref{eq::simple-mean-zero-x} ensures
\beginy\label{alg_chif}
N^{-1} \chi^\T\chi = \left(
\begin{array}{ccc}
1& p_1 & 0_J^\T\\
p_1 & p_1 &p_1p_0\htx^\T\\
0_J& p_1p_0\htx& \lambda \sxx
\end{array}\right), \quad \text{with} \quad  \left(
\begin{array}{cc}
1& p_1\\
p_1 & p_1 
\end{array}\right)^{-1} =  p_0 ^{-1}\left(
\begin{array}{cc}
1& - 1\\
- 1 &p_1^{-1} 
\end{array}\right),
\endy
such that 
\begina
N^{-1}\chi^\T \chi \left(
\begin{array}{c}
\hmu_\fisher\\
\htF\\
\hgf
\end{array}\right) = N^{-1} \chi^\T Y 
\quad \Longleftrightarrow\quad 
 \left(
\begin{array}{ccc}
1& p_1 & 0_J^\T\\
p_1 & p_1 & p_1p_0\htx^\T\\
0_J&p_1p_0\htx& \lambda \sxx
\end{array}\right)
\left(
\begin{array}{c}
\hmu_\fisher \\
\htF\\
\hgf
\end{array}\right) =  
\left(
\begin{array}{c}
\hY\\
p_1 \hY  (1) \\
\lambda \qxy
\end{array}\right).
\enda
Directly comparing the rows verifies $
\hat \mu_\fisher = \hY - p_1 \htf$, $ \htF = \htN -\hgf^\T\htx$, and $\hgf =  (X^\T X)^{-1} X^\T Y - \lambda ^{-1}p_1p_0\htf (S_x^2)^{-1} \htx$. 
The  expression of $\epfi$ then follows.

Third, $N^{-1}\sumi \epfi^4 = O(1)$ ensures $\hsfz = O(1)$ for $z=0,1$. 
With $\ese_\fisher^2$ being the $(2,2)$th element of 
\begina
 \frac{\sumi \ep_{\fisher,i}^2}{N-J-2} (\chi^\T \chi)^{-1},
\enda
the limit of $ \ese_\fisher^2 $ follows from 
\begina
(N-1)^{-1}\sumi \ep_{\fisher,i}^2 =  \frac{N_1-1}{N-1} \hs^2_{\fisher(1)} + \frac{N_0-1}{N-1} \hs^2_{\fisher(0)} +  \frac{N}{N-1} p_0p_1 \{\hat\ep_{\fisher}(1) - \hat\ep_{\fisher}(0)\}^2 = p_1  \hs^2_{\fisher(1)}  + p_0 \hs^2_{\fisher(0)} + o(1)
\enda
and 
\begina  
  (N^{-1}\chi^\T \chi)^{-1} =
 \left(
\begin{array}{ccc}
p_0^{-1}& - p_0^{-1} & 0_J^\T\\
- p_0^{-1} & p_0^{-1}p_1^{-1} & 0_J^\T \\
0_J&0_J &   (\sxx)^{-1}
\end{array}\right)
+o(1).
\enda
Further let  $\Delta = \diag(\epfi^2)$. The robust covariance estimator is 
$
\tilde V_\fisher =  ( \chi^\T \chi)^{-1}( \chi^\T \Delta \chi) ( \chi^\T \chi)^{-1}
$
with $\tse_\fisher^2$ as the $(2,2)$th element. 
That $N^{-1}\sumi \|x_i\|_4^4 = O(1)$ and $N^{-1} \sumi \epfi^4 = O(1)$ ensures 
\begina
N^{-1}\chi^\T \Delta \chi =N^{-1}
\left(\begin{array}{ccc}
\on ^\T \Delta \on  & \on ^\T \Delta Z &\on ^\T \Delta X \\
Z^\T \Delta \on  & Z^\T \Delta Z &Z^\T \Delta X \\
X^\T \Delta \on  & X^\T \Delta Z &X^\T \Delta X
\end{array}\right)  = O(1)
\enda
with 
(i) $ 0 <    \ont \Delta Z =   Z^\T \Delta Z  \leq    \on ^\T \Delta \on  \leq (N\sumi \epfi^4)^{1/2} = O(N)$, 
(ii) $X^\T\Delta X =\sumi \epfi^2 x_i x_i^\T= O(N)$,  
and (iii) 
$\ont \Delta X \leq  (\ont \Delta \on )^{1/2} (X^\T\Delta X)^{1/2}$ and $Z^\T \Delta X \leq  (Z^\T \Delta Z )^{1/2} (X^\T\Delta X)^{1/2}$. 
This, together with $\htx  =o(1)$ in \eqref{alg_chif}, ensures
\begina
N \tilde V_\fisher = p_0^{-1} \left(\begin{array}{ccc}
1& - 1 & 0_J^\T \\
- 1 &p_1^{-1} & 0_J^\T\\
 0_J& 0_J & p_0 S_x^2
\end{array}\right) N^{-1}
\left(\begin{array}{ccc}
\on ^\T \Delta \on  & \on ^\T \Delta Z &\on ^\T \Delta X \\
Z^\T \Delta \on  & Z^\T \Delta Z &Z^\T \Delta X \\
X^\T \Delta \on  & X^\T \Delta Z &X^\T \Delta X
\end{array}\right)
p_0^{-1}\left(\begin{array}{ccc}
1& - 1 & 0_J^\T \\
- 1 &p_1^{-1} & 0_J^\T\\
 0_J& 0_J & p_0 S_x^2
\end{array}\right) + o(1)
\enda
with the $(2,2)$th element as
\begina
N\tse_\fisher^2 
&=& p_0^{-2}N^{-1} (-1,p_1^{-1}) \left(\begin{array}{cccc}
\on ^\T \Delta \on  & \on ^\T \Delta Z \\
Z^\T \Delta \on  & Z^\T \Delta Z 
\end{array}\right) \left(\begin{array}{c}
-1\\ p_1^{-1}
\end{array}\right) + o(1) \\
&=& p_0^{-2} N^{-1}\left(\on ^\T\Delta \on  - p_1^{-1}Z^\T\Delta \on  - p_1^{-1} \on ^\T\Delta Z + p_1^{-2} Z^\T\Delta Z\right) + o(1)\\
&=& (p_0p_1)^{-2}   N^{-1}(Z - p_1 \on )^\T\Delta (Z - p_1 \on ) + o(1)\\
&=& p_1^{-1} \hs^2_{\fisher(1)} + p_0^{-1} \hs^2_{\fisher(0)} +o(1).
\enda
\end{proof}

 \begin{lemma}\label{as_Z}
Assume Condition \ref{asym} and complete randomization.
As $N\to \infty$, the following results hold $\asz$: 
\begin{enumerate}[(a)]
\item\label{as_1} $
\hY(z) - \bY(z) = \oo, \ \  \hs^2_z - S^2_z = \oo, \ \  \hat x(z) = \oo, \ \ \hs^2_{x(z)} - S_x^2 = \oo, \ \ \hs_{xY(z)} - S_{xY(z)} = o(1)$  for $z = 0,1$.
\item\label{as_2} 
The sequence of finite populations $(Y_i, Y_i, x_i)_{i=1}^N$ satisfies Condition \ref{asym} with $N^{-1}\sumi Y_i^4 = O(1)$ and $
\hY - \{ p_1   \bar{Y} (1) + p_0    \bar{Y} (0) \} =\oo $,  $\hat S^2 - S^2 - p_1p_0\tau^2 = \oo $, $\hgr  - \gamma =o(1)$, $\hat  S_e^2- S_a^2- p_1p_0\tau^2 = \oo$ giving the analogs of $\bar Y(z)$, $S_z^2$,  $\gamma_z$, and both $S_{a(z)}^2$ and $S_{b(z)}^2$ for $z=0,1$, respectively.

\item\label{as_3}  $
\hat e = 0, \quad \hat  S_e^2 =  S_a^2+ p_1p_0\tau^2+o(1), \quad N^{-1}\sumi e_i^4 = O(1)$.
\item\label{as_4} $\hat \delta = N^{-1}\sumi \delta_i = 0, \quad \hat S_\delta^2 = (N-1)^{-1}\sumi (\delta_i - \hat\delta)^2 = p_1p_0 + \oo, \quad N^{-1}\sumi \delta_i^4 = O(1)$. 
\item\label{as_f} $\hat\mu_\fisher - \hY = o (1), \quad \htf - \tau = o(1), \quad \hgf - \gamma=o(1),$ with \\
$
\hat\ep_\fisher = N^{-1}\sumi \epfi = 0, \quad \hat  S_\fisher^2 = (N-1)^{-1} \sumi\epfi ^2
= S_a^2 + o(1), \quad N^{-1}\sumi \epfi^4 = O(1).$ 
\end{enumerate}

\end{lemma}

Lemma \ref{as_Z} ensures it suffices to focus on the sequences of $Z$ that satisfy \eqref{as_1}--\eqref{as_f} when verifying  results on almost sure convergence  under $P_Z$.

\begin{proof}[Proof of Lemma \ref{as_Z}]
Recall from Condition \ref{asym} that $w_i(z)= (S_x^2)^{-1}x_iY_i(z) = x_iY_i(z)$ with mean and covariance matrix
$\bar w(z) = N^{-1}\sumi w_i(z)  = \lambda S_{xY(z)}$ and $S^2_{w(z)} =(N-1)^{-1} \sumi \{w_i(z)- \bar w(z)\}^2$. Its observed analog
$w_i = x_i Y_i$ has sample means and covariance matrices $\hat w(z) = N_z^{-1} \sum_{i: Z_i = z} w_i$ and $\hs^2_{w(z)} = (N_z-1)^{-1} \sum_{i: Z_i = z} \{w_i - \hat w(z)\}\{w_i - \hat w(z)\}^\T$ for $z = 0,1$.
Lemma \ref{slln_rem} ensures 
\beginy\label{slln_w}
\hat w(z)  - S_{xY(z)} = \ooas, \quad \hs^2_{w(z)} - S^2_{w(z)} = \ooas
\endy
under Condition \ref{asym}. 
The result for $\hY(z)$, $\hs_z^2$, $\hat x(z)$, and $\hs_{x(z)}^2$ in statement \eqref{as_1} follows from Lemma \ref{slln_rem} directly.
The result for $\hat S_{xY(z)}$ then follows from 
\begina
\hat S_{xY(z)} = (N_z-1)^{-1}\sum_{i:Z_i = z} \{x_i - \hat x(z)\} \{Y_i - \hat Y(z)\} = \frac{N_z}{N_z-1}\hat w(z) - \frac{N_z}{N_z-1} \hat x(z) \hat Y(z). 
\enda
This verifies statement \eqref{as_1} and \eqref{slln_w} hold $\asz$, such that it suffices to verify  statements \eqref{as_2}--\eqref{as_4} hold for $Z$'s that satisfy  statement \eqref{as_1} and \eqref{slln_w}. 
Fix one such sequence for the rest of proof.

For statement \eqref{as_2}, the correspondence between the analogs follows from definitions with the analogs of $a_i(z)$ and $b_i(z)$ given by  
$Y _i - \hY - x_i^\T \hgr = e_i$ ensured by \eqref{algebra_r}.
The limits of $\hY$ and $\hs^2$ follow from 
 statement \eqref{as_1} and \eqref{as_numeric}.
With $(w_i, w_i)$ as the analogs of $\{w_i(1), w_i(0)\}$ in the finite population $(Y_i, Y_i, x_i)_{i=1}^N$, the limits of $\hat w = N^{-1}\sumi w_i$ and $\hat S^2_w = (N-1)^{-1}\sumi (w_i - \hat w) (w_i - \hat w)^\T$ follow from 
\eqref{slln_w} and applying \eqref{as_numeric} entry-wise.
This in turn ensures $\hat S_{xY}$, as the analog of  $S_{xY(z)}$, satisfies $\hat S_{xY}=\hgr =  \lambda ^{-1}  \hat w = \gamma +\oo$ with $
\hat  S_e^2  = \hat S^2 - \|\hgr\|^2_2$ having finite positive limit $S^2 +p_1p_0\tau^2 - \|\gamma\|^2_2= S^2_a +p_1p_0\tau^2$ by \eqref{Sab} and  \eqref{algebra_r}. 
This verifies Condition \ref{asym}(ii). 
Further, $N^{-1}\sumN Y_i^4 \leq N^{-1}\sumN \{Y^4_i(1) + Y^4_i(0)\} \leq 2c_0$; likewise for $N^{-1}\sumN \|w_i\|_4^4 \leq 2c_0$.
This verifies Condition \ref{asym}(iii) and hence statement \eqref{as_2}. 

For statement \eqref{as_3}, that $\hat e = 0$ follows from \eqref{algebra_r} and the variance follows from statement \eqref{as_2}. Further, $\hgr = \gamma+o(1)$  from statement \eqref{as_2} ensures $\|\hgr\|_\infty = O(1)$ such that $(x_i^\T \hgr)^4 \leq c_1 \|x_i\|_4^4$ for some $c_1$ independent of $N$.
This, together with 
$\ei^4 = (\ei^2)^2 \leq \{ 2(Y _i - \hat Y)^2 + 2(x_i^\T \hgr)^2 \}^2 \leq 8 \{ (Y _i - \hat Y)^4 +  (x_i^\T \hgr)^4  \}$, ensures 
$
N^{-1}\sumi \ei^4  = O(1)$ 
and hence  statement \eqref{as_3}.

For statement \eqref{as_4}, the limit of $\hs^2_\delta$ follows from \eqref{H} and $\htx = 0_J+ \oo$ by  statement \eqref{as_1}. 
With $\delta_i = Z_i -p_1 - \lambda^{-1}p_1p_0 x_i^\T\htx$ from \eqref{H}, we have  
$
 \delta_i^4 \leq  8  \{  (Z_i -p_1)^4  +   (\lambda^{-1}p_1p_0 x_i^\T\htx)^4  \}
$  and hence $
N^{-1}\sumi \delta_i^4 = O(1) 
$ by the same reasoning as that for $N^{-1}\sumi e_i^4 = O(1)$ in statement \eqref{as_3}. 

Statement \eqref{as_f} follows from $N^{-1} \zt(I-H) Y= N^{-1}\sumi \delta_iY_i = p_1\hat Y(1) - p_1\hY - p_1p_0\hgr^\T\htx = p_1p_0\htn  - p_1p_0\hgr^\T\htx$  and $N^{-1} Z(I-H)Z =  p_1p_0 + o(1)$ by \eqref{H}. 
This ensures $\htf = \tau+o(1)$ by Lemma \ref{alg_f} and thus the result on $\hat\mu_\fisher$ and $\hgf$. 
Further, replace $Y_i$ with $\ep_{\fisher,i}$ in \eqref{as_numeric} to see
\begina
\hs_\fisher^2 = 
 \frac{N_1-1}{N-1} \hs_{\fisher(1)}^2 + \frac{N_0-1}{N-1} \hs_{\fisher(0)}^2 
 + \frac{N}{N-1} p_0p_1 \left\{ \hat\ep_{\fisher}(1) - \hat\ep_{\fisher}(0) \right\} ^2,
\enda
where $\hsfz = S_{a(z)}^2 + o(1) $ by \eqref{Sab} and \eqref{syse}, and $ \hat\ep_{\fisher}(1) - \hat\ep_{\fisher}(0) = \htn - \htf - \htx^\T\hgf = o(1) $ by $\epfi = Y_i -\hY - (Z_i-p_1) \htf - x_i^\T\hgf$ from Lemma \ref{alg_f}. We have  $\hat\ep_\fisher = 0$ and 
$
N^{-1}\sumi \epfi^4 \leq  27 \{ N^{-1}\sumi (Y_i -\hY)^4  + N^{-1}\sumi Z_i-p_1)^4 \htf ^4+  N^{-1}\sumi (x_i^\T \hgf)^4\} = O(1) 
$ by the Cauchy--Schwarz inequality. 
\end{proof}

Technically, the first strategy by \citet{CovAdjRosen02} takes $e$ as the fixed input for conducting {\frt} and thus has no counterpart for $\hgr$ under $Z_\pi$.
Nevertheless, the procedure is identical to one that takes  $(Y_i, x_i)_{i=1}^N$ as the fixed input, regresses $Y$ on $(\on, X)$ to generate $e^\pi = e$ and $\hgr^\pi = \hgr$ independent of $Z_\pi$, and then regresses $e^\pi=e$ on $(\on,Z_\pi)$ to generate $\htrp$, $\ese_\rosenbaum^\pi$, and $\tse_\rosenbaum^\pi$. 
This unifies the four procedures,  $*=\neyman, \rosenbaum, \fisher,\lin$, as all taking $(Y_i,  x_i)_{i=1}^N$ as the fixed input for conducting {\frt}. 
We take this perspective to simplify the presentation. 

\begin{lemma}\label{v_lim}
{\asym}
\begin{enumerate}[(a)]
\item\label{vl_1}$\hg_*-\gamma=\oo$  for $* = \rosenbaum,\fisher, \quad \hg_\lin -(p_0\gamma_1+ p_1\gamma_0) =\oo$,\\
$N \sen ^2  -   (p_1p_0)^{-1} S^2 =o(1),  \quad  N \tse_\neyman^2 - \left(p_1^{-1}  S_1^2 + p_0^{-1}  S_0^2 \right)=o(1),$\\
$N \ese_*^2  -   (p_1p_0)^{-1}  S_a^2  =o(1),  \quad N \tse_*^2 - (p_1^{-1}  S_{a(1)}^2+ p_0^{-1}  S_{a(0)}^2)=o(1) \quad \text{for $*=\rosenbaum,\fisher$}$,\\
$N \ese_\lin^2  -   (p_1p_0)^{-1}  S_b^2 =o(1), \quad N \tse_\lin^2 - (p_1^{-1}  S_{b(1)}^2 + p_0^{-1}  S_{b(0)}^2 ) =o(1)$

hold $\asz$, with $N\ese_*^2$ and $N\tse_*^2$ all having positive finite limits.
\item\label{vl_2} $\hg_*^\pi  -\gamma=\oo \quad \text{for $* = \rosenbaum, \fisher, \lin$}$, with $\hgr^\pi = \hgr$, \\
$N (\sen^2)^\pi   -   (p_1p_0)^{-1} S^2 - \tau^2 =\oo,  \quad N (\tsen^2)^\pi   -   (p_1p_0)^{-1} S^2  - \tau^2 =\oo$,\\
$N (\ese^2_*)^\pi -   (p_1p_0)^{-1}  S_a^2 - \tau^2  =\oo,  \quad N (\tse^2_*)^\pi  -   (p_1p_0)^{-1}  S_a^2  - \tau^2 =\oo \quad \text{ for $* = \rosenbaum,\fisher,\lin$}$

hold $\aszp$, with $N(\ese_*^2)^\pi$ and $N(\tse_*^2)^\pi$ all having positive finite limits.
\end{enumerate}
\end{lemma}

\citet[][Theorem 2]{Lin13} showed $N\ese_\fisher^2 - ( p_0^{-1}   S_{a(1)}^2 + p_1^{-1}   S_{a(0)}^2)=\opz$ and $N\tse_\fisher^2 - ( p_1^{-1}S_{a(1)}^2 + p_0^{-1}S_{a(0)}^2) = \opz $ as the probability limits. Lemma \ref{v_lim} strengthens the results to almost sure convergence.

\begin{proof}[Proof of Lemma \ref{v_lim}]
{\ass}

For statement \eqref{vl_1} regarding the sampling distributions, the limits of $\hg_*$ follow from $
\hgr =(\sxx)^{-1}\hat S_{xY}$,  $ \hgf = \hgr -  \lambda^{-1}p_1p_0\htf ( S_x^2)^{-1} \htx$, and $\hg_\lin = p_0 \hg_{\lin,1}+p_1\hg_{\lin,0}$ with $\hg_{\lin,z} = (\hs_{x(z)}^2)^{-1} \hs_{xY(z)}$ by Proposition \ref{FLRtoN}, where $\hgr = o(1)$, $\htf - \tau =o(1)$, and $\hg_{\lin,z} - \gamma_z = o(1)$ by Lemma  \ref{as_Z}. 
This ensures $\hat S_{\neyman(z)}^2 - S_{z}^2 =\oo$, $\hat S_{*(z)}^2 - S_{a(z)}^2 =\oo$ for $* = \rosenbaum,\fisher$, and $\hat S_{\lin(z)}^2 - S_{b(z)}^2 =\oo$ by \eqref{Sab} and \eqref{syse},   
and allows us to unify the standard errors as 
\begina
N \ese_*^2 - \left( p_0^{-1} \hat  S_{*(1)}^2 +  p_1^{-1}  \hat  S_{*(0)}^2 \right) =  o(1),\quad
N\tse_*^2 - \left( p_1^{-1}\hs_{*(1)}^2 + p_0^{-1}\hs_{*(0)}^2\right) = o(1) \ \ph.
\enda
In particular, the result for $\ese_*^2$ follows from \eqref{syse} provided $\hat S_{*(z)}^2$ all have finite limits. 
The result for $\tse_\neyman^2$ and $\tse_\rosenbaum^2$ follows from 
 \eqref{syse}.
The result for $\tse_\fisher^2$ follows from Lemma \ref{alg_f} with the regularity condition ensured by Lemma \ref{as_Z}\eqref{as_f}. 
The result for $\tse_\lin^2$ follows from \citet[][Theorem 8]{LD20} given Condition \ref{asym} implies  \citet[][Condition 1]{LD20} by \citet[][Proposition 1]{wuanding2020jasa}.
Alternatively,  almost identical algebra as in the proof of Lemma \ref{alg_f} for the limiting value of $N\tse_\fisher^2$ attains the same end. 
Condition \ref{asym}(ii) and \eqref{Sab} ensure the limits are all positive.

For statement \eqref{vl_2} regarding the  randomization distributions,  {\frt} takes $(Y_i, x_i)_{i=1}^N$ as the fixed input for permuting the treatment vector, and thereby induces a sequence of finite populations $(Y_i, Y_i, x_i)_{i=1}^N$ with $Y_i$ as the ``pseudo potential outcomes" under both treatment and control. The way we chose the fixed $Z$ further ensures $(Y_i, Y_i, x_i)_{i=1}^N$ satisfies Condition \ref{asym}.
The result in statement \eqref{vl_1} regarding the sampling distributions thus also holds for $\hg_*^\pi $ and $\hsszp$ $\asp$ if we replace $\gamma$ with $\hgr$, $S_z^2$ with $\hs^2$,  and both $S_{a(z)}^2$  and $S_{b(z)}^2$  with  $\hs_e^2$ by the correspondence result from Lemma \ref{as_Z}\eqref{as_2}.
The result follows from $\hgr - \gamma = o(1)$, $\hs^2-  S^2- p_1p_0\tau^2 = o(1)$, and $ \hs^2_e- S_a^2-p_1p_0\tau^2 = o(1)$.

\end{proof}


\begin{lemma}\label{juxta} 
\asym
\begin{enumerate}[(a)]
\item 
$
 \htL  \dsim \htN- (p_0\gamma_1 + p_1\gamma_0)^\T \htx, \quad \htF \dsim \htR \dsim \htn  - \gamma^\T \htx$.
\item $
 \htrp \dsim  \htfp \dsim \htlp \dsim \htnp - \gamma^\T \htxp$ holds $\asz$. 
 \end{enumerate}
\end{lemma}

\begin{proof}[Proof of Lemma \ref{juxta}]
First, let $\gamma_\rosenbaum = \gamma_\fisher = \gamma$ and $\gamma_\lin = p_0 \gamma_1+ p_1\gamma_0$ to write $\hg_* - \gamma_* = \ooas$ by Lemma \ref{v_lim}.
This, together with $ \hts = \htn - \hg_*^\T \htx $ by Proposition \ref{FLRtoN} and $
\rtn  \htx \rightsquigarrow \mN\{0_J, (p_0p_1)^{-1}I_J\}$ by \textsc{fpclt}, ensures
$
\rtn  \left\{ \htH - ( \htN - \gamma_*^\T \htx)\right\}  = (\hg_* - \gamma_*)^\T \rtn  \htx = \opz 
$
by  Slutsky's theorem  and  hence  $\htH \dsim  \htN - \gamma_*^\T \htx$. 

Second, Proposition \ref{FLRtoN} is algebraic and holds under the probability measure induced by $\pi$ as well. This ensures $\htsp = \htnp - \gamma^\T \htxp - (\hg_*^\pi - \gamma)^\T \htxp$ for $* = \rosenbaum, \fisher,\lin$, with $\hgr^\pi = \hgr = \gamma + \oo \ \asz $ and 
$\hg_*^\pi - \gamma = \oo \ \asp$ for $*=\fisher,\lin$ $\asz$ by Lemma \ref{v_lim}. 
Almost identical reasoning as that for $\hts$ ensures $ (\hg_*^\pi - \gamma)^\T \rtn  \htxp = \opp$ holds $\asz$  for $* = \rosenbaum, \fisher,\lin$ and hence the result.
\end{proof}

\subsection{Proofs of the main results under the finite-population framework}

Theorem \ref{neyman} is  a special case of \citet{wuanding2020jasa}. 
We give below a unified proof. 

\begin{proof}[Proof of Theorems \ref{neyman}--\ref{lin}] 
First, Condition \ref{asym} ensures $ S_{a(z)}^2 $, $ S_a^2 $, $ S_{b(z)}^2 $, and $ S_b^2 $ have positive finite limits, and $ S_\xi^2 $ and $S_\tau^2$  have finite limits by \eqref{Sab}.

The result for $\htN$ and $\htnp$ follows from \citet{DD18},
with 
$
v_\neyman =  p_1^{-1} S_1^2  + p_0^{-1} S_0^2  -  S_\tau^2$ and $
 v_{\neyman0} = \lim_{N\to\infty}(p_1p_0)^{-1} \hs^2 =  p_0^{-1} S_1^2  + p_1^{-1} S_0^2 + \tau^2$.
This ensures $\rtn (\htN-\tau)/v_\neyman^{1/2} \rightsquigarrow  \mN(0, 1)$, and $
\rtn \htnp/v_{\neyman0}^{1/2} \rightsquigarrow  \mN(0, 1) $ $P_Z$-a.s.. 
The result for the studentized variants follows from 
\beginy \label{slutsky_neyman}
\wa{1.4}
\begin{array}{ll}
({\htN}-\tau)/{ \sen } 
= \rtn ({\htN}-\tau)/v_\neyman^{1/2} \cdot ({ v_\neyman }/{ N\sen ^2})^{1/2},
&\ 
\htnp/\senp  
= \rtn \htnp/v_{\neyman0}^{1/2} \cdot  \{{ v_{\neyman0} }/ N (\ese^2_{\neyman})^\pi \}^{1/2},\\
({\htN}-\tau)/{ \tse_\neyman} 
= \rtn ({\htN}-\tau)/v_\neyman^{1/2} \cdot  ({ v_\neyman }/{ N\tse_\neyman^2})^{1/2},
&
\ 
\htnp/\tsenp  
= \rtn \htnp/v_{\neyman0}^{1/2} \cdot  \{{ v_{\neyman0} }/ N (\tse^2_{\neyman})^\pi \}^{1/2}
\end{array}
\endy
by Slutsky's theorem with $(v_\neyman/ N \sen ^2) = c'_\neyman + \ooas, \ (v_\neyman /N \tse_\neyman^2)= c_\neyman + \ooas, \ v_{\neyman0}/ N(\ese^2_{\neyman})^\pi =   1 + \ooaszp$, and $v_{\neyman0} /N(\tse^2_{\neyman})^\pi =   1+ \ooaszp$
 by Lemma \ref{v_lim}. 

For the nine covariate-adjusted variants with $*=\rosenbaum, \fisher, \lin$, 
the result for $\hts$ and $\htsp$ follows from the result for $\htN$ and $\htnp$, \textsc{fpclt},  and Lemma  \ref{juxta}, 
 with  
\begina
v_\fisher = v_\rosenbaum & =& \var_\infty(\htr) = v_\neyman +  (p_1p_0)^{-1}  \|\gamma\|^2_2 - 2  (p_1^{-1} \gamma_1 + p_0^{-1} \gamma_0)^\T\gamma 
= p_1^{-1}  S_{a(1)}^2+ p_0^{-1} S_{a(0)}^2 -  S_\tau^2,\\
 v_\lin &=&   \var_\infty(\htl) = v_\neyman  -  (p_1p_0)^{-1}\|p_0\gamma_1 + p_1\gamma_0\|^2_2 = p_1^{-1}  S_{b(1)}^2+ p_0^{-1} S_{b(0)}^2 -  S_\xi^2,\\
 v_{*0} &=&   \var_\infty(\htsp)   = v_{\neyman0} -  (p_1p_0)^{-1} \|\gamma\|^2_2 = (p_1p_0)^{-1}  S_a^2 + \tau^2 \quad\text{for $*=\rosenbaum,\fisher, \lin$}
\enda
by \eqref{Sab}. 
This ensures   $\rtn (\hts-\tau)/{ v_* ^{1/2}}\rightsquigarrow  \mN(0, 1)$, and $\rtn \htsp/{ v_{*0}^{1/2}} \rightsquigarrow  \mN(0, 1) \ \asz$ for $* = \rosenbaum,\fisher,\lin$. 
The result for the studentized variants then follows from Slutsky's theorem and Lemma \ref{v_lim} via the same reasoning as in \eqref{slutsky_neyman}.
\end{proof}

\begin{proof}[Proof of Proposition \ref{FLRtoN}]
That $\htR = \htN - \htx^\T \hgr$ follows from \eqref{algebra_r}. 
That $\htF = \htN- \htx^\T\hgf $ and the expression for $\hgf$ follows from Lemma \ref{alg_f}. 
That $\htl = \htN - \{\hat x   (1)\}^\T \hglo +  \{\hat x   (0)\}^\T \hglz$ for centered covariates follows from \citet{Lin13}, which further  simplifies to $\htl = \htN-  \htx^\T\hg_{\lin}$ because of \eqref{eq::simple-mean-zero-x}.
\end{proof}

%
%
%

\begin{proof}[Proof of Corollary \ref{power}]
The result follows from Lemma \ref{v_lim} with $S^2_{b(z)} \leq S^2_{a(z)}$ and $S^2_{b(z)} \leq S^2_{z}$ by \eqref{Sab}.  
\end{proof}

\section{{\textsc{FRT}} under ReM} \label{sec:app_rem}
%
%

\begin{lemma}\label{v_Rem}
Assume Condition \ref{asym}. 
Lemmas \ref{as_Z}--\ref{v_lim} also hold under ReM.
\end{lemma}

\begin{proof}
Lemmas \ref{as_Z}--\ref{v_lim} follow from the strong law of large numbers in Lemma \ref{slln_rem} via the same reasoning as that under complete randomization,  
with the convergence in probability of the sample variances to their respective finite-population variances ensured by  
\citet[][Lemma A16]{LD2018}. 
\end{proof}

\begin{proof}[Proof of Theorem \ref{Rem}]
The probability measure under ReM is equivalent to the probability measure under complete randomization conditioning on $\ma$. 
We take the distribution under complete randomization as the default distribution throughout this proof, and use ``$|\text{ }\mM$'' to denote the condition of either $\htx^\T\{\cov(\htx)\}^{-1}\htx < a$ under the sampling distribution or
$(\htxp)^\T\{\cov(\htxp)\}^{-1}\htxp < a$ under the randomization distribution.
As such, $T\mid \ma$ gives the sampling distribution of $T$ under ReM, and $T^\pi \mid \ma \sim T^{\pi|\ma}$  gives the randomization distribution of $T$ under {\frt} with ReM. 

The result for $\htl$, $\htl/\ese_\lin$, and $\htl/\tse_\lin$ follows from the asymptotic independence between $\htl$ and $\htx$ and between $\htlp$ and $\htxp$ under complete randomization \citep{LD20}, such that ReM in either case does not affect the limiting distributions.
We verify below the result for $* = \neyman, \rosenbaum,\fisher$ together, starting from the unstudentized $\hts$ and $\htsp$ for $*=\neyman,\rosenbaum,\fisher$ and then moving onto the studentized variants. 

For the distributions of the unstudentized $\hts$ and $\htsp$, where $*=\neyman,\rosenbaum,\fisher$, Lemma \ref{juxta} ensures $\htfp \dsim \ \htrp \dsim \htlp$ such that it suffices to verify the result for $\htn$, $\htr$, $\htf$, 
and $\htnp$. 
Let $D\sim \mN(0_J,I_J)$ and $\epsilon \sim \mN(0,1)$ be two independent standard normals to represent 
\begina\rtn \htx \rightsquigarrow v_xD,\quad \rtn \htxp \rightsquigarrow v_xD, \quad \rtn (\htl -\tau) \rightsquigarrow v_\lin^{1/2} \epsilon, \quad \rtn \htlp \rightsquigarrow v_{\lin0}^{1/2} \epsilon,
\enda
where $ v_x = (p_1p_0)^{-1/2}$ by \textsc{fpclt}.
With 
\begina
\htN  \dsim \htL + (p_0\gamma_1 + p_1\gamma_0)^\T \htx,  \quad
\htF \dsim \htR \dsim \htL  - (p_1-p_0)(\gamma_1-\gamma_0)^\T \htx,\quad
\htnp \dsim \htlp + \gamma^\T \htx \ \asz
\enda
also from Lemma \ref{juxta}, we have 
 \begin{eqnarray}
\rtn  (\htN -\tau) \mid\mM &\rightsquigarrow& v_\lin^{1/2} \epsilon   +v_x (p_0\gamma_1 + p_1\gamma_0)^\T D \mid (\| D\|^2_2 <a),\label{N}\\
\rtn \htnp \mid\mM &\rightsquigarrow& v_{\lin0}^{1/2} \epsilon  + v_x\gamma^\T D \mid (\| D\|^2_2 <a) \quad\quad  \asz,\label{Np}\\
 \rtn ( \hts-\tau)\mid\mM &\rightsquigarrow&  v_\lin^{1/2} \epsilon  - v_x(p_1-p_0)(\gamma_1-\gamma_0)^\T D \mid (\| D\|^2_2 <a) \quad\text{for $*=\rosenbaum,\fisher$},\label{FR}
\end{eqnarray}
with
\begin{eqnarray}\label{vs}
v_x (p_0\gamma_1 + p_1\gamma_0)^\T D &\sim& \mN(0, v_\neyman - v_\lin), 
\quad 
v_x\gamma^\T D \sim \mN(0, v_{\neyman0} - v_{\lin0}),\nonumber\\ 
-v_x(p_1-p_0)(\gamma_1-\gamma_0)^\T D&\sim& \mN(0, v_\rosenbaum - v_\lin).
\end{eqnarray}
Let $u = (v_\neyman - v_\lin )^{-1/2}v_x (p_0\gamma_1 + p_1\gamma_0)$ be a unit vector with $\|u\|^2_2 = 1$ and 
$D_\neyman = u^\T D \sim \mN(0,1)$ by the first ``$\sim$'' in \eqref{vs}. 
Complete $u$ into an orthogonal matrix $\Gamma_\neyman$ with $u^\T$ as the first row and $D_\neyman$ as the first element of $\Gamma_\neyman D$.
With  $\Gamma_\neyman D \sim \mN(0_J, I_J)$ and $\|\Gamma_\neyman D\|^2_2 = \|D\|^2_2$, it follows from \eqref{vs} that 
$$
v_x (p_0\gamma_1 + p_1\gamma_0)^\T D \mid (\| D\|^2_2 <a) \sim  (v_\neyman - v_\lin )^{1/2} D_\neyman \mid  (\|\Gamma_\neyman D\|^2_2 <a)
\sim   (v_\neyman - v_\lin )^{1/2} \mL.
$$
Plugging this back in \eqref{N} proves
$$ \rtn (\htN-\tau) \mid\mM  \rightsquigarrow   v_\lin^{1/2} \epsilon +  (v_\neyman - v_\lin )^{1/2} \mL = v^{1/2}_\neyman\big\{ (1-\rho ^2_\neyman)^{1/2} \cdot \epsilon + \rho _\neyman\cdot \mL \big\} =  v^{1/2}_\neyman \cdot \mU (\rho _\neyman).$$
The same reasoning verifies $
\rtn \htnp \mid \mM 
 \rightsquigarrow  v_{\neyman0}^{1/2} \cdot \mU (\rho _{\neyman0})$ and $
 \rtn (\htH -\tau) \mid \mM   \rightsquigarrow   v_*^{1/2} \cdot \mU (\rho _*)$ for $* =  \rosenbaum, \fisher$
from \eqref{Np} and \eqref{FR} 
with $u = (v_{\neyman0} - v_{\lin0} )^{-1/2}v_x\gamma$ and $u = (v_\rosenbaum - v_\lin )^{-1/2}v_x(p_1-p_0)(\gamma_1-\gamma_0)$, respectively. 
This verifies the result for $\hts$ and $\htsp$, and ensures $\rtn (\htH-\tau)/ v_*^{1/2} \mid \mM  \rightsquigarrow \mU (\rho _*)$  for  $*=\neyman, \rosenbaum, \fisher$, and 
\begina
\rtn \htnp/ v_{\neyman0}^{1/2} \mid\mM  \rightsquigarrow  \mU (\rho _{\neyman0})\quad\asz, \quad \rtn \hts^\pi /v_{*0}^{1/2} \mid \mM  \rightsquigarrow  \mN(0,1) \quad\asz\quad\text{for $*=\rosenbaum, \fisher$.}
\enda
It follows from Lemma \ref{v_Rem} and Slutsky's  Theorem that 
\begina
\rtn (\htH-\tau)/\ese_* \mid \mM 
&=& \rtn (\htH-\tau)/ v_*^{1/2}\cdot (v_* /N \ese^2_*)^{1/2} \mid \mM  \rightsquigarrow (c'_*)^{1/2}  \cdot \mU (\rho _*),\\
\rtn (\htH-\tau)/\tse_* \mid \mM 
&=& \rtn (\htH-\tau)/ v_*^{1/2}\cdot (v_* / N\tse^2_*)^{1/2}  \mid \mM  \rightsquigarrow  c^{1/2}_*\cdot \mU (\rho _*)
\enda
for $* = \neyman, \rosenbaum, \fisher$, and 
\begina
\rtn \htnp/\senp  \mid \mM 
&=& \rtn \htnp/ v_{\neyman0}^{1/2}\cdot (v_{\neyman0} / N\sen ^2)^{1/2}  \mid \mM  \rightsquigarrow   \mU (\rho _{\neyman0}),\\
\rtn \htnp/\tsenp \mid \mM 
&=& \rtn \htnp/ v_{\neyman0}^{1/2}\cdot (v_{\neyman0} /N\tse_\neyman^2)^{1/2} \mid \mM  \rightsquigarrow   \mU (\rho _{\neyman0}),\\
\rtn \hts^\pi/\sesp  \mid \mM
& \rightsquigarrow & \mN(0, 1),\quad\quad 
\rtn \hts^\pi/\tsesp  \mid \mM
 \rightsquigarrow  \mN(0, 1) \quad \text{for $* =\rosenbaum,  \fisher$}
\enda
holds $\asz$. This completes the proof.  
\end{proof}

\begin{proof}[Proof of Corollary \ref{cor:rem}]
Wider or equal asymptotic central quantile ranges imply greater or equal asymptotic variance. 
A test statistic is thus not proper if the asymptotic variance of its randomization distribution is not  greater than or equal to that of its sampling distribution for all $\mathcal{S}$.

For the unadjusted $\htn$, $\htns $, and $\htn/\tse_\neyman$, that $v_\lin/v_\neyman$ can be either greater or less than $v_{\lin0}/v_{\neyman0}$ suggests $\rho _\neyman/\rho _{\neyman0}$, and thus $v (\rho _\neyman)/v(\rho _{\neyman0})$, can be either greater or less than 1. We have  
\begina
\frac{\var_\infty(\htN)}{\var_\infty(\htnpa)} \to c'_\neyman \cdot \frac{v (\rho _\neyman)}{v(\rho _{\neyman0})}, \quad 
\frac{\var_\infty(\htns )}{\var_\infty\{\htnspa\}} \to c'_\neyman \cdot \frac{v (\rho _\neyman)}{v(\rho _{\neyman0})}, \quad 
\frac{\var_\infty(\htnsr)}{\var_\infty\{\htnsrpa\}} \to c_\neyman\cdot \frac{ v(\rho _\neyman)}{ v(\rho _{\neyman0})},
\enda
have limiting values that can be either greater or less than 1. This shows none of $\htN$, $\htns $, and $\htnsr$ are proper under ReM.
Likewise for  
\begina
\frac{\var_\infty(\hts)}{\var_\infty (\htspa)} \to c'_\rosenbaum \cdot v (\rho _\rosenbaum), \quad 
\frac{\var_\infty(\htss)}{\var_\infty\{\htsspa\}} \to c'_\rosenbaum \cdot v (\rho _\rosenbaum)\quad \text{ for $*=\rosenbaum,\fisher$}
\enda
to have limiting values that can be either greater or less than 1 with $c'_\rosenbaum$ can be either greater or less than 1. This shows none of 
$\hts$ and $\htss$   are proper under ReM for $*=\rosenbaum,\fisher$. 

Further, 
it follows from \citet[][Lemma A4]{LD20} that $q_{1-\alpha/2}(\rho )   \leq q_{1-\alpha/2}(0) =   q_{1-\alpha/2}$ for any $0< \alpha <1$, where $q_{1-\alpha/2}(\rho )$ is the $(1-\alpha/2)$th quantile of $\mU (\rho )$.
This ensures  $|\mU(\rho)|$ is stochastically dominated by $|\epsilon|$ for arbitrary $0\leq \rho \leq 1$, and thus the properness of $\hts/\tse_*$ under ReM for $*=\rosenbaum,\fisher, \lin$. 
\end{proof}

\begin{proof}[Proof of Proposition \ref{thm:rem_d}]
The sampling distributions of $\hts$ follow from \citet[][Theorem 2]{LD20}. The sampling distributions of $\htss$ and $\htssr$ follow from Slutsky's theorem and \citet[][Lemma A5]{LD20}, which ensures Lemma \ref{v_lim} holds under ReM even if $\xdi  \neq x_i$. 
\end{proof}

\begin{proof}[Proof of Corollary \ref{cor:rem_d}]
The result follows from comparing the asymptotic sampling distributions in Proposition \ref{thm:rem_d} with the asymptotic randomization distributions under unrestricted \textsc{frt} in Theorems \ref{neyman}--\ref{lin}. 
The reasoning is almost identical to that in the proof of Corollary \ref{cor:rem}, with $|\mU(\rho)|$ being stochastically dominated by $|\epsilon|$ for arbitrary $0\leq \rho \leq 1$. 
\end{proof}

\section{Proofs of the results in Section \ref{sec:connections}}\label{sec:app_sp}
 
\newcommand{\sels}{(\ese'_\lin)^2}
\newcommand{\tsels}{(\tse'_\lin)^2}
\newcommand{\esel}{\ese_\lin^2}
\newcommand{\tsel}{\tse_\lin^2}
\subsection{Connection with the super-population framework}
Let $
\gamma =  p_1  \gamma_1+ p_0  \gamma_0$, $\sigma^2 = p_1 \sigma_1^2 +p_0 \sigma_0^2$, $ \sigma_a^2  = p_1 \sigma_{a(1)}^2 + p_0 \sigma_{a(0)}^2$, and $\sigma_b^2  = p_1 \sigma_{b(1)}^2  + p_0 \sigma_{b(0)}^2$ analogous to $\gamma$, $S^2$, $S_a^2$, and $S_b^2$ in the finite-population setting, respectively. 
Let $\sels$ and $\tsels$ be the unmodified classic and robust standard errors of $\htl$ under the super-population framework such that 
$
\ese_\lin^2 = \sels + \hat\theta^\T S_x^2 \hat\theta/N
$
and
$
\tse_\lin^2 = \tsels + \hat\theta^\T S_x^2 \hat\theta/N.
$

\begin{lemma}\label{v_lim_sp}Assume Condition \ref{asym_SP}. 
\begin{enumerate}[(a)]
\item\label{as_sp} Lemma \ref{as_Z} statements \eqref{as_1}--\eqref{as_3} hold $\assp$. 
\item\label{vlim_sp} Lemma \ref{v_lim} statement \eqref{vl_1} hold $\assp$ after changing $\ese_\lin^2$ and $\tse_\lin^2$ to $\sels$ and $\tsels$, respectively,  and changing $S^2$ and $S^2_*$  to $\sigma^2$ and $\sigma^2_*$ for $* = a, b, z, a(z), b(z)$, where $z = 0,1$. 

\end{enumerate}

\end{lemma}

\begin{proof}
The proof is almost identical with that of Lemmas \ref{as_Z} and \ref{v_lim} under the finite-population framework. The classical strong law of large numbers ensures  Lemma \ref{as_Z} statement \eqref{as_1} and \eqref{slln_w} hold $\assp$ under Condition \ref{asym_SP}. 
The rest of Lemma \ref{as_Z} then follows. Proposition \ref{FLRtoN} is algebraic and thus holds also under the super-population framework after replacing $x_i$ with $x_i - \bar x$. Statement \eqref{vlim_sp} then follows from statement \eqref{as_sp} as Lemma \ref{v_lim} follows from Lemma \ref{as_Z}.
\end{proof}

\begin{proof}[Proof of Theorem \ref{lin_SP}]
The result for the sampling distributions of $\hts$, where $* = \neyman, \fisher, \lin$, including $v_\lin \leq v_*$ for $* = \neyman, \fisher$, follows from \citet{negi2020revisiting}. 
The asymptotic equivalence of $\rtn  \htr$ and $\rtn  \htf$ follows from 
$
\rtn (\htf - \htr) = \rtn\htx^\T (\hgf - \hgr) = \ops
$
by Proposition \ref{FLRtoN}, Lemma \ref{v_lim_sp}, and Slutsky's theorem. 

The result for  $\htss$ and $\htssr$ follows from Lemma \ref{v_lim_sp} \eqref{vlim_sp} and Slutsky's theorem via the same reasoning as in the proof of Theorems \ref{neyman}--\ref{lin}.
In particular, recall  $\hg_{\lin,z} = (\hat S^2_{x(z)})^{-1}\hat S_{xY(z)}$ as the coefficient of $x_i$ from the \textsc{ols} fit of $Y_i$ on $(1,x_i)$ with units in treatment group $z$. 
Algebraically, we have $\hat \theta = \hg_{\lin,1}-\hg_{\lin,0}$ such that $\hat\theta = \gamma_1-\gamma_0 +\oo \ \assp$ under 
Condition \ref{asym_SP} with $ \hg_{\lin,z} = \gamma_z +\oo \ \assp$.  
This, together with $S_x^2 = \sigma_x^2 + \oo \ \assp$, ensures $\hat\Delta_{\bar{x}}  = \Delta_{\bar{x}} + \oo \ \assp$ such that 
$N\esel = p_0^{-1}  \sigma_{b(1)}^2  + p_1^{-1}  \sigma_{b(0)}^2 + \Delta_{\bar{x}} + \oo\ \assp $ and 
$N\tsel = v_{\lin}+ \oo\ \assp$ by Lemma \ref{v_lim_sp}\eqref{vlim_sp}.

The result for the randomization distributions follows from Lemma \ref{v_lim_sp} \eqref{as_sp} and the sampling distribution results in Theorems \ref{neyman}--\ref{lin} via almost identical reasoning as the proof of the randomization distribution results in Theorems  \ref{neyman}--\ref{lin}. 

In particular, 
Lemma \ref{v_lim_sp} \eqref{as_sp} ensures the sequence of $\{(Y_i,Y_i,x_i)\}_{i=1}^N$ that the {\frt} procedure takes as fixed input satisfies Condition \ref{asym} $\assp$ under Condition \ref{asym_SP}, with $\hgr$, $\hat S^2$, and $\hat S^2_{e}$ giving the analogs of $\gamma_z$, $S^2_z$, and both $S^2_{a(z)}$ and $S^2_{b(z)}$, respectively. 
This ensures the analog of $\Delta_{\bar{x}}$ for the sequence of $\{(Y_i,Y_i,x_i)\}_{i=1}^N$ equals zero such that $\hat\Delta^\pi_{\bar{x}}  =  \oo \ \assp$, and thus $ \rtn  \ese_\lin^\pi  - \rtn  (\ese_\lin')^\pi = \oo \ \assp$, and $\rtn  \tse_\lin^\pi - \rtn  (\tse_\lin')^\pi =  \oo \ \assp$. 

The result follows from  replacing the limiting values of $\gamma_z$, $S_z^2$, $S_{a(z)}^2$, and $S_{b(z)}^2$ in the sampling distributions in Theorems \ref{neyman}--\ref{lin} with those of $\hgr$, $\hat S^2$, $\hat S^2_{e}$,  and $\hat S^2_{e}$, respectively, namely $\lim_{N\to\infty} \hgr = \gamma$,  $\lim_{N\to\infty} \hat S^2  = p_1  \sigma_1^2  + p_0 \sigma_0^2+p_1p_0\tau^2$ and $\lim_{N\to\infty} \hat S_e^2  =p_1  \sigma_{a(1)}^2  + p_0 \sigma_{a(0)}^2+p_1p_0\tau^2$.

The nuance here is that we are using the sampling distributions rather than the randomization distributions in Theorems \ref{neyman}--\ref{lin} for the above reasoning. The change would be ``replacing $S^2$ and $S^2_*$  with $\sigma^2$ and $\sigma^2_*$ for $* = a, b, z, a(z), b(z)$'' if we use the randomization distributions in Theorems \ref{neyman}--\ref{lin} instead; see the proof of Theorems \ref{neyman}--\ref{lin} for the correspondence between the sampling and randomization distributions under the finite-population framework in the first place. 
\end{proof}

\subsection{Permutation tests for linear models}\label{app_ols}
Recall $\delta = (I-H)Z$ as the residual vector from the {\olsf} of $Z$ on $(\on, X)$,  with $\delta_i= Z_i - p_1 - \lambda^{-1}p_1p_0 x_i^\T\htx$ by \eqref{H}.
Let $ C = (\dsq)^{-1} =\{\zt(I-H)Z\}^{-1}$. Lemma \ref{as_Z}\eqref{as_4} ensures
\beginy\label{limc}
NC= (p_1p_0)^{-1} + \ooas .
\endy

\begin{lemma}\label{fwl} 
{\dgp} and $\pi \in\Pi$, 
\begin{enumerate}[(a)]
\item\label{betas} the coefficients have explicit forms
\begina
&&\htrp =  \frac{Z_\pi^\T(I-H)e}{Z_\pi^\T{ (I-H_1)}Z_\pi}  , \quad\quad 
\htfp= \frac{Z_\pi^\T{(I-H)}e}{Z_\pi^\T(I-H)Z_\pi},\\
&&\hbflp = \hbkp = \frac{Z^\T{(I-H)}e_\pi}{Z^\T(I-H)Z},\quad \quad \hbtbp - \htf =   \frac{Z^\T{(I-H)}\epfp}{Z^\T(I-H)Z}, 
\quad\quad
\hbmp = \frac{Z^\T{(I-H)}Y_\pi}{Z^\T(I-H)Z};
\enda 
\item\label{ese} the classic standard errors have explicit forms
\begina
&&(\seflp)^2 =  \frac{1}{N-2-J} \left\{ C \cdot   \esq - C \cdot   e_\pi^\T H e_\pi - (\hbflp)^2 \right\}, \\
&&(\sekp)^2 =  \frac{1}{N-2} \left\{ C \cdot   \esq - (\hbkp)^2\right\}, \\
&&(\setbp)^2 
=\frac{1}{N-2-J}\left\{ C \cdot   \epfsq - C \cdot  \epfp ^\T H\epfp     - (\hbtbp -  \htf)^2  \right\},\\
&&(\semp)^2 
= \frac{1}{N-2-J}\left\{ C \cdot  \ysq - C \cdot  Y_\pi^\T H Y_\pi  - (\hbmp)^2\right\},
\enda
and the robust standard errors have a unified form
$$
(\tsesp)^2  
= C^2  (\eta_*^\pi)^\T \dd  \eta_*^\pi  \quad \quad (*=\fl, \knd, \tb, \mly)
$$ 
where 
$\efl= (I-H) e_\pi - \delta^\T\hbflp$, 
$\ek= e_\pi - \delta^\T\hbkp$, $\etb= (I-H)  \epfp  - \delta^\T(\hbtbp - \htf)$, 
and $\emly= (I-H)  Y_\pi -   \delta^\T\hbmp$. 
\end{enumerate}
\end{lemma}

\begin{proof}[Proof of Lemma \ref{fwl}]
For $\htr$, let $(I-H_1) Z$ and $(I-H_1) e$ be the residual vectors from the \textsc{ols} fits of $Z$ on $1_N$ and $e$ on $1_N$, respectively. By \textsc{fwl}, $\htR$ equals the coefficient of $(I-H_1)Z$ from the \textsc{ols} fit of $(I-H_1)e$ on $ (I-H_1)Z$ by Lemma \ref{lmz}. 
This ensures
\begina
\htR =  \frac{Z^\T(I-H_1)e}{Z^\T{ (I-H_1)}Z} = \frac{Z^\T(I-H)e}{Z^\T{ (I-H_1)}Z},
\enda
where the last identity follows from $He = H_1 e = 0_J$ by \eqref{H}.  

The result for $\htF$ follows from Lemma \ref{alg_f} and $(I-H)Y = e$. 
%

The result for the Freedman--Lane procedure follows from replacing $Y$ with $Y^\pi_\fl= HY + e_\pi$ in Lemma \ref{alg_f}, with $(I-H)Y_\fl^\pi = (I-H)(HY + e_\pi) = (I-H) e_\pi$.


For the Kennedy procedure, 
let $\hb'$, $( \ese')^2$, and $( \tse')^2$ be the coefficient and standard errors from the {\olsf} of $e_{\pi}$ on $\delta$. By \eqref{H},
 $\delta = (I-H_1)\delta$ and $e_\pi=(I-H_1)e_\pi$, so $\delta$ and $e_\pi$ can also be viewed as  the residual vectors from the {\ols} fits of $\delta$ on $ 1_N$ and $e_{\pi} $ on $1_N$, respectively. 
With 
\begina
\hbkp = \hb', \quad (\sekp)^2 = \frac{N-1}{N-2} ( \se')^2, \quad 
(\tsekp)^2 = ( \tse')^2 
\enda
by \textsc{fwl}, the result follows from $ \hb' = C  \cdot  \delta^\T e_\pi$, 
$
  ( \se')^2 = (N-1)^{-1}\{ C \cdot \esq  - (\hb')^2\}$, and $  ( \tse')^2 = C^2  (e_\pi - \delta\hb' )^\T \dd  (e_\pi - \delta\hb')$ by replacing $u$ with $e_\pi$ and $v$ with $\delta$ in Lemma \ref{lmz}.

%

Recall $\hat \mu_\fisher$ and $\hgf$ as the coefficients of $1_N$ and $X$ in the {\olsf} of $Y$ on $( 1_N, Z, X)$. 
The  ter Braak procedure uses $Y_\tb^\pi= \on  \hat\mu_\fisher + Z \htf + X \hgf +\epfp $ as the synthetic outcome vector for estimating $\hbtbp$, $\setbp$, and $\tsetbp$, with 
\begina
(I-H)Y_\tb^\pi &=&(I-H)(\on  \hat\mu_\fisher+ Z \htf  + X \hgf +\epfp )  = (I-H) (Z \htf + \epfp ),\\
(Y_\tb^\pi)^\T(I-H)Y_\tb^\pi &=& (Z \htf + \epfp )^\T(I-H) (Z \htf + \epfp )\\
&=& C^{-1}   \htf^2 + 2\htf Z^\T(I-H)\epfp  +\epfp ^\T (I-H)\epfp 
\enda 
by \eqref{H}. Replace $Y$ with $Y_\tb^\pi$ in Lemma \ref{alg_f} to see
\begina
\hbtbp &=&C \cdot Z^\T{(I-H)}Y_\tb^\pi =  \htf + C \cdot Z^\T{(I-H)}\epfp , \\
(\setbp)^2 
&=& \frac{1}{N-2-J}\left\{ C \cdot (Y_\tb^\pi)^\T(I-H)Y_\tb^\pi   - (\hbtbp)^2\right\}\\
&=& \frac{1}{N-2-J}\left\{ C \cdot \ep_{\fisher,\pi}^\T(I-H)\epfp + 2\htf(\hbtbp-\htf) +  \htf^2   -  (\hbtbp)^2   \right\}\\
&=& \frac{1}{N-2-J}\left\{  C \cdot  \epfsq -  C \cdot \epfp ^\T H\epfp - (\hbtbp -  \htf)^2  \right\},\\
(\tsetbp)^2 &= &C^2 (\etb)^\T \dd  \etb.
\enda

The result for $\hbmp$, $\semp$, and $\tsemp$ follows from replacing $Y$ with $Y_\pi$ in Lemma \ref{alg_f}.
\end{proof}

 Let $\Lambda  = \dd$ and $\langle u,v\rangle = u^\T \Lambda  v$ for $u = (u_1, \dots, u_N)^\T$ and $v = (v_1, \dots, v_N)^\T$ be the corresponding inner product to simplify the presentation. The Cauchy--Schwarz inequality implies 
\beginy \label{ada} 
\langle u,u\rangle = \sumi u_i^2 \delta_i^2 \leq \left(  \|u\|_4^4 \right)^{1/2} \left( \|\delta\|_4^4 \right)^{1/2},\quad 
 \langle u,v\rangle  \leq  \langle u,u\rangle^{1/2} \langle v,v\rangle^{1/2} .
\endy

\begin{lemma}\label{qf}
Assume Condition \ref{asym} and a sequence of $Z$ that satisfies Lemma \ref{as_Z} statements \eqref{as_1}--\eqref{as_f}. 
\begin{enumerate}[(a)]
\item\label{qf_1} ${N^{-1/2}} \delta^\T e_\pi \rightsquigarrow \mN\{0, p_1p_0(S^2_a + p_1p_0\tau^2)\}, \quad {N^{-1/2}} \delta^\T \epfp\rightsquigarrow \mN(0,  p_1p_0 S^2_a),$\\
${N^{-1/2}} \delta^\T Y_\pi \rightsquigarrow \mN\{0,p_1p_0(S^2+ p_1 p_0 \tau^2)\}$.
\item\label{qf_2}
$N^{-1} X^\T e_\pi = \opp, \quad N^{-1} X^\T \epfp= \opp, \quad N^{-1} X^\T Y_\pi = \opp$. 
\item\label{qf_3}
$N^{-1} \langle e_\pi,e_\pi\rangle=   p_1p_0( S_a^2 + p_1p_0\tau^2) + \opp,\quad 
  N^{-1}\langle \epfp,\epfp\rangle =    p_1p_0 S_a^2 + \opp,$\\
  $
N^{-1} \langle Y_\pi - \on \hY, Y_\pi - \on \hY\rangle =   p_1p_0  (S^2  + p_1p_0\tau^2)  + \opp.
$
\end{enumerate}
\end{lemma}

\begin{proof}

Statement \eqref{qf_1} follows from Lemma \ref{romanoS33}\eqref{clt_2} by letting $u = \delta$ and $v = e, \epf, Y$, respectively, with the respective means, variances, and bounded fourth moments following from Lemma \ref{as_Z}.

For statement \eqref{qf_2}, $N^{-1} X_j^\T e_\pi = \opp$ follows from Markov's inequality given $E(N^{-1} X_j^\T e_\pi) = 0$ and $\var(N^{-1} X_j^\T e_\pi) =  N^{-1}\lambda   \hs^2_e = o(1)$ by Lemmas \ref{romanoS33}\eqref{clt_1} and \ref{as_Z}\eqref{as_3}. 
The proof for the rest of statement \eqref{qf_2} is almost identical and thus omitted.

For statement \eqref{qf_3}, with a slight abuse of notation, let $\delta^2 = (\delta_1^2, \dots, \delta_N^2)^\T$,  $e^2 = (e_1^2, \dots, e_N^2)^\T$,  $\epf^2 = (\ep_{\fisher,1}^2, \dots, \ep_{\fisher,N}^2)^\T$, and 
$\nu = (\nu_1, \dots, \nu_N)^\T$, where $\nu_i = (Y_i - \hY)^2$,  to write  
\begina
&&N^{-1}\langle e_\pi,e_\pi\rangle = N^{-1} \sumi \delta_i^2 e^2_{\pi(i)} = N^{-1} (\delta^2)^\T ( e^2)_\pi,\quad
N^{-1}\langle \epfp,\epfp\rangle  = N^{-1} (\delta^2)^\T ( \epf^2)_\pi,\quad\\
&&N^{-1}\langle Y_\pi - \on \hY, Y_\pi - \on \hY\rangle = N^{-1} (\delta^2)^\T v_\pi.
\enda
With $\hat{\delta^2} = N^{-1} \|\delta\|^2_2 = p_1p_0 + o(1)$, $\hat{e^2} = N^{-1} \|e\|_2^2 = S_a^2 +p_1p_0\tau^2 + o(1)$  and  
$S^2_{\delta^2} = (N-1)^{-1}\sumi (\delta_i^2 - \hat{\delta^2} )^2 = O(1)$, $S^2_{e^2} = (N-1)^{-1}\sumi (e_i^2 - \hat{e^2} )^2 = O(1)$ by Lemma \ref{as_Z}\eqref{as_3}--\eqref{as_4}, 
 Lemma \ref{romanoS33}\eqref{clt_1} ensures $
E\left\{N^{-1}(\delta^2)^\T( e^2)_\pi\right\} = \hat{\delta^2}\hat{e^2} =   p_1 (S_a^2 + p_1p_0\tau^2)+ o(1)$ and $
\var\left\{N^{-1}(\delta^2)^\T( e^2)_\pi\right\} = N^{-1}\lambda S^2_{\delta^2} S^2_{e^2}=o(1)$, which, coupled with Markov's inequality, imply
$$
N^{-1}\langle e_\pi,e_\pi\rangle= E\{N^{-1}(\delta^2)^\T( e^2)_\pi\} + \opp=  p_1p_0 ( S_a^2 + p_1p_0\tau^2) + \opp . 
$$

The proof for the rest of statement \eqref{qf_3} is almost identical, with 
 $\hat{\ep_\fisher^2} = N^{-1} \|\epf\|^2_2 =S_a^2 + o(1)$, $S^2_{\ep^2} = (N-1)\sumi (\ep_{\fisher,i}^2 -\hat{\ep_\fisher^2})^2 = O(1)$, 
  $\hat \nu = \lambda  \hs^2 = S^2  + p_1p_0\tau^2 + o(1)$, and $S_\nu = (N-1)^{-1}\sumi (v_i - \hat\nu)^2 = O(1)$ by Lemma \ref{as_Z} statements \eqref{as_2} and \eqref{as_f}.
\end{proof}

\citet{fl} and \citet{ad} sketched proofs for the reference distributions of the statistics studentized by the classic standard errors. Below we give a unified proof for the statistics studentized by both the classic and robust standard errors.

\begin{proof}[Proof of Theorem \ref{nay_fl}]
We first verify the asymptotic normality of $\rtn  \hbsp$ for $* = \fl, \knd, \tb,\mly$, and then prove the result on the studentized variants based on it.
\paragraph{Unstudentized coefficients}
 Lemma \ref{fwl}\eqref{betas} ensures 
$
\rtn  \hbflp = \rtn  \hbkp   
= (NC) \cdot  N^{-1/2} \delta^\T e_\pi 
$,
$
\rtn  (\hbtbp - \htf) =  (NC) \cdot N^{-1/2} \delta^\T \epfp 
$
and
$
\rtn  \hbmp =  (NC)  \cdot N^{-1/2} \delta^\T Y_\pi 
$,
where 
$N C= (p_1p_0)^{-1} + o(1) \ \asz$ by \eqref{limc}. 
The asymptotic normality of $  \hbsp$ follows from Lemma \ref{qf}\eqref{qf_1} by Slutsky's theorem.
This ensures  
\beginy\label{as_b}
\hbsp = \opp \quad \asz \quad \text{for $* = \fl, \knd, \mly$}, \quad  \hbtbp - \htf =\opp \quad \asz.
 \endy 

\paragraph{Studentized coefficients}
The result for $\hbssp$ and $\hbssrp$ follows from Slutsky's theorem and 
\beginy\label{pt_se}
\begin{array}{lll}
N (\seflp)^2 = (p_1p_0)^{-1} S_a^2 + \tau^2 + \opp, && N( \tseflp)^2 = (p_1p_0)^{-1} S_a^2+ \tau^2 + \opp, \\
N (\sekp)^2 =(p_1p_0)^{-1} S_a^2 + \tau^2 + \opp, && N( \tsekp)^2 = (p_1p_0)^{-1} S_a^2+ \tau^2  + \opp, \\
N (\setbp)^2 = (p_1p_0)^{-1} S_a^2+ \opp, && N( \tsetbp)^2 = (p_1p_0)^{-1} S_a^2 + \opp,\\
N (\semp)^2 = (p_1p_0)^{-1} S^2 + \tau^2+ \opp, && N( \tsemp)^2 = (p_1p_0)^{-1} S^2  + \tau^2+ \opp
\end{array}
\endy
hold $\asz$.

We finish the proof by verifying \eqref{pt_se}. 
Lemma  \ref{as_Z}  ensures  it suffices to focus on sequences of $Z$ that satisfy  Lemma  \ref{as_Z} statements \eqref{as_1}--\eqref{as_f} and \eqref{as_b}. 
Fix one such sequence with $N c = (p_1p_0)^{-1} + \oo$ by \eqref{limc}. 

For the classic standard errors $\ese^\pi_*$, with $N C = (p_1p_0)^{-1} + \oo$, $\hbsp = \opp$ by \eqref{as_b}, and 
\begina
N^{-1}\esq = S_a^2 + p_1p_0\tau^2 + \oo, \quad N^{-1}\epfsq = S_a^2  + \oo, \quad N^{-1} \|Y\|_2^2 - \hY^2 = S^2 + p_1p_0\tau^2 + \oo
\enda by Lemma \ref{as_Z}, a direct comparison of the expressions of $(\ese^\pi_*)^2$ in Lemma \ref{fwl}\eqref{ese} with \eqref{pt_se} suggests it suffices to verify
\beginy\label{ehe}
N^{-1} e_\pi^\T H e_\pi  =\opp,\quad N^{-1} \epfp^\T H \epfp   =\opp, \quad N^{-1} Y_\pi^\T H Y_\pi  -\hat Y^2 = \opp. 
\endy
Recall $H = N^{-1}\on \ont + (N-1)^{-1}X X^\T$  from \eqref{H} to write  
\beginy\label{he}
He_\pi =\lambda^{-1}N^{-1}X X^\T e_\pi, \quad H\epfp =\lambda^{-1}N^{-1}X X^\T \epfp, \quad H Y_\pi = \on\hY + \lambda^{-1}N^{-1}X X^\T Y_\pi.
\endy 
By Lemma \ref{qf}\eqref{qf_2}, 
$
N^{-1} e_\pi^\T H e_\pi  = \lambda^{-1}(N^{-1} e_\pi^\T X) ( N^{-1} X^\T e_\pi) =\opp 
$
and likewise $N^{-1} \epfp^\T H \epfp   =\opp$ and $ N^{-1} Y_\pi^\T H Y_\pi  =\hat Y^2 + \opp$. 
This verifies \eqref{ehe} and thus the result for $\sesp$. 

For the robust standard errors, with
$
N(\tsesp)^2  
= (NC)^2 \cdot N^{-1} \langle\eta_*^\pi, \eta_*^\pi\rangle$ by Lemma \ref{fwl}\eqref{ese}, where  $N C  = (p_1p_0)^{-1} + \oo$,
it suffices to determine the probability limits of $ N^{-1} \langle\eta_*^\pi, \eta_*^\pi\rangle$ for 
\begina
\begin{array}{ll}
\efl= (I-H) e_\pi - \delta^\T \hbflp, &\quad \ek= e_\pi - \delta^\T\hbkp,\\
\etb= (I-H)  \epfp - \delta^\T(\hbtbp - \htf),
&\quad\emly= (I-H)  Y_\pi -  \delta \hbmp = (I-H)  (Y-\on\hat Y)_\pi  -  \delta \hbmp.
\end{array}
\enda
Lemma \ref{as_Z} ensures 
\beginy\label{eq:bd}
\|Y\|_4^4 = O(N), \quad\|X_j\|_4^4 = O(N), \quad \|e\|_4^4= O(N), \quad \|\delta\|_4^4=O(N),\quad  \|\epf\|_4^4 = O(N). 
\endy

For $\ek =e_\pi - \delta\hbkp$,   \eqref{ada}--\eqref{eq:bd} ensures $N^{-1}\langle\delta ,   \delta\rangle = O(1)$ and  $N^{-1} \langle e_\pi, \delta \rangle= O(1)$  such that 
\beginy\label{tse_knd}
N^{-1}\langle \ek,\ek \rangle 
&=&  N^{-1}\langle e_\pi ,   e_\pi\rangle -  2N^{-1}\langle e_\pi ,   \delta\rangle \cdot \hbkp  + N^{-1}\langle\delta ,   \delta\rangle \cdot  (\hbkp)^2\nonumber \\
&=& N^{-1}\langle e_\pi ,   e_\pi\rangle + o(1)  
= p_1p_0(S_a^2 + p_1 p_0  \tau^2) + \opp
\endy
by \eqref{as_b} and $N^{-1}\langle e_\pi ,   e_\pi\rangle=p_1p_0(S_a^2 + p_1 p_0  \tau^2)$ from Lemma \ref{qf}\eqref{qf_3}. This verifies the result for $\tsekp$. 

For $\efl = e_\pi - \delta\hbflp - He_\pi$, \eqref{he} ensures
\begina
N^{-1} \langle He_\pi, He_\pi\rangle 
= \lambda^{-2}   (N^{-1}X^\T e_\pi)^\T (N^{-1}X^\T \Lambda X) (N^{-1}X^\T e_\pi)
 = \opp, 
\enda
where the last identity follows from 
$N^{-1}X^\T e_\pi=\opp$  by Lemma \ref{qf}\eqref{qf_2} and $N^{-1}X^\T \Lambda X = O(1)$ by $N^{-1}\|\delta\|_4^4 = O(1)$ and $N^{-1}\|X_j\|_4^4 = O(1)$ from \eqref{eq:bd}. 
This, together with 
\begina
N^{-1} \langle e_\pi - \delta\hbflp,e_\pi - \delta\hbflp\rangle 
=N^{-1}  \langle e_\pi ,   e_\pi\rangle + o(1) = p_1p_0 (S_a^2 + p_1 p_0  \tau^2) + \opp
\enda by the same reasoning as \eqref{tse_knd}, ensures
\begina
N^{-1} \langle e_\pi - \delta\hbflp, He_\pi\rangle  \leq  
\Big(N^{-1}  \langle e_\pi - \delta\hbflp,e_\pi - \delta\hbflp\rangle \Big)^{1/2} 
\Big( N^{-1}  \langle He_\pi, He_\pi\rangle  \Big)^{1/2}= \opp
\enda
by \eqref{ada} and thus 
\begina
N^{-1} \langle \efl,\efl \rangle  
&=& N^{-1}\langle e_\pi - \delta\hbflp - He_\pi,e_\pi - \delta\hbflp - He_\pi\rangle\\
&=& N^{-1}\langle e_\pi - \delta\hbflp,e_\pi - \delta\hbflp\rangle  +N^{-1}\langle He_\pi, He_\pi\rangle   -  2N^{-1}\langle e_\pi - \delta\hbflp, He_\pi\rangle \\
&=& p_1p_0 (S_a^2 + p_1 p_0  \tau^2) + \opp.
\enda
This verifies the result for $\tseflp$. 

The result for $\tsetbp$ and $\tsemp$ follows from almost identical reasoning after replacing $e_\pi$ with $\epfp$ and $Y_\pi - \on \hat Y$ in the proof of $\tseflp$. 
\end{proof}

\end{document}

%% file: cmdlist_frt_mar.tex
\usepackage{xcolor}

\def\nm{\neyman}
\def\orange{\color{orange}}

\def\pso{pseudo-outcome}
\def\mo{model-output}
\newcommand{\dgp}{For  $(Y_i, x_i, Z_i)_{i=1}^N$ from arbitrary data generating process with $\bar x=0_J$,} 
\def\prob{P}

\def\begini{\begin{itemize}}
\def\endi{\end{itemize}}

\newcommand{\ols}{\textsc{ols}}
\newcommand{\oo}{o(1)}
\newcommand{\asz}{P_Z\text{-}\as}
\newcommand{\asp}{P_\pi\text{-}\as}
\newcommand{\assp}{P_\textsc{s}\text{-}\as}
\newcommand{\aszp}{P_{(Z,\pi)}\text{-}\as}
\newcommand{\ooas}{\oo \ \asz}

\newcommand{\ooaszp}{\oo \ \aszp}

\newcommand{\xdi}{d_i}

\newcommand{\htd}{\hat\tau_d}

\newcommand{\olsf}{{\textsc{ols} fit}}
\newcommand{\ass}{Lemma \ref{as_Z} ensures that it suffices to verify the result for sequences of $Z$ that satisfy Lemma \ref{as_Z} statements \eqref{as_1}--\eqref{as_f}. 
Fix one such sequence for the rest of the proof. }

\newcommand{\asym}{Assume  Condition \ref{asym} and complete randomization.}

\newcommand{\assrem}{Assume Condition \ref{asym}, ReM in design, and $p_{\frt,\ma}$ in \eqref{pfrt_a} in analysis.}
\newcommand{\htxp}{\htau_{x}^\pi}

\newcommand{\htns}{\htn/\sen}
\newcommand{\htnsr}{\htn/\tsen}

\newcommand{\tsen}{\tse_\neyman}
\newcommand{\htnpa}{\htn^{\pi|\ma}}
\newcommand{\htnspa}{(\htns)^{\pi|\ma}}
\newcommand{\htnsrpa}{(\htnsr)^{\pi|\ma}}

 \newcommand{\htsspa}{(\htss)^{\pi|\ma}}

\newcommand{\htspa}{\hts^{\pi|\ma}}

\newcommand{\htnp}{\htn^\pi}
 \newcommand{\htrp}{\htr^\pi }
\newcommand{\htlp}{\htl^\pi}
\newcommand{\htfp}{\htf^\pi}
 \newcommand{\hts}{ \htau_{*}}
\newcommand{\htsp}{ \htau^\pi_*}
\newcommand{\htss}{\htau_* /\ese_*}
\newcommand{\htssr}{\htau_* /\tse_*}

\newcommand{\tsef}{{\tse_\fisher}}

\newcommand{\sen}{{\ese_\neyman}}

\newcommand{\sef}{{\ese_\fisher}}

\newcommand{\senp}{{\ese_\neyman^\pi}}

\newcommand{\tsenp}{{\tse_\neyman^\pi}}

\newcommand{\usq}{\|u\|^2_2}
\newcommand{\vsq}{\|v\|^2_2}

\newcommand{\dsq}{\|\delta\|^2_2}
\newcommand{\esq}{\|e\|^2_2}
\newcommand{\epfsq}{\|\epf\|^2_2}
\newcommand{\ysq}{\|Y\|^2_2}

\newcommand{\zt}{Z^\T}

\newcommand{\on}{1_N}
\newcommand{\ont}{1_N^\T}

\newcommand{\epf}{\epsilon_{\fisher}}
\newcommand{\epfp}{\epsilon_{\fisher,\pi}}

\newcommand{\epfi}{\epsilon_{\fisher,i}}

\newcommand{\hsrz}{\hat S^2_{\rosenbaum(z)}}
\newcommand{\hsfz}{\hat S^2_{\fisher(z)}}
\newcommand{\hslz}{\hat S^2_{\lin(z)}}
\newcommand{\hssz}{\hat S^2_{*(z)}}
\newcommand{\hsnz}{\hat S^2_{\neyman(z)}}

\newcommand{\hsszp}{(\hat S^2_{*(z)})^\pi}

\newcommand{\efl}{\eta_\fl^\pi}
\newcommand{\ek}{\eta_\knd^\pi}
\newcommand{\emly}{\eta_\mly^\pi}

\newcommand{\etb}{\eta_\tb^\pi}

\newcommand{\op}{{ o_P(1)}}
\newcommand{\opz}{o_{P,Z}(1)}
\newcommand{\opp}{o_{P,\pi}(1)}
\newcommand{\ops}{o_{P,\textsc{s}}(1)}
\newcommand{\opss}{o_{P,*}(1)}

\newcommand{\hgr}{\hat\gamma_\rosenbaum}
\newcommand{\hgf}{\hat\gamma_\fisher}
\newcommand{\hglo}{\hat\gamma_{\lin,1}}
\newcommand{\hglz}{\hat\gamma_{\lin,0}}
\newcommand{\hglzz}{\hat\gamma_{\lin,z}}

\newcommand{\htf}{\hat\tau_\fisher}
\newcommand{\htl}{\hat\tau_\lin}

\newcommand{\htr}{\hat\tau_\rosenbaum}
\newcommand{\htn}{\hat\tau_\neyman}

\newcommand{\hbflp}{ \hb_\fl^\pi}
\newcommand{\seflp}{ \ese_\fl^\pi}
\newcommand{\tseflp}{ \tse_\fl^\pi}

\newcommand{\hbsp}{ \hb_*^\pi}
\newcommand{\sesp}{ \ese_*^\pi}
\newcommand{\tsesp}{ \tse_*^\pi}

\newcommand{\hbssp}{ \hbsp /\sesp }

\newcommand{\hbssrp}{ \hbsp /\tsesp }

\newcommand{\hbkp}{\hb_\knd^\pi}

\newcommand{\sekp}{ \ese_\knd^\pi}
\newcommand{\tsekp}{ \tse_\knd^\pi}

\newcommand{\hbmp}{ \hb_\mly^\pi}
\newcommand{\semp}{ \ese_\mly^\pi}
\newcommand{\tsemp}{ \tse_\mly^\pi}

\newcommand{\hbtbp}{ \hb_\tb^\pi}
\newcommand{\setbp}{ \ese_\tb^\pi}
\newcommand{\tsetbp}{ \tse_\tb^\pi}

\newcommand{\fl}{{\textsc{fl}}}

\newcommand{\knd}{{\textsc{k}}}
\newcommand{\mly}{{\textsc{m}}}
\newcommand{\tb}{{\textsc{tb}}}

\newcommand{\qxy}{\hat S_{xY}}
\newcommand{\sxx}{S^2_{x}}
\newcommand{\hs}{\hat S}
\newcommand{\ep}{\epsilon}
 
 \newcommand{\leqst}{\leq_\textup{st}}

\newcommand{\dsim}{\overset{.}{\sim}}
\newcommand{\wa}[1]{\renewcommand{\arraystretch}{#1}}
\newcommand{\ph}{\ (* = \neyman, \rosenbaum, \fisher, \lin)}
\newcommand{\phs}{* = \neyman, \rosenbaum, \fisher, \lin}

\newcommand{\sumi}{\sum_{i=1}^N}
\newcommand{\sumpi}{\sum_{\pi\in\Pi}}
\newcommand{\sx}{S_x^2}
\newcommand{\mz}{\mathcal{Z}}

\def\T{{ \mathrm{\scriptscriptstyle T} }}

\def\sumM{\sum_{i=1}^M}

\def\se{\textup{se}}
\def\ese{\hat{\textup{se}}}
\def\tse{\tilde{\textup{se}}}

\def\neyman{{\textsc{n}}}
\def\fisher{{\textsc{f}}}
\def\lin{{\textsc{l}}}
\def\rosenbaum{{\textsc{r}}}

\newcommand{\pr}{{\rm pr}}

 \newcommand{\HN}{H_{0\neyman}}
 \newcommand{\HF}{H_{0\fisher}}

\newcommand{\bY}{\bar Y}

\newcommand{\hmu}{\hat{\mu}}
\newcommand{\hb}{\hat{\beta}}
\newcommand{\hg}{\hat\gamma}

\newcommand{\he}{\hat e}

\newcommand{\htau}{\hat \tau}

\newcommand{\htx}{\hat \tau_x}
\newcommand{\htL}{\hat \tau_\lin}
\newcommand{\htR}{\hat \tau_\rosenbaum}
\newcommand{\htN}{\hat \tau_\neyman}
\newcommand{\htH}{\hat \tau_*}
\newcommand{\htF}{\hat \tau_\fisher}

 \newcommand{\mM}{\mathcal{A}}
  \newcommand{\mMc}{\mathcal{A}^c}

 \newcommand{\ma}{\mathcal{A}}
  \newcommand{\mU}{\mathcal{U}}

\newcommand{\mL}{\mathcal{L}}

\newcommand{\ei}{e_i}
\newcommand{\e}{e}

\newcommand{\Yi}{Y_i}
\newcommand{\Zi}{Z_i}
\newcommand{\hY}{\hat Y}

\newcommand{\mZ}{\mathcal{Z}}
\newcommand{\mN}{\mathcal{N}}
\newcommand{\frt}{\textsc{frt}}
\newcommand{\sumN}{\sum_{i=1}^N}
\newcommand{\limN}{\lmt{N}}

\newcommand{\as}{{ \text{a.s.}}}
 
\newcommand{\ind}{\overset{\textup{d}}{\to}}

\newcommand{\lmt}[1]{\lim_{#1 \to \infty}}

\newcommand{\ot}[1]{1, \ldots,#1}

\DeclareMathOperator{\var}{var}

\DeclareMathOperator{\cov}{cov}

\DeclareMathOperator{\diag}{diag}

\DeclareMathOperator{\Unif}{Unif}
\DeclareMathOperator{\unif}{Unif}

\allowdisplaybreaks

    \newcommand{\ki}{{[k]i}}

\newcommand{\sumj}{\sum_{k=1}^K}
\newcommand{\sumji}{\sum_{i=1}^{\nj}}

  \newcommand{\nj}{N_{[k]}}

      \newcommand{\tauj}{\tau_{[k]}}

\newcommand{\pia}{\mz_a}

\def\beginy{\begin{eqnarray}}
\def\endy{\end{eqnarray}}
\def\begina{\begin{eqnarray*}}
\def\enda{\end{eqnarray*}}
\def\begine{\begin{enumerate}}
\def\ende{\end{enumerate}}

\usepackage{enumerate} 

%% file: x-frt_joe_r2.bbl
\begin{thebibliography}{96}
\providecommand{\natexlab}[1]{#1}
\providecommand{\url}[1]{\texttt{#1}}
\expandafter\ifx\csname urlstyle\endcsname\relax
  \providecommand{\doi}[1]{doi: #1}\else
  \providecommand{\doi}{doi: \begingroup \urlstyle{rm}\Url}\fi

\bibitem[Anderson and Legendre(1999)]{anderson1999empirical}
M.~J. Anderson and P.~Legendre.
\newblock An empirical comparison of permutation methods for tests of partial
  regression coefficients in a linear model.
\newblock \emph{Journal of Statistical Computation and Simulation},
  62:\penalty0 271--303, 1999.

\bibitem[Anderson and Robinson(2001)]{ad}
M.~J. Anderson and J.~Robinson.
\newblock Permutation tests for linear models.
\newblock \emph{Australian and New Zealand Journal of Statistics}, 43:\penalty0
  75--88, 2001.

\bibitem[Angrist and Pischke(2009)]{AngristEcon}
J.~D. Angrist and J.-S. Pischke.
\newblock \emph{Mostly Harmless Econometrics}.
\newblock Princeton University Press, 2009.

\bibitem[Aronow et~al.(2014)Aronow, Green, and Lee]{agl}
P.~Aronow, D.~Green, and D.~Lee.
\newblock Sharp bounds on the variance in randomized experiments.
\newblock \emph{Annals of Statistics}, 42:\penalty0 850--871, 2014.

\bibitem[Athey et~al.(2018)Athey, Eckles, and Imbens]{athey2018exact}
S.~Athey, D.~Eckles, and G.~W. Imbens.
\newblock Exact p-values for network interference.
\newblock \emph{Journal of the American Statistical Association}, 113:\penalty0
  230--240, 2018.

\bibitem[Bahadur and Savage(1956)]{bahadur1956nonexistence}
R.~R. Bahadur and L.~J. Savage.
\newblock The nonexistence of certain statistical procedures in nonparametric
  problems.
\newblock \emph{Annals of Mathematical Statistics}, 27:\penalty0 1115--1122,
  1956.

\bibitem[Bai et~al.(2021)Bai, Shaikh, and Romano]{bai2020}
Y.~H. Bai, A.~M. Shaikh, and J.~P. Romano.
\newblock Inference in experiments with matched pairs.
\newblock \emph{Journal of the American Statistical Association}, page to
  appear, 2021.

\bibitem[Banerjee et~al.(2020)Banerjee, Chassang, Montero, and
  Snowberg]{banerjee}
A.~V. Banerjee, S.~Chassang, S.~Montero, and E.~Snowberg.
\newblock A theory of experimenters: Robustness, randomization, and balance.
\newblock \emph{American Economic Review}, 110:\penalty0 1206--1230, 2020.

\bibitem[Basse et~al.(2019)Basse, Ding, Feller, and Toulis]{Basse2019}
G.~Basse, P.~Ding, A.~Feller, and P.~Toulis.
\newblock Randomization tests for peer effects in group formation experiments.
\newblock \emph{arXiv}, page 1904.02308, 2019.

\bibitem[Berk et~al.(2013)Berk, Pitkin, Brown, Buja, George, and
  Zhao]{berk2013covariance}
R.~Berk, E.~Pitkin, L.~Brown, A.~Buja, E.~George, and L.~Zhao.
\newblock Covariance adjustments for the analysis of randomized field
  experiments.
\newblock \emph{Evaluation Review}, 37:\penalty0 170--196, 2013.

\bibitem[Bind and Rubin(2020)]{bind2020possible}
M.~A.~C. Bind and D.~B. Rubin.
\newblock {When possible, report a Fisher-exact P value and display its
  underlying null randomization distribution}.
\newblock \emph{Proceedings of the National Academy of Sciences of the United
  States of America}, 117:\penalty0 19151--19158, 2020.

\bibitem[Bloniarz et~al.(2016)Bloniarz, Liu, Zhang, Sekhon, and Yu]{LassoTE16}
A.~Bloniarz, H.~Liu, C.~Zhang, J.~Sekhon, and B.~Yu.
\newblock Lasso adjustments of treatment effect estimates in randomized
  experiments.
\newblock \emph{Proceedings of the National Academy of Sciences of the United
  States of America}, 113:\penalty0 7383--7390, 2016.

\bibitem[Brillinger et~al.(1978)Brillinger, Jones, and
  Tukey]{brillinger1978management}
D.~R. Brillinger, L.~V. Jones, and J.~W. Tukey.
\newblock The management of weather resources.
\newblock Technical report, US Government Printing Office, Washington, DC,
  1978.

\bibitem[Bruhn and McKenzie(2009)]{bm}
M.~Bruhn and D.~McKenzie.
\newblock In pursuit of balance: Randomization in practice in development field
  experiments.
\newblock \emph{American Economic Journal: Applied Economics}, 1:\penalty0
  200--232, 2009.

\bibitem[Bugni et~al.(2018)Bugni, Canay, and Shaikh]{bugni2018inference}
F.~A. Bugni, I.~A. Canay, and A.~M. Shaikh.
\newblock Inference under covariate-adaptive randomization.
\newblock \emph{Journal of the American Statistical Association}, 113:\penalty0
  1784--1796, 2018.

\bibitem[Bugni et~al.(2019)Bugni, Canay, and Shaikh]{bugni2019}
F.~A. Bugni, I.~A. Canay, and A.~M. Shaikh.
\newblock Inference under covariate-adaptive randomization with multiple
  treatments.
\newblock \emph{Quantitative Economics}, 10:\penalty0 1747--1785, 2019.

\bibitem[Canay et~al.(2017)Canay, Romano, and Shaikh]{canay2017randomization}
I.~A. Canay, J.~P. Romano, and A.~M. Shaikh.
\newblock Randomization tests under an approximate symmetry assumption.
\newblock \emph{Econometrica}, 85:\penalty0 1013--1030, 2017.

\bibitem[Cattaneo et~al.(2015)Cattaneo, Frandsen, and
  Titiunik]{cattaneo2015randomization}
M.~D. Cattaneo, B.~R. Frandsen, and R.~Titiunik.
\newblock {Randomization inference in the regression discontinuity design: An
  application to party advantages in the US Senate}.
\newblock \emph{Journal of Causal Inference}, 3:\penalty0 1--24, 2015.

\bibitem[Chen et~al.(2020)Chen, Liu, Ma, and Zhang]{chen2020efficient}
X.~Chen, Y.~Liu, S.~Ma, and Z.~Zhang.
\newblock Efficient estimation of general treatment effects using neural
  networks with a diverging number of confounders.
\newblock \emph{arXiv preprint arXiv:2009.07055}, 2020.

\bibitem[Chong et~al.(2016)Chong, Cohen, Field, Nakasone, and
  Torero]{chong2016iron}
A.~Chong, I.~Cohen, E.~Field, E.~Nakasone, and M.~Torero.
\newblock {Iron deficiency and schooling attainment in Peru}.
\newblock \emph{American Economic Journal: Applied Economics}, 8:\penalty0
  222--55, 2016.

\bibitem[Chung and Romano(2013)]{Romano13}
E.~Chung and J.~P. Romano.
\newblock Exact and asymptotically robust permutation tests.
\newblock \emph{Annals of Statistics}, 41:\penalty0 484--507, 2013.

\bibitem[Chung and Romano(2016)]{chung2016asymptotically}
E.~Chung and J.~P. Romano.
\newblock {Asymptotically valid and exact permutation tests based on two-sample
  $U$-statistics}.
\newblock \emph{Journal of Statistical Planning and Inference}, 168:\penalty0
  97--105, 2016.

\bibitem[Cohen and Fogarty(2020)]{colin}
P.~L. Cohen and C.~B. Fogarty.
\newblock Gaussian prepivoting for finite population causal inference.
\newblock \emph{https://arxiv.org/abs/2002.06654}, 2020.

\bibitem[Cox(1982)]{cox:1982}
D.~R. Cox.
\newblock Randomization and concomitant variables in the design of experiments.
\newblock In P.~R.~Krishnaiah G.~Kallianpur and J.~K. Ghosh, editors,
  \emph{Statistics and Probability: Essays in Honor of C. R. Rao}, pages
  197--202. North-Holland, Amsterdam, 1982.

\bibitem[Dasgupta et~al.(2015)Dasgupta, Pillai, and Rubin]{DasFact15}
T.~Dasgupta, N.~Pillai, and D.~B. Rubin.
\newblock Causal inference from $2^{K}$ factorial designs by using potential
  outcomes.
\newblock \emph{Journal of the Royal Statistical Society, Series B (Statistical
  Methodology)}, 77:\penalty0 727--753, 2015.

\bibitem[DiCiccio and Romano(2017)]{romano}
C.~J. DiCiccio and J.~P. Romano.
\newblock Robust permutation tests for correlation and regression coefficients.
\newblock \emph{Journal of the American Statistical Association}, 112:\penalty0
  1211--1220, 2017.

\bibitem[Ding(2020)]{fwl20}
P.~Ding.
\newblock {The Frisch--Waugh--Lovell} theorem for standard errors.
\newblock \emph{Statistics and Probability Letters}, 168:\penalty0 108945,
  2020.

\bibitem[Ding and Dasgupta(2018)]{DD18}
P.~Ding and T.~Dasgupta.
\newblock A randomization-based perspective of analysis of variance: a test
  statistic robust to treatment effect heterogeneity.
\newblock \emph{Biometrika}, 105:\penalty0 45--56, 2018.

\bibitem[Draper and Stoneman(1966)]{draper}
N.~R. Draper and D.~M. Stoneman.
\newblock Testing for the inclusion of variables in linear regression by a
  randomisation technique.
\newblock \emph{Technometrics}, 8:\penalty0 695--699, 1966.

\bibitem[Eicker(1967)]{eicker1967limit}
F.~Eicker.
\newblock Limit theorems for regressions with unequal and dependent errors.
\newblock In \emph{Proceedings of the fifth Berkeley symposium on mathematical
  statistics and probability}, volume~1, pages 59--82. Berkeley, CA: University
  of California Press, 1967.

\bibitem[Farrell et~al.(2021)Farrell, Liang, and Misra]{farrell2020deep}
M.~H. Farrell, T.~Liang, and S.~Misra.
\newblock Deep neural networks for estimation and inference.
\newblock \emph{Econometrica}, 89:\penalty0 181--213, 2021.

\bibitem[Fisher(1935)]{Fisher35}
R.~A. Fisher.
\newblock \emph{The Design of Experiments}.
\newblock Edinburgh, London: Oliver and Boyd, 1st edition, 1935.

\bibitem[Fogarty(2018{\natexlab{a}})]{fogarty2018mitigating}
C.~B. Fogarty.
\newblock On mitigating the analytical limitations of finely stratified
  experiments.
\newblock \emph{Journal of the Royal Statistical Society: Series B (Statistical
  Methodology)}, 80:\penalty0 1035--1056, 2018{\natexlab{a}}.

\bibitem[Fogarty(2018{\natexlab{b}})]{fogarty2018regression}
C.~B. Fogarty.
\newblock Regression assisted inference for the average treatment effect in
  paired experiments.
\newblock \emph{Biometrika}, 105:\penalty0 994--1000, 2018{\natexlab{b}}.

\bibitem[Freedman and Lane(1983)]{fl}
D.~Freedman and D.~Lane.
\newblock A nonstochastic interpretation of reported significance levels.
\newblock \emph{Journal of Business and Economic Statistics}, 1:\penalty0
  292--298, 1983.

\bibitem[Freedman(2008)]{Freedman08a}
D.~A. Freedman.
\newblock On regression adjustments to experimental data.
\newblock \emph{Advances in Applied Mathematics}, 40:\penalty0 180--193, 2008.

\bibitem[Fuller(2009)]{fuller}
W.~A. Fuller.
\newblock Some design properties of a rejective sampling procedure.
\newblock \emph{Biometrika}, 96:\penalty0 933--944, 2009.

\bibitem[Gail et~al.(1988)Gail, Tan, and Piantadosi]{gail1988tests}
M.~H. Gail, W.~Y. Tan, and S.~Piantadosi.
\newblock Tests for no treatment effect in randomized clinical trials.
\newblock \emph{Biometrika}, 75:\penalty0 57--64, 1988.

\bibitem[Ganong and J{\"a}ger(2018)]{ganong2018permutation}
P.~Ganong and S.~J{\"a}ger.
\newblock A permutation test for the regression kink design.
\newblock \emph{Journal of the American Statistical Association}, 113:\penalty0
  494--504, 2018.

\bibitem[Guo and Basse(2020)]{guo2020generalized}
K.~Guo and G.~Basse.
\newblock The generalized {Oaxaca--Blinder} estimator.
\newblock \emph{arXiv}, page 2004.11615, 2020.

\bibitem[H{\'a}jek(1961)]{hajek1961some}
J.~H{\'a}jek.
\newblock {Some extensions of the Wald--Wolfowitz--Noether theorem}.
\newblock \emph{Annals of Mathematical Statistics}, 32:\penalty0 506--523,
  1961.

\bibitem[Heckman and Karapakula(2021)]{hk}
J.~J. Heckman and G.~Karapakula.
\newblock Using a satisficing model of experimenter decision-making to guide
  finite-sample inference for compromised experiments.
\newblock \emph{Econometrics Journal}, page in press, 2021.

\bibitem[Heckman et~al.(2020)Heckman, Pinto, and Shaikh]{heckman2020}
J.~J. Heckman, R.~Pinto, and A.~M. Shaikh.
\newblock Inference with imperfect randomization: The case of the {Perry}
  preschool program.
\newblock Working paper, University of Chicago, 2020.

\bibitem[Hennessy et~al.(2016)Hennessy, Dasgupta, Miratrix, Pattanayak, and
  Sarkar]{hennessy2016conditional}
J.~Hennessy, T.~Dasgupta, L.~Miratrix, C.~Pattanayak, and P.~Sarkar.
\newblock A conditional randomization test to account for covariate imbalance
  in randomized experiments.
\newblock \emph{Journal of Causal Inference}, 4:\penalty0 61--80, 2016.

\bibitem[Hoeffding(1952)]{hoeffding1952large}
W.~Hoeffding.
\newblock The large-sample power of tests based on permutations of
  observations.
\newblock \emph{The Annals of Mathematical Statistics}, 23:\penalty0 169--192,
  1952.

\bibitem[Huber(1967)]{huber::1967}
P.~J. Huber.
\newblock The behavior of maximum likelihood estimates under nonstandard
  conditions.
\newblock In Lucien M.~Le Cam and Jerzy Neyman, editors, \emph{Proceedings of
  the Fifth Berkeley Symposium on Mathematical Statistics and Probability},
  volume~1, pages 221--233. Berkeley, California: University of California
  Press, 1967.

\bibitem[Janssen(1997)]{Janssen97}
A.~Janssen.
\newblock Studentized permutation tests for non-iid hypotheses and the
  generalized {B}ehrens--{F}isher problem.
\newblock \emph{Statistics and Probability Letters}, 36:\penalty0 9--21, 1997.

\bibitem[Jiang et~al.(2019)Jiang, Tian, Fu, Hasegawa, and Wei]{jiang2019robust}
F.~Jiang, L.~Tian, H.~Fu, T.~Hasegawa, and L.~J. Wei.
\newblock {Robust alternatives to ANCOVA for estimating the treatment effect
  via a randomized comparative study}.
\newblock \emph{Journal of the American Statistical Association}, 114:\penalty0
  1854--1864, 2019.

\bibitem[Kennedy(1995)]{knd}
P.~E. Kennedy.
\newblock Randomization tests in econometrics.
\newblock \emph{Journal of Business and Economic Statistics}, 13:\penalty0
  85--94, 1995.

\bibitem[Lee and Shaikh(2014)]{lee2014}
S.~Lee and A.~M. Shaikh.
\newblock Multiple testing and heterogeneous treatment effects: Re-evaluating
  the effect of progresa on school enrollment.
\newblock \emph{Journal of Applied Econometrics}, 29:\penalty0 612--626, 2014.

\bibitem[Lehmann(1975)]{lehmann1975nonparametrics}
E.~L. Lehmann.
\newblock \emph{Nonparametrics: Statistical Methods Based on Ranks}.
\newblock San Francisco: Holden-Day, Inc., 1975.

\bibitem[Lehmann and Romano(2005)]{Lehmann2005}
E.~L. Lehmann and J.~P. Romano.
\newblock \emph{Testing Statistical Hypotheses}.
\newblock New York: Springer, 3rd edition, 2005.

\bibitem[Lei and Bickel(2020)]{lei2019assumption}
L.~Lei and P.~J. Bickel.
\newblock An assumption-free exact test for fixed-design linear models with
  exchangeable errors.
\newblock \emph{Biometrika}, page in press, 2020.

\bibitem[Lei and Ding(2020)]{lei2018regression}
L.~Lei and P.~Ding.
\newblock Regression adjustment in completely randomized experiments with a
  diverging number of covariates.
\newblock \emph{Biometrika}, page in press, 2020.

\bibitem[Li and Ding(2017)]{DingCLT}
X.~Li and P.~Ding.
\newblock General forms of finite population central limit theorems with
  applications to causal inference.
\newblock \emph{Journal of the American Statistical Association}, 112:\penalty0
  1759--1169, 2017.

\bibitem[Li and Ding(2020)]{LD20}
X.~Li and P.~Ding.
\newblock Rerandomization and regression adjustment.
\newblock \emph{Journal of the Royal Statistical Society, Series B
  (Methodological)}, 82:\penalty0 241--268, 2020.

\bibitem[Li et~al.(2018)Li, Ding, and Rubin]{LD2018}
X.~Li, P.~Ding, and D.~B. Rubin.
\newblock Asymptotic theory of rerandomization in treatment--control
  experiments.
\newblock \emph{Proceedings of the National Academy of Sciences of the United
  States of America}, 115:\penalty0 9157--9162, 2018.

\bibitem[Lin(2013)]{Lin13}
W.~Lin.
\newblock {Agnostic notes on regression adjustments to experimental data:
  Reexamining Freedman's critique}.
\newblock \emph{Annals of Applied Statistics}, 7:\penalty0 295--318, 2013.

\bibitem[Liu and Yang(2020)]{liu2019regression}
H.~Liu and Y.~Yang.
\newblock Regression-adjusted average treatment effect estimates in stratified
  randomized experiments.
\newblock \emph{Biometrika}, 107:\penalty0 935--948, 2020.

\bibitem[Lu(2016)]{lu2016covariate}
J.~Lu.
\newblock {Covariate adjustment in randomization-based causal inference for
  $2^{K}$ factorial designs}.
\newblock \emph{Statistics and Probability Letters}, 119:\penalty0 11--20,
  2016.

\bibitem[MacKinnon and Webb(2020)]{mackinnon2020randomization}
J.~G. MacKinnon and M.~D. Webb.
\newblock Randomization inference for difference-in-differences with few
  treated clusters.
\newblock \emph{Journal of Econometrics}, 218:\penalty0 435--450, 2020.

\bibitem[Manly(1997)]{manly}
B.~F.~J. Manly.
\newblock \emph{Randomization, Bootstrap and Monte Carlo Methods in Biology}.
\newblock Chapman \& Hall, 1997.

\bibitem[Middleton(2018)]{middleton2018unified}
J.~A. Middleton.
\newblock A unified theory of regression adjustment for design-based inference.
\newblock \emph{arXiv preprint arXiv:1803.06011}, 2018.

\bibitem[Middleton and Aronow(2015)]{MiddletonCl15}
J.~A. Middleton and P.~M. Aronow.
\newblock Unbiased estimation of the average treatment effect in
  cluster-randomized experiments.
\newblock \emph{Statistics, Politics and Policy}, 6:\penalty0 39--75, 2015.

\bibitem[Moore and van~der Laan(2009)]{moore2009covariate}
K.~L. Moore and M.~J. van~der Laan.
\newblock Covariate adjustment in randomized trials with binary outcomes:
  targeted maximum likelihood estimation.
\newblock \emph{Statistics in Medicine}, 28:\penalty0 39--64, 2009.

\bibitem[Moore et~al.(2011)Moore, Neugebauer, Valappil, and van~der
  Laan]{moore2011robust}
K.~L. Moore, R.~Neugebauer, T.~Valappil, and M.~J. van~der Laan.
\newblock Robust extraction of covariate information to improve estimation
  efficiency in randomized trials.
\newblock \emph{Statistics in Medicine}, 30:\penalty0 2389--2408, 2011.

\bibitem[Morgan and Rubin(2012)]{morgan2012rerandomization}
K.~L. Morgan and D.~B. Rubin.
\newblock Rerandomization to improve covariate balance in experiments.
\newblock \emph{Annals of Statistics}, 40:\penalty0 1263--1282, 2012.

\bibitem[Mukerjee et~al.(2018)Mukerjee, Dasgupta, and Rubin]{mukerjee2018using}
R.~Mukerjee, T.~Dasgupta, and D.~B. Rubin.
\newblock Using standard tools from finite population sampling to improve
  causal inference for complex experiments.
\newblock \emph{Journal of the American Statistical Association}, 113:\penalty0
  868--881, 2018.

\bibitem[Negi and Wooldridge(2021)]{negi2020revisiting}
A.~Negi and J.~M. Wooldridge.
\newblock Revisiting regression adjustment in experiments with heterogeneous
  treatment effects.
\newblock \emph{Econometric Reviews}, 40:\penalty0 504--534, 2021.

\bibitem[Neyman(1923/1990)]{Neyman23}
J.~Neyman.
\newblock On the application of probability theory to agricultural experiments
  (with discussion).
\newblock \emph{Statistical Science}, 5:\penalty0 465--472, 1923/1990.

\bibitem[Neyman(1935)]{neyman1935statistical}
J.~Neyman.
\newblock Statistical problems in agricultural experimentation (with
  discussion).
\newblock \emph{Supplement to the Journal of the Royal Statistical Society},
  2:\penalty0 107--180, 1935.

\bibitem[Ottoboni et~al.(2018)Ottoboni, Lewis, and Salmaso]{Otto}
K.~Ottoboni, F.~Lewis, and L.~Salmaso.
\newblock An empirical comparison of parametric and permutation tests for
  regression analysis of randomized experiments.
\newblock \emph{Statistics in Biopharmaceutical Research}, 10:\penalty0
  264--273, 2018.

\bibitem[Pauly et~al.(2015)Pauly, Brunner, and Konietschke]{Pauly15}
M.~Pauly, E.~Brunner, and F.~Konietschke.
\newblock Asymptotic permutation tests in general factorial designs.
\newblock \emph{Journal of the Royal Statistical Society, Series B (Statistical
  Methodology)}, 77:\penalty0 461--473, 2015.

\bibitem[Proschan and Dodd(2019)]{proschan2019re}
M.~A. Proschan and L.~E. Dodd.
\newblock Re-randomization tests in clinical trials.
\newblock \emph{Statistics in Medicine}, 38:\penalty0 2292--2302, 2019.

\bibitem[Raz(1990)]{raz1990testing}
J.~Raz.
\newblock Testing for no effect when estimating a smooth function by
  nonparametric regression: a randomization approach.
\newblock \emph{Journal of the American Statistical Association}, 85:\penalty0
  132--138, 1990.

\bibitem[Romano(1990)]{romano1990behavior}
J.~P. Romano.
\newblock On the behavior of randomization tests without a group invariance
  assumption.
\newblock \emph{Journal of the American Statistical Association}, 85:\penalty0
  686--692, 1990.

\bibitem[Rosenbaum(1999)]{rosenbaum1999reduced}
P.~R. Rosenbaum.
\newblock Reduced sensitivity to hidden bias at upper quantiles in
  observational studies with dilated treatment effects.
\newblock \emph{Biometrics}, 55:\penalty0 560--564, 1999.

\bibitem[Rosenbaum(2002)]{CovAdjRosen02}
P.~R. Rosenbaum.
\newblock Covariance adjustment in randomized experiments and observational
  studies.
\newblock \emph{Statistical Science}, 17:\penalty0 286--327, 2002.

\bibitem[Rosenbaum(2003)]{Rosenbaum:2003en}
P.~R. Rosenbaum.
\newblock Exact confidence intervals for nonconstant effects by inverting the
  signed rank test.
\newblock \emph{American Statistician}, 57:\penalty0 132--138, 2003.

\bibitem[Rosenbaum(2010)]{rosenbaum2020}
P.~R. Rosenbaum.
\newblock \emph{Design of Observational Studies}.
\newblock New York: Springer, 2nd edition, 2010.

\bibitem[Stephens et~al.(2013)Stephens, Tchetgen~Tchetgen, and
  De~Gruttola]{stephens2013flexible}
A.~J. Stephens, E.~J. Tchetgen~Tchetgen, and V.~De~Gruttola.
\newblock {Flexible covariate-adjusted exact tests of randomized treatment
  effects with application to a trial of HIV education}.
\newblock \emph{Annals of Applied Statistics}, 7:\penalty0 2106--2137, 2013.

\bibitem[Su and Ding(2021)]{DS}
F.~Su and P.~Ding.
\newblock Model-assisted analyses of cluster-randomized experiments.
\newblock \emph{arXiv}, page 2104.04647, 2021.

\bibitem[ter Braak(1992)]{tb}
C.~J.~F. ter Braak.
\newblock {Permutation versus bootstrap significance tests in multiple
  regression and ANOVA}.
\newblock In K.H. J{\"o}ckel, G.~Rothe, and W.~Sendler, editors,
  \emph{Bootstrapping and Related Techniques}, pages 79--85. Berlin:
  Springer-Verlag, 1992.

\bibitem[Tsiatis et~al.(2008)Tsiatis, Davidian, Zhang, and
  Lu]{tsiatis2008covariate}
A.~A. Tsiatis, M.~Davidian, M.~Zhang, and X.~Lu.
\newblock Covariate adjustment for two-sample treatment comparisons in
  randomized clinical trials: a principled yet flexible approach.
\newblock \emph{Statistics in Medicine}, 27:\penalty0 4658--4677, 2008.

\bibitem[Tukey(1993)]{tukey1993tightening}
J.~W. Tukey.
\newblock Tightening the clinical trial.
\newblock \emph{Controlled Clinical Trials}, 14:\penalty0 266--285, 1993.

\bibitem[van~der vaart and Wellner(1996)]{vdv-yellowbook}
A.~W. van~der vaart and J.~Wellner.
\newblock \emph{{Weak Convergence and Empirical Processes: With Applications to
  Statistics}}.
\newblock New York: Springer Verlag, 1996.

\bibitem[Wager et~al.(2016)Wager, Du, Taylor, and Tibshirani]{wager2016high}
S.~Wager, W.~Du, J.~Taylor, and R.~J. Tibshirani.
\newblock High-dimensional regression adjustments in randomized experiments.
\newblock \emph{Proceedings of the National Academy of Sciences of the United
  States of America}, 113:\penalty0 12673--12678, 2016.

\bibitem[White(1980)]{White80}
H.~White.
\newblock A heteroskedasticity-consistent covariance matrix estimator and a
  direct test for heteroskedasticity.
\newblock \emph{Econometrica}, 48:\penalty0 817--838, 1980.

\bibitem[Wu and Gagnon-Bartsch(2018)]{wu2018loop}
E.~Wu and J.~A. Gagnon-Bartsch.
\newblock {The LOOP estimator: Adjusting for covariates in randomized
  experiments}.
\newblock \emph{Evaluation Review}, 42:\penalty0 458--488, 2018.

\bibitem[Wu and Ding(2020)]{wuanding2020jasa}
J.~Wu and P.~Ding.
\newblock Randomization tests for weak null hypotheses in randomized
  experiments.
\newblock \emph{Journal of American Statistical Association}, 105:\penalty0 in
  press, 2020.

\bibitem[Ye et~al.(2020{\natexlab{a}})Ye, Shao, and Zhao]{ye2020principles}
T.~Ye, J.~Shao, and Q.~Zhao.
\newblock Principles for covariate adjustment in analyzing randomized clinical
  trials.
\newblock \emph{arXiv preprint arXiv:2009.11828}, 2020{\natexlab{a}}.

\bibitem[Ye et~al.(2020{\natexlab{b}})Ye, Yi, and Zhao]{ye2020inference}
T.~Ye, Y.~Yi, and Q.~Zhao.
\newblock Inference on average treatment effect under minimization and other
  covariate-adaptive randomization methods.
\newblock \emph{arXiv preprint arXiv:2007.09576}, 2020{\natexlab{b}}.

\bibitem[Young(2019)]{young2019channeling}
A.~Young.
\newblock Channeling {F}isher: Randomization tests and the statistical
  insignificance of seemingly significant experimental results.
\newblock \emph{Quarterly Journal of Economics}, 134:\penalty0 557--598, 2019.

\bibitem[Zhang et~al.(2008)Zhang, Tsiatis, and Davidian]{zhang2008improving}
M.~Zhang, A.~A. Tsiatis, and M.~Davidian.
\newblock Improving efficiency of inferences in randomized clinical trials
  using auxiliary covariates.
\newblock \emph{Biometrics}, 64:\penalty0 707--715, 2008.

\bibitem[Zhang and Zheng(2020)]{zhang2020}
Y.~C. Zhang and X.~Zheng.
\newblock Quantile treatment effects and bootstrap inference under
  covariate-adaptive randomization.
\newblock \emph{Quantitative Economics}, 11:\penalty0 957--982, 2020.

\bibitem[Zheng and Zelen(2008)]{zheng2008multi}
L.~Zheng and M.~Zelen.
\newblock Multi-center clinical trials: Randomization and ancillary statistics.
\newblock \emph{Annals of Applied Statistics}, 2:\penalty0 582--600, 2008.

\end{thebibliography}
